\documentclass[a4paper,3p,times]{elsarticle}

\usepackage{amsfonts,amsmath,amsthm}
\usepackage{algpseudocode}
\usepackage{algorithm}
\usepackage{tabularx}
\usepackage{subcaption}
\usepackage{pifont}
\usepackage[svgnames]{xcolor}
\usepackage{hyperref}
\usepackage{textcomp}
\usepackage{multirow}
\newcolumntype{R}{>{\raggedleft\arraybackslash}X}

\newcommand{\smallHead}[1]{
	\vspace{0.5mm}\noindent\textbf{#1}}

\newcommand{\smallHeadIndent}[1]{
	\vspace{0.5mm}\textbf{#1}}

\newtheorem{lemma}{Lemma}

\newcommand{\ODL}{\mathrm{ODL}}
\newcommand{\SDL}{\mathrm{SDL}}
\newcommand{\policy}[1]{\textsf{#1}}
\newcommand{\LP}{\mathrm{LP}}
\newcommand{\RP}{\mathrm{RP}}

\algtext*{EndProcedure}
\algtext*{EndFunction}
\algtext*{EndIf}
\algtext*{EndFor}
\algtext*{EndWhile}

\algrenewcommand\algorithmicindent{1em}%

\let\oldbibliography\thebibliography
\renewcommand{\thebibliography}[1]{%
  \oldbibliography{#1}%
  \setlength{\itemsep}{0pt}
}

\journal{Artificial Intelligence}

\begin{document}

\begin{frontmatter}


\title{COL-Trees: Efficient Hierarchical Object Search in Spatial Networks}

\author{Tenindra Abeywickrama\texorpdfstring{\corref{cor1}}{}}
\ead{tenindra.abeywickrama@riken.jp}
\cortext[cor1]{Corresponding author}
\affiliation{organization={Center for Computational Science, RIKEN},
            country={Japan}}

\author{Muhammad Aamir Cheema}
\ead{aamir.cheema@monash.edu}
\affiliation{organization={Faculty of Information Technology, Monash University},
            country={Australia}}

\author{Sabine Storandt}
\ead{sabine.storandt@uni-konstanz.de}
\affiliation{organization={Department of Computer and Information Science, University of Konstanz},
            country={Germany}}



\begin{abstract}
Many spatial network applications, such as map-based services, rely heavily on efficient methods to search for relevant nearby points-of-interest (POIs). A $k$ Nearest Neighbor ($k$NN) query is one such example that finds the $k$ closest POIs from an agent's location. While most existing techniques for spatial networks focus on retrieving nearby POIs for a single agent, these search heuristics do not translate to many other applications. For example, Aggregate $k$ Nearest Neighbor (A$k$NN) queries require POIs that are close to \emph{multiple} agents. $k$ Farthest Neighbors ($k$FN) queries require POIs that are the antithesis of nearest. Such problems naturally benefit from a hierarchical approach, but existing methods rely on Euclidean-based heuristics, which have diminished effectiveness in spatial networks such as road network graphs. We propose a novel data structure, COL-Tree (Compacted Object-Landmark Tree), to address this gap by enabling efficient hierarchical graph traversal using a more accurate landmark-based heuristic. We then present query algorithms that utilize COL-Trees to efficiently answer A$k$NN, $k$FN, and other queries. Our experiments evaluate performance on real-world and synthetic datasets including multiple types of spatial networks, such as road networks of varying sizes and edge weights, terrain networks, and 2D physics meshes. The results demonstrate that our techniques significantly outperform existing approaches, achieving up to three orders of magnitude improvement. Moreover, this comes at a small pre-processing overhead in both theory and practice.
\end{abstract}




\end{frontmatter}


\section{Introduction}\label{intro}

Finding nearby relevant objects, and in particular points-of-interest (POIs), efficiently is a critical task in a variety of heuristic search and planning applications including, e.g., real-time path planning in video game maps \cite{bulitko2010case} or efficiently retrieving POIs in location-based services. Nearby relevant POIs are typically obtained by using a range query or a $k$ Nearest Neighbor ($k$NN) query. A range query~\cite{papadias2003ine} returns all POIs within a given distance from an agent's location, e.g., find all restaurants within $1$km from a user. A $k$NN query~\cite{zhong2015gtree} returns the $k$ closest POIs from an agent's location, e.g., find the three nearest gas stations to a taxi driver. The distance metric may be Euclidean distance (i.e., ``as the crow flies'') or network distance (i.e., the length of the shortest path). In this paper, we consider network distance as it is a more accurate and versatile measure of proximity in many real-world applications such as map-based services and is capable of representing a variety of metrics such as travel time, physical distance, toll cost, emissions, etc.

In the context of graphs representing spatial networks, computing shortest path distances (or network distances) is significantly more expensive than Euclidean distance. As a result, efficient heuristics are needed to retrieve candidate POIs to answer the aforementioned queries and these typically rely on two implicit assumptions to reduce the overall cost: (a) the search occurs from a single agent; and (b) the final result POIs are in the general vicinity of the single agent's location. However, there are many important queries where one or both of these assumptions do not hold. For example, an Aggregate $k$ Nearest Neighbor (A$k$NN) query~\cite{yiu2006aknn} is a natural extension of a $k$NN query for \textit{multiple agents}, where POIs are retrieved considering their aggregate distance from all agents. Given a POI $p$ and a set of agents $Q$, the aggregate distance of $p$ from the agents is $d_{agg} (Q,p)= agg(d(q_i,p),\forall {q_i\in Q})$ where $d(q_i,p)$ denotes the network distance from an agent $q_i$ to $p$ and $agg()$ is an aggregate function. For example, when the aggregate function is $sum$, $d_{agg}(Q,p)= \sum_{q_i\in Q} d(q_i,p)$ and when the aggregate function is $max$, $d_{agg}(Q,p)= \max_{q_i\in Q} d(q_i,p)$. An A$k$NN query returns the $k$ POIs with smallest aggregate distances. Another pertinent example is the $k$ Farthest Neighbors ($k$FN) problem~\cite{wang2016rkfn}. The opposite of the $k$NN problem, $k$FN queries involve retrieving the $k$ farthest POIs from an agent's location. In both A$k$NN and $k$FN queries, due to their definition, the result objects are often not among the closest to the agent locations and heuristics that retrieve candidates incrementally from the agents' locations (as in most existing work) are unlikely to be efficient.

Answering queries that consider multiple agents is highly relevant for many real-world applications. A$k$NN queries are widely used in transportation, e.g., ride-sharing \cite{stiglic2015benefits,drews2013multi}. For example, a group of friends planning to meet may want to choose a restaurant such that the total distance they need to travel is minimized. They can issue an A$k$NN query to find $k$ restaurants with the smallest aggregate distances where the aggregate function is $sum$. Similarly, a group of hikers may issue an A$k$NN query on a terrain graph to find the campsite that minimizes the sum of their travel distances. Consider an urban planning problem that requires building a fire station at one of many candidate sites with the goal of minimizing the maximum distance of the station from residential areas in the city. In this case, an A$k$NN ($k=1$) query can be used to choose a candidate site $p$ such that its aggregate distance (with function $max$) from the set of residential areas is minimum among all candidate sites. A$k$NN search also has another important real-world application in \textit{shortest detour queries}, which are provided in many popular navigation apps. Given a source $s$ and a target $t$, a shortest detour query returns a POI $p$ such that $d(s,p) + d(p,t)$ is minimized. Consider a user who wants to stop at a gas station on their route from home to work. In this case, a shortest detour query helps find the gas station that minimizes the total distance the user needs to travel. Note that this is a special case of an A$k$NN query where $s$ and $t$ are considered to be the agent locations and the aggregate function is $sum$.

$k$FN queries also find several important real-world applications. They can be used to find spatial outliers, for example, to allow a healthcare facility to identify elderly residents that are farthest by travel time for potential remote monitoring, or to allow a delivery company to find the farthest customers from its warehouse to understand the worst-case delivery times. $k$FN queries may also be used in \textit{obnoxious facility location}~\cite{church2022obnoxiousreview} where the goal is to place facilities with less desirable characteristics as far away as possible from a sensitive location (or group of locations~\cite{gao2011afn}). For example, among a set of candidate sites, we may wish to place a potentially dangerous facility as far away by travel time from the center of a city as possible, such as a biological research facility or quarantine station to discourage civilians visiting or a maximum security prison to give authorities additional time to respond before civilians are exposed to any risk from escapees. The POI need not be obnoxious per se, e.g., placing a backup site or vault as far away from an existing site as possible to maximize security. Each of these could be answered by a $k$FN query and would benefit from increased accuracy based on network distance metrics. Note that the objective is different from A$k$NN queries with $max$ aggregate function, which is a minmax problem. Farthest neighbor computation is also frequently used as a subroutine in solutions to many other problems, such as hierarchical clustering~\cite{dasgupta2005perfcluster}, feed-link queries~\cite{bose2013farthestptdiag}, traveling salesman or vehicle routing problems~\cite{rosenkrantz1977farthesttsp}, and several metric $k$-center algorithms \cite{gonzalez1985cluster,dyer1985pcenter,diaz2017kcenter3approx}, including Farthest Point Sampling on mesh networks, where shortest paths approximate geodesic distance~\cite{tenenbaum2000geodesic,peyre2006geodesic}. 

Clearly, both A$k$NN and $k$FN queries have important and commonly used real-world applications where graph-based metrics would help improve solution accuracy. This necessitates efficient methods and, by extension, heuristics to solve them. While road networks are an obvious application and form the main focus of our evaluations, our intuitions generalize to spatial networks more broadly and we include studies on other spatial networks such as meshes and terrain graphs (Section~\ref{sec:exp:other}). It is worth noting that these queries have applications in a wide variety of domains and can be used for other types of graphs such as social networks and wireless sensor networks. In the rest of the paper, we use the term query location to refer to the location of an agent whenever clear from context. For ease of exposition, we assume the input graph to be undirected (i.e., $d(x,y) = d(y,x)$ for any two points $x$ and $y$). The proposed techniques can be easily extended to directed graphs (e.g., road networks with one-way roads).

\subsection{Limitations of Existing Techniques}

The key to efficiently retrieving POIs in spatial networks is in developing heuristics to retrieve candidate POIs that are most promising. Heuristics to answer A$k$NN queries have largely been borrowed from $k$NN search~\cite{yiu2006aknn,zhu2010vn3aknn,yao2018fann}, which is problematic in several ways. First, intuitions to find $k$NNs do not translate to A$k$NNs. For example, the most efficient $k$NN heuristic \cite{abey2017knn} uses a recurrence rule stating that the $k$-th nearest POI must be adjacent to at least one of the $k-1$ nearest POIs. The rule is exploited by storing the nearest POI of every potential query location and adjacency relationships between regions associated with POIs. However, A$k$NN queries involve retrieving POIs by an aggregate distance from multiple query locations. Thus, the recurrence rule is no longer true and cannot be applied, rendering the heuristic unsuitable. This approach is even less desirable for $k$FN queries as it would involve evaluating all POIs to determine the farthest.

A better approach to find A$k$NNs and $k$FNs is to conduct a hierarchical search on the network. One A$k$NN method~\cite{yiu2006aknn} searches an R-tree~\cite{guttman1984rtree} containing the POIs in a branch-and-bound manner according to a Euclidean distance heuristic, similar to another $k$NN technique~\cite{papadias2003ine}. R-trees are constructed in a way that recursively divides POIs into subsets by Minimum Bounding Rectangles (MBRs). During search, a lower-bound aggregate distance for all POIs in a child R-tree node is computed using the Euclidean distances to its MBR. Now the most promising tree branches can be visited to pinpoint result POIs. However, this is not ideal in important application settings like road networks, as Euclidean distance is only a loose lower-bound on network distance. This results in a less efficient heuristic, especially when the edge weights represent travel time. Moreover, the inefficiency is exacerbated for A$k$NNs as the error will also be aggregated, e.g., making it harder to distinguish between candidates. Furthermore, this method is not admissible for $k$FNs, and cannot be used to answer such queries at all. Although a maximum Euclidean distance could be computed to MBRs, Euclidean distance is still a lower-bound. $k$FN being a maximization problem, candidates must be retrieved by upper-bound to terminate optimally.

Landmark Lower-Bounds (LLBs) are a more accurate alternative to Euclidean distance~\cite{goldberg2005alt}. 
They involve pre-computing and storing certain distances to \textit{landmark} vertices in the graph, which can then be used to compute a lower-bound between any two locations using the triangle inequality. While LLBs have been employed to answer $k$NN queries~\cite{abey2017knn}, this entailed significant pre-processing time and space. Specifically, the number of landmarks that can be chosen is limited due to high indexing cost, which leads to poor lower-bounds when applied to $k$NN search. Furthermore, there is no efficient data structure to compute minimum LLBs to groups of POIs to perform hierarchical search similar to R-trees and Euclidean distance.

\subsection{Contributions}

Motivated by the aforementioned under-served real-world POI search queries and by the gap in the existing work, our primary contribution is the \textbf{COL-Tree index}, an efficient data structure for hierarchical search in spatial networks to locate POIs by landmark lower-bounds. COL-Tree enables utilization of significantly more landmarks than existing work, while keeping the pre-processing cost small in theory and practice. Our further contributions are as follows:

\begin{itemize}
    \itemsep0em
    \item SUL-Trees: a supporting index that significantly accelerates COL-Tree construction with pre-processing cost theoretically and experimentally lower than or comparable to popular shortest path indexes. 
    \item Highly efficient heuristic search algorithms that utilize COL-Tree to answer A$k$NN, $k$FN, and range queries. We show COL-Trees are particularly suited for A$k$NN search due to a COL-Tree property that allows candidate retrieval to be made more efficient when the aggregate function is convexity-preserving, as is the case for the most commonly used $sum$ and $max$ functions, and for $k$FN search by supporting efficient upper-bound-based incremental candidate retrieval in spatial networks.
    \item Our extensive experiments demonstrate the significant improvement in query time and heuristic efficiency of our techniques, achieving up to three orders of magnitude of improvement.
\end{itemize}

This manuscript is an extended version of a conference paper~\cite{abey2020aknn} selected as the best paper at the 30th International Conference on Automated Planning and Scheduling (ICAPS 2020) and invited for submission as an extended journal paper. Compared to the conference version, we have: (a) proposed highly efficient query algorithms for two new queries ($k$FN and range queries); (b) reduced pre-processing time of SUL-Trees with a novel variant of Dijkstra search for subgraphs; (c) improved SUL-Tree index size and implementation efficiency by realizing a subgraph reordering technique that was proposed theoretically in the previous version; (d) improved query performance with a root-based lower-bound optimization and exhaustive parameter testing (with new experiments and algorithmic insights on parameter choice); and (e) verified that the gains of our techniques carry over to spatial networks more broadly with new evaluations on distance edge-weights and on non-road spatial networks (a 2D mesh and a terrain graph). Note that we have renamed the data structures to COL-Tree (formerly COLT) and SUL-Tree (formerly SL-Tree) from the previous version for clarity and easier pronunciation. Apart from improved pre-processing, the data structures are unchanged.

\section{Preliminaries}\label{prelim}

\smallHead{Spatial Network:}  
We define a spatial network (e.g., a road network) as a graph $G = (V,E)$. $V$ is a vertex set and $E$ is an edge set. Each edge $(u,v) \in E$ connects two vertices with weight $w(u,v)$ representing any real positive metric, e.g., the length or travel time of the edge. Network distance $d(s,t)$ is the minimum sum of weights connecting vertices $s$ and $t$. We consider queries and objects (POIs) on graph vertices for simpler exposition. Table~\ref{tab:symbols} summarizes frequently used notation.

\begin{table}[htb]
\centering
\begin{tabular}{|c||l|} \hline
\textbf{Symbol} & \textbf{Description} \\ \hline\hline
$G$       & A spatial network graph \\ \hline
$V$       & The vertex set of $G$, e.g., representing intersections \\ \hline
$E$       & The edge set of $G$, e.g., representing road segments \\ \hline
$P$       & The POI (object) vertex set ($P \subseteq V$) \\ \hline
$k$       & The number of result objects a query must return ($k \leq |P|$) \\ \hline
$d(s,t)$  & Shortest path (network) distance between $s$ and $t$ in $G$ \\ \hline
$LB(s,t)$ & A lower-bound on network distance, i.e., $LB(s,t) \leq d(s,t)$ \\ \hline
$UB(s,t)$ & An upper-bound on network distance, i.e., $UB(s,t) \geq d(s,t)$ \\ \hline
$agg()$   & A monotonic function that aggregates multiple input values, e.g., $sum$ or $max$ \\ \hline
$\ODL$     & List of object-landmark distances $(o,d(l,o))\; \forall o \in n_{\mathrm{leaf}} \subseteq P$ for COL-Tree leaf node $n_{\mathrm{leaf}}$ \\ \hline
$\SDL$     & List of subgraph vertex-landmark distances $(v,d(l,v))\; \forall v \in n_T$ for SUL-Tree node $n_T$ \\ \hline
\end{tabular}
\caption{Summary of frequently used notations and symbols}
\label{tab:symbols}
\end{table}

\smallHead{Object Search Queries:}  
Finding relevant objects according to some criteria on their network distance is an extremely common task in spatial applications such as map-based services. As motivated in Section~\ref{intro}, we focus on solving the following important object search queries for a set of object vertices $P \subseteq V$:

\begin{enumerate}
    \itemsep0em
    \item \textbf{Aggregate $k$ Nearest Neighbor (A$k$NN) Queries:}  
    Given a set of query vertices $Q \subseteq V$ and an aggregation function $agg$, an A$k$NN query retrieves the $k$ objects in $P$ with minimum aggregate distances from set $Q$. Aggregate distance $d_{agg}(Q,p)$ to object $p$ is computed by aggregating the network distances to $p$ from each query vertex $q \in Q$ by function $agg$. As is common practice in the literature, this work evaluates techniques on $sum$ or $max$ as the aggregate function. We remark that A$k$NN queries are also known as Group $k$NN queries \cite{papadias2004gnn}.
    \item \textbf{$k$ Farthest Neighbors ($k$FN) Queries:}  
    Given a query vertex $q \in V$, a $k$FN query retrieves the $k$ objects in $P$ with maximum network distances from $q$.
    \item \textbf{Range Queries:}  
    Given a query vertex $q \in V$ and a distance threshold $r$, a range query retrieves \textit{all} objects $p \in P$ that satisfy $d(q,p) \leq r$ (i.e., within the range).
\end{enumerate}

Note that $k$NN queries are a special case of A$k$NN queries where $Q$ contains only a single query location. For completeness we present results for $k$NN queries in~\ref{app:knn}.

{
	\begin{figure}[htb]
		\centering
		\begin{minipage}[c]{0.4\linewidth}
			\centering
			\includegraphics[width=\linewidth]{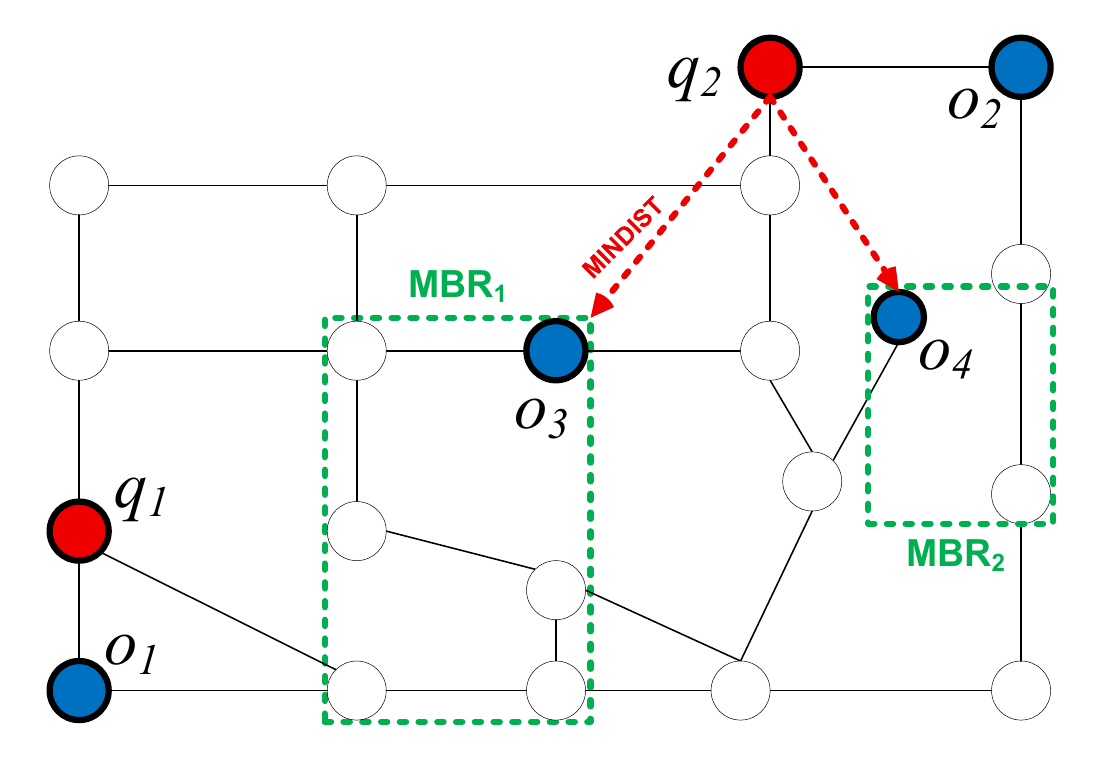}
		\end{minipage}%
		\hspace{5mm}%
		\begin{minipage}[c]{0.25\linewidth}
			\centering
			\small
			\begin{tabular}{|c||c|c|} \hline
				      & $q_1$ & $q_2$ \\ \hline\hline
				$o_1$ & 1     & 6     \\ \hline
				$o_2$ & 6     & 1     \\ \hline
				$o_3$ & 3     & 3     \\ \hline
				$o_4$ & 5     & 4     \\ \hline
			\end{tabular}
		\end{minipage}
		\caption{Running example graph with sample query and object vertices, and their network distances}
		\label{diag:query_example}
	\end{figure}
}

\smallHead{Query Example:} Consider the graph in Figure~\ref{diag:query_example} with query vertices $q_1$ and $q_2$ and objects $o_1$ to $o_4$ and the relevant network distances in the table on the right (assuming unit edge-weights for simplicity and the MBRs are indicative). Let us then consider an A$k$NN query with $Q{=}\{q_1,q_2\}$ and $k{=}1$. Both $sum$ and $max$ will return $o_3$ as the result, as $d_{sum}(Q,o_3){=}6$ and $d_{max}(Q,o_3){=}3$, versus at least $7$ and $5$ for the other objects, respectively. If an expansion heuristic were used, a search would begin from each query vertex and only generate a candidate once they met. This requires settling objects that are not results, such as $o_1$ and $o_2$ (the $1$NNs of $q_1$ and $q_2$), before the expansions meet at $o_3$. A Euclidean heuristic faces a different problem, with its as-the-crow-flies lower-bounds being less effective in graphs. For example, the MBR of $o_4$ is closer to $q_2$ than that of $o_3$, even though $o_4$ is farther from $q_2$ by network distance ($4$ versus $3$). Thus, using Euclidean distance to generate candidates would result in a misordered retrieval and false positive network distance computations. It is easy to see how such errors can compound in A$k$NN search as more candidates are retrieved.

Figure~\ref{diag:query_example} can also be used to illustrate the other queries. The $1$FN of $q_1$ is $o_2$ at a network distance of $6$; however, this would be the last object retrieved if an expansion heuristic were used. A range query from $q_2$ with $r{=}3$ returns $\{o_2,o_3\}$ as the final results. Since $o_4$ is closer by Euclidean distance to $q_2$ than $o_3$, it would still be retrieved under a Euclidean heuristic, resulting in a false positive network distance computation.

\smallHead{Landmark Lower-Bounds (LLBs):} 
Landmark lower-bounds, also known as differential heuristics \cite{goldenberg2011cdh}, involve selecting a set $L$ of $m$ ``landmark'' vertices and then pre-computing distances from each landmark to all vertices in $V$. Given two vertices $q$ and $p$, a lower-bound on network distance can be computed by Eq.~(\ref{eq:llb}) using the distances to landmark $l_i$ and the triangle inequality. The maximum lower-bound over all $m$ landmarks given by Eq.~(\ref{eq:llbmax}) gives the tightest lower-bound and is typically accurate even for small $m$ \cite{goldberg2005alt}.

\begin{equation}\label{eq:llb}
\left. LB_{l_i}(q,p) = |d(l_i,q)-d(l_i,p)| \right. \leq d(q,p)
\end{equation}

\begin{equation}\label{eq:llbmax}
\left. LB_{max}(q,p) \right. = \max\limits_{l_i \in L}(|d(l_i,q)-d(l_i,p)|)
\end{equation}

\smallHead{Landmark Upper-Bounds (LUBs):}  
Using the same pre-computed distances, an upper-bound can also be computed using the triangle inequality as in Eq.~(\ref{eq:lub}). In this case, the minimum upper-bound is chosen over all $m$ landmarks to obtain the tightest upper-bound. 

\begin{equation}\label{eq:lub}
\left. UB_{min}(q,p) \right. = \min\limits_{l_i \in L}(d(l_i,q)+d(l_i,p))
\end{equation}

\smallHead{Lower-Bound Aggregate Distances:} 
Lemma \ref{lemma:agg} \cite{yiu2006aknn} shows that, for monotonic aggregate functions, the aggregate of lower-bound distances to object $p$ from each query vertex in $Q$ is a lower-bound $LB_{agg}(Q,p)$ on aggregate distance $d_{agg}(Q,p)$. 

\begin{lemma}\label{lemma:agg}
	Given a monotonic aggregate function $agg$ and lower-bound distances $LB(q_i,p)$ from each query vertex $q_i \in Q$ to object $p$, the aggregation of the lower-bound distances gives a lower-bound aggregate distance on the true aggregate distance. I.e., $LB_{agg}(Q,p) = agg(LB(q_1,p),\ldots,LB(q_{|Q|},p)) \leq d_{agg}(Q,p)$.
\end{lemma}

\section{Data Structures}\label{index}

Before introducing our main data structure, the Compacted Object-Landmark Tree (COL-Tree), to enable efficient hierarchical object search in spatial networks, we first describe a supporting index called a Subgraph-Landmark Tree (SUL-Tree). SUL-Trees are not strictly necessary to build COL-Trees, but significantly accelerate their construction and also help introduce the hierarchical subgraph structure.

\begin{figure}[hbt]
    \centering
    \begin{subfigure}[b]{0.4\linewidth}
        \includegraphics[width=\linewidth]{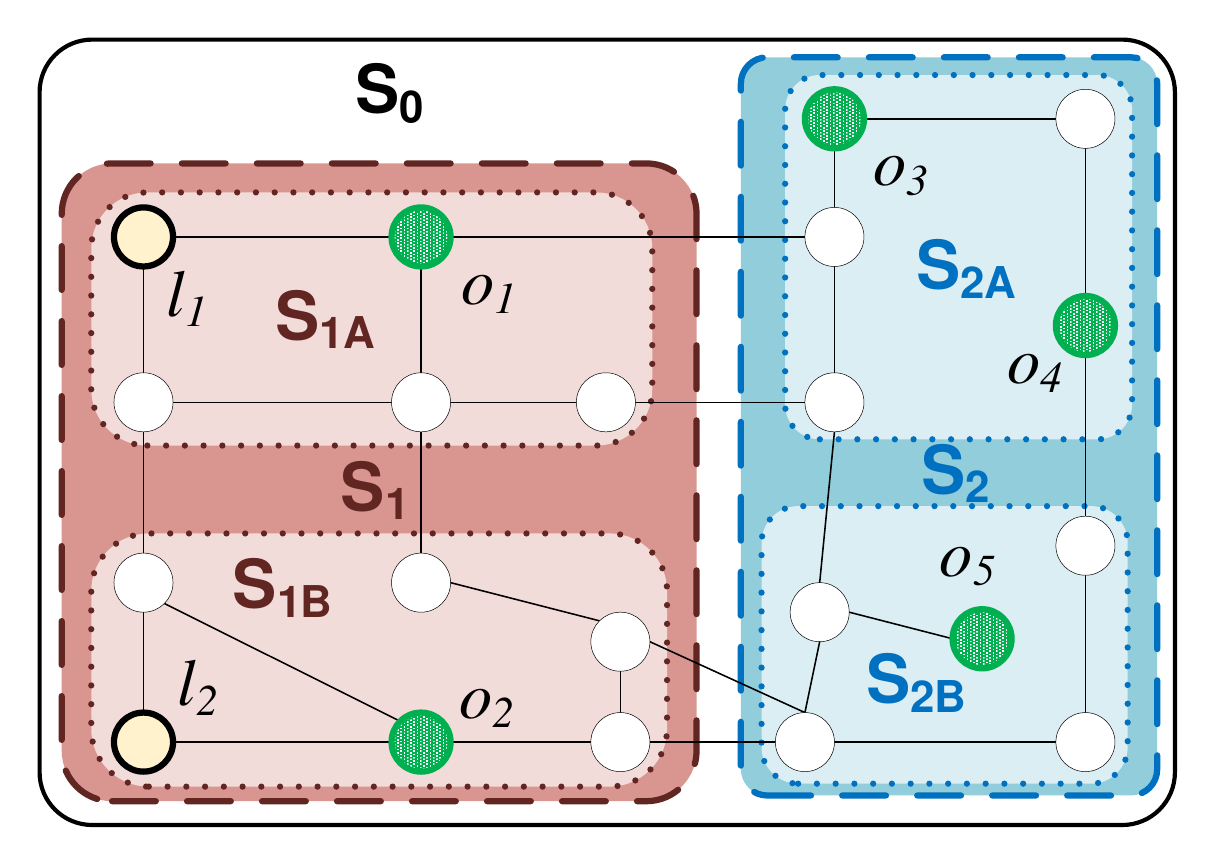}
        \caption{SUL-Tree Subgraph Partitions}
        \label{diag:vstree_graph}
    \end{subfigure}
    \hspace*{-3mm}
    \begin{subfigure}[b]{0.35\linewidth}
        \includegraphics[width=\linewidth]{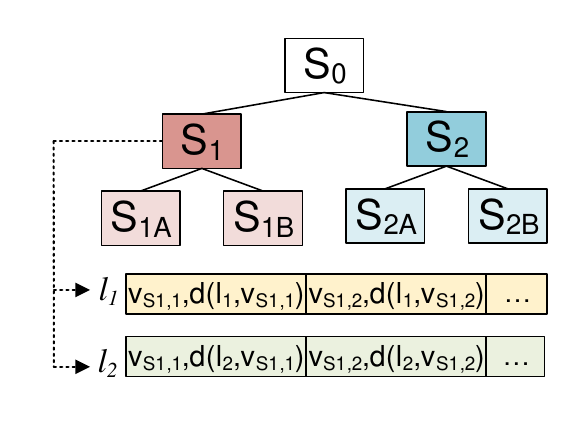}
        \caption{SUL-Tree Node Hierarchy}
        \label{diag:vstree_hierarchy}
    \end{subfigure}
	\caption{Subgraph-Landmark Tree (SUL-Tree)}
	\label{diag:vstree}
\end{figure}

\subsection{Graph Index: Subgraph-Landmark Tree (SUL-Tree)}\label{sec:sultreeconstruction}
A SUL-Tree is an index for a given graph $G$. Each node in the SUL-Tree represents a subgraph of the input graph with $G$ being the tree's root. $G$ is recursively partitioned into $b$ disjoint subgraphs of equal size, stopping when a subgraph has no more than $\alpha$ vertices. Figure \ref{diag:vstree_graph} shows such a partitioning for $b=2$ and $\alpha=6$ with the corresponding SUL-Tree shown in Figure \ref{diag:vstree_hierarchy}. For example, root $S_0$ is partitioned into subgraphs $S_1$ and $S_2$ in Figure~\ref{diag:vstree_graph}, which are represented as the children of $S_0$ in Figure \ref{diag:vstree_hierarchy}. Note that we refer to tree \textit{nodes} and graph \textit{vertices}.

For each tree node $n_T$, we select $m$ of its vertices as \textit{local landmarks}, e.g., $l_1$ and $l_2$ with $m=2$ in $S_1$ in Figure \ref{diag:vstree_graph} (landmarks for other nodes are omitted for clarity). We then compute a Subgraph Distance List $\SDL_i$ for each landmark $l_i$ of node $n_T$, containing the distances from $l_i$ to all vertices contained in the subgraph associated with $n_T$. Figure \ref{diag:vstree_hierarchy} shows the $\SDL$s for the landmarks of $S_1$ (those for other nodes are again omitted for clarity). $\SDL_i$ can be computed using multi-target Dijkstra search with the landmark $l_i$ as the source and all subgraph vertices as the targets. Note that the Dijkstra search must expand vertices irrespective of whether they are contained in $n_T$ or not until all vertices of $n_T$ are settled. This is to ensure correctness because shortest paths between $l_i$ and some subgraph vertices may leave and then re-enter the subgraph. In Section~\ref{sec:optimizations:search}, we present an optimized subgraph Dijkstra that restricts this search using landmarks in higher levels of the SUL-Tree. Furthermore, subgraph vertices are stored implicitly by vertex ID ranges using a subgraph vertex ordering presented in Section~\ref{sec:optimizations:ordering} (along with other space optimizations). Any graph partitioning method that can output similar-sized subgraphs can be used, e.g., we utilize the widely used METIS~\cite{karypis1998fast}.

\subsection{Object Index: Compacted Object-Landmark Tree (COL-Tree)}\label{sec:coltreeconstruction}
The SUL-Tree can be used to efficiently construct our object index, the Compacted Object-Landmark Tree (COL-Tree). COL-Tree is a carefully compacted version of the SUL-Tree for object set $P$. Compaction ensures that there are $m$ local landmarks for at most $\lambda$ objects in leaf nodes, to increase the likelihood of finding a tighter lower-bound for more objects. Note that the SUL-Tree is shared between initial construction and any rebuilding of all COL-Trees, i.e., for many different object sets and potentially multiple times.

Given SUL-Tree $T$ and object set $P$, a COL-Tree $C$ is constructed by visiting nodes in $T$ in a top-down manner and creating corresponding nodes in $C$. Let $n_T$ be the currently visited node in $T$ (initially the root). A corresponding node $n_C$ is created in $C$ for $n_T$. Let $\lambda$ be the maximum number of objects allowed in a leaf node of a COL-Tree. If $n_T$ contains more than $\lambda$ objects, the search expands to its children. Otherwise, the construction is pruned at $n_C$, which becomes a leaf node of $C$. For the new leaf $n_C$, an \textit{Object Distance List} $\ODL_i$ is created in $n_C$ for each landmark $l_i$ in $n_T$. These are simply the $\SDL$s of $n_T$, except with only the distances for object vertices from $P$. 
Figure \ref{diag:hpol_tree} shows a COL-Tree index for $\lambda=2$ constructed from the SUL-Tree in Figure \ref{diag:vstree_hierarchy} based on the 5 object vertices (green shaded vertices) in Figure \ref{diag:vstree_graph}. Note that $S_{1A}$ and $S_{1B}$ were removed as construction was pruned at $S_1$ due to its number of objects (the $\ODL$s of other nodes are omitted for clarity). Also note that we have chosen $\lambda < \alpha$ to simplify the example, but generally require $\lambda \geq \alpha$ to guarantee no leaf node, and therefore no $\ODL$, contains more than $\lambda$ objects. For example, if every vertex in a leaf node is an object there could be $\alpha$ objects in its $\ODL$.

Each $\ODL_i$ of a leaf node $n_{\mathrm{leaf}}$ in $C$ is sorted by distance. In non-leaf nodes $n_C$, we only store the minimum distance $M^- := min_{n_C,l_i}$ and maximum distance $M^+ := max_{n_C,l_i}$ to any object in the node from each landmark $l_i$. These can be computed by lookup of the Subgraph Distance List $\SDL_i$ of the SUL-Tree node corresponding to $n_C$. As we describe next, these bookkeeping components of COL-Tree allow computation of lower-bounds to nodes and traversal of the hierarchy.

\begin{figure}[htbp]
    \centering
    \begin{subfigure}[b]{0.35\linewidth}
        \includegraphics[width=\linewidth]{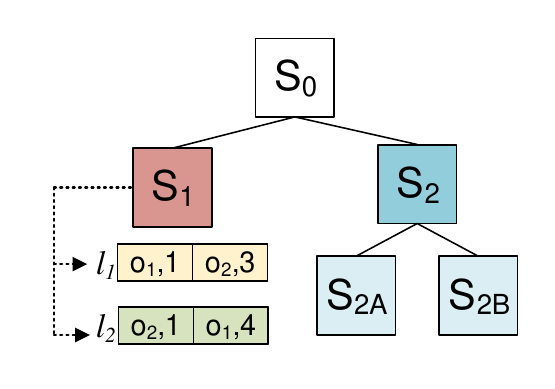}
        \caption{COL-Tree Hierarchy}
        \label{diag:hpol_tree}
    \end{subfigure}
    \hspace*{-3mm}
    \begin{subfigure}[b]{0.4\linewidth}
        \includegraphics[width=\linewidth]{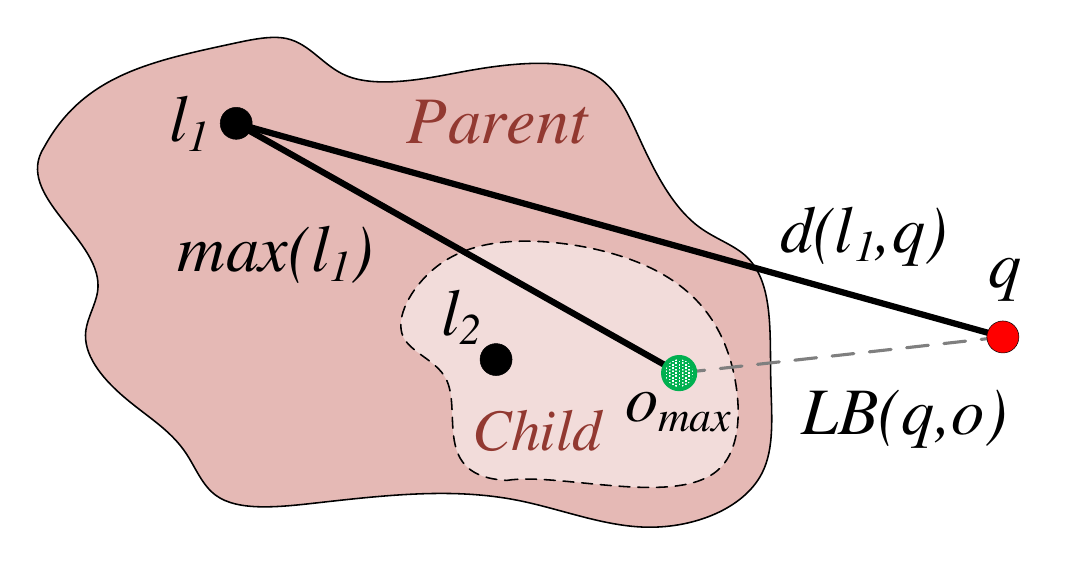}
        \caption{COL-Tree Lower-Bounds}
        \label{diag:hpol_lmk}
    \end{subfigure}
	\caption{Compacted Object-Landmark Tree (COL-Tree)}
	\label{diag:hpol}
\end{figure}

\subsection{Lower-Bound Heuristic for Graph Traversal}\label{index:COL-Treellbs}

Similar to Eq.~(\ref{eq:llb}), we can compute a lower-bound for all objects contained within a node $n_C$ in COL-Tree $C$ by Eq.~(\ref{eq:nodellb}) for one landmark $l_i$ of $n_C$. Eq.~(\ref{eq:nodellbmax}) gives the best lower-bound over all $m$ landmarks of $n_C$:

\begin{equation}\label{eq:nodellb}
LB_{l_i}(q,n_C) =
\begin{cases}
d(l_i,q)-M^+  & \mbox{if } d(l_i,q) \geq M^+ \\
M^- - d(l_i,q) & \mbox{if } d(l_i,q) \leq M^- \\
0 & \text{otherwise}
\end{cases}
\end{equation}

\begin{equation}\label{eq:nodellbmax}
\left. LB_{max}(q,n_C) \right. = \max\limits_{l_i \in n_C}(LB_{l_i}(q,n_C))
\end{equation}

Similar to Eq.~(\ref{eq:lub}), we can also compute an upper-bound on the distance to any object by Eq.~(\ref{eq:nodelubmax}). This will be highly useful for efficiently computing $k$FNs and other queries based on farthest metrics. 

\begin{equation}\label{eq:nodelubmax}
\left. UB_{min}(q,n_C) \right. = \min\limits_{l_i \in n_C}(d(l_i,q)+M^+)
\end{equation}

\smallHead{Network Distance Substitution:} $M^- := min_{n_C,l_i}$ and $M^+ := max_{n_C,l_i}$ for $n_C$ are already available in COL-Tree. However, for non-root nodes, $d(l_i,q)$ in Eq.~(\ref{eq:nodellb}) is only available if $l_i$ and $q$ are in the same subgraph. Pre-computing this distance for all $V$ and landmarks is infeasible given the space implications and defeats the purpose of using subgraph landmarks. Alternatively, computing $d(l_i,q)$ on the fly during query time is expensive and may be wasteful if the node does not contain results. Clearly, Eq.~(\ref{eq:nodellb}) still holds if we replace the distances with lower-bound $LB(l_i,q)$ and upper-bound $UB(l_i,q)$, as in Eq.~(\ref{eq:nodellb_relax}). The $\SDL$s of the SUL-Tree root or lowest common ancestor node can be conveniently used to compute the best $LB(l_i,q)$ and $UB(l_i,q)$ by Eq.~(\ref{eq:llbmax}) and Eq.~(\ref{eq:lub}), respectively. Choosing the tightest over all landmarks of $n_C$ gives an inexpensive and accurate bound even for a small number of landmarks:

\begin{equation}\label{eq:nodellb_relax}
LB_{l_i}(q,n_C) =
\begin{cases}
LB(l_i,q)- M^+   & \mbox{if } LB(l_i,q) \geq M^+ \\
M^- - UB(l_i,q) & \mbox{if } UB(l_i,q) \leq M^- \\
0                        & \text{otherwise}
\end{cases}
\end{equation}

\smallHead{Alternate Bound by Root Landmarks:} Another cheaper alternative to Eq.~(\ref{eq:nodellb}) can be computed using the root landmarks. During COL-Tree construction, in addition to $M^-$ and $M^+$ for the node's landmarks, we can also compute the minimum and maximum distances, $M^-_R$ and $M^+_R$, from the root landmarks to any object essentially for free. This is because $\SDL_i$ of SUL-Tree root node $n_R$ contains the distances from each root landmark $l_i$ to all vertices. Thus, during querying, we can similarly retrieve $d(q,l_i)$ in $O(1)$ time for each root landmark by looking up $\SDL_i$ of $n_R$ and utilize it in Eq.~(\ref{eq:nodellb}) with $M^-_R$ and $M^+_R$ to compute an alternate lower-bound. We can simply choose the tighter value between this and Eq.~(\ref{eq:nodellb_relax}). The alternate bound is likely to help when  $LB(l_i,q)$ or $UB(l_i,q)$ in Eq.~(\ref{eq:nodellb_relax}) are less accurate. 

\subsubsection{SUL-Tree Landmark Selection Policy}\label{index:landmarksel}
Landmark vertices must be chosen during SUL-Tree construction such that they produce tight lower-bounds irrespective of object set $P$. Past studies have shown that randomly chosen landmarks work reasonably well in  practice~\cite{goldberg2005alt}, which we also find to be true.  However, as observed in past work, some landmark selection policies may result in better lower-bounds. For example, \citet{goldberg2005alt} showed that landmarks at fringes of the graph resulted in better lower-bounds because they increased the likelihood of a source vertex being an ancestor (or having less deviation from the ancestor) of a target vertex in the shortest path tree from the landmark. Inspired by their approaches, we propose several landmark selection policies, which we evaluate in our experiments against a random baseline (\policy{Random}):

\begin{itemize}
    \itemsep0em
    \item \textit{Farthest Border} (\policy{Farthest}): Inspired by the \textit{farthest vertex} selection policy of \cite{goldberg2005alt}, we adapt this definition for subgraphs. Starting with a randomly chosen border vertex, the remaining border that is farthest from the landmarks chosen thus far is selected as the next landmark. If there are insufficient borders to select $m$ landmarks, random subgraph vertices are chosen. Note that this policy essentially computes $m{-}1$ $\SDL$s, which can be saved to avoid recomputation during SUL-Tree construction.
    \item \textit{Slice-Based Farthest Border} (\policy{Slice}): We first partition the plane into $m$ equally sized slices around the Euclidean center of the subgraph associated with the SUL-Tree node, as if we were cutting a cake. We choose a border as a landmark from each slice. If a slice has no borders, we choose the slice vertex farthest from the Euclidean center. If a slice has no vertices, we randomly choose a subgraph vertex to ensure $m$ landmarks. The intuition for this method is that there is always a landmark ``behind'' the objects and a query vertex, which is more likely to produce better lower-bounds.
    \item \textit{Border Minmax} (\policy{Minmax}): Existing work on landmark selection has focused on improving point-to-point shortest path queries, but their intuitions do not necessarily translate to object search. For example, a landmark ``behind'' the objects may only produce a good lower-bound for one query vertex, which is unlikely to help A$k$NN queries with multiple query vertices. To better accommodate such scenarios, we propose the \policy{Minmax} policy, which attempts to find an approximate subgraph  vertex center. After randomly choosing a small subset of subgraph borders, we compute the maximum distance from each subgraph vertex to these borders (performed efficiently via  Dijkstra search on the reverse graph). We choose the $m$ subgraph vertices with minimum maximum distances to the subset of borders.
\end{itemize}

\subsubsection{Tightening Lower-Bounds with Tree Depth}
The accuracy of COL-Tree's inexpensive lower-bounds increases as we delve deeper into the hierarchy. In Figure \ref{diag:hpol_lmk}, let us say we use $l_1$ and its maximum object distance to compute a lower-bound for the parent node. At the lower level, we may use the child's landmarks like $l_2$, which are local to the objects and more likely to produce a better lower-bound. This lets us differentiate tree branches and pinpoint the most promising candidates. Finally, at the leaf nodes (with objects), COL-Tree guarantees there are at most $\lambda$ objects per landmark to maximize differentiation.

\section{Query Algorithm for A\textit{k}NN Search}\label{query:aknn}

The \textit{decoupled heuristic} paradigm is a highly effective strategy for POI search in spatial networks such as road networks~\cite{abey2016knn,yao2018fann}. This involves iteratively retrieving candidate objects (POIs) by an admissible heuristic and using a separate shortest path index to compute exact distances. A set of the best candidates found so far is maintained with termination occurring at optimality when no remaining object can improve the result set. The exact cost computation is ``decoupled'' from the search heuristic allowing any shortest path index to be used, which have been extensively studied for decades~\cite{wu2012shortest,kawata2014phl}. Here we show that COL-Trees are ideally suited to answer A$k$NN queries using this paradigm. Specifically, COL-Tree considers fewer A$k$NN candidates to find the object with the minimum lower-bound using a property of $\ODL$s in leaf nodes.

\subsection{Object Distance Lists and Convexity}\label{query:convex}

Navigating the nodes of a COL-Tree to find the next best A$k$NN candidate object (i.e., with the minimum lower-bound) is a simple application of the per query vertex lower-bound described in Section~\ref{index:COL-Treellbs} and Lemma~\ref{lemma:agg}. However, once a leaf node $n_{\mathrm{leaf}}$ is encountered, we must retrieve the actual object vertex with the minimum lower-bound by inspecting $n_{\mathrm{leaf}}$'s $\ODL$s. A naive approach to retrieving this object is to inspect all $O(\lambda)$ objects and use the $\ODL$ distance to compute an aggregate lower-bound (by applying Eq.~(\ref{eq:llb}) per query vertex, combined by Lemma~\ref{lemma:agg}), with the lowest giving the best candidate. However, properties of the aggregate functions allow us to improve on this approach.

Eq.~(\ref{eq:llb}) can be expressed as an absolute value function of the form $f(x) = |C-x|$ where $C$ is the constant distance $d(l,q)$ between the landmark $l$ and a query vertex $q$. Then $\ODL_l$ of landmark $l$ essentially stores the sorted domain of $x$, i.e., the $d(l,o)\; \forall o \in n_{\mathrm{leaf}}$. We can find the object $p \in n_{\mathrm{leaf}}$ with the minimum lower-bound distance from $q$ by finding the minimizing $x$ for $f(x)$. Since absolute value functions are convex, this is easily achieved through a modified binary search for $d(l,p)$ closest to $d(l,q)$ in only $O(\log\lambda)$ time.  This was utilized for $k$NN search~\cite{abey2017knn}, i.e., involving only a single query vertex. 

It is not immediately clear whether this idea can be extended to minimum lower-bound \textit{aggregate} distance, which is complicated by the presence of multiple query locations $q_i \in Q$. From Lemma~\ref{lemma:agg}, we know the aggregate of \textit{lower-bound} distances from each query vertex $q_i$ is a lower-bound on the overall aggregate distance for monotonic functions~\cite{yiu2006aknn}. Therefore, the function to minimize becomes $f(x) = agg(|C_1-x|,\ldots,|C_n-x|)$ for a monotonic aggregate function $agg$ where $C_i$ is $d(l,q_i)$ for the given landmark $l$ and a query vertex $q_i \in Q$. At first glance, it appears we need to search $\ODL$s for multiple values (i.e., for each $C_i$) to find the object with minimum aggregate lower-bound. However, this is not necessary for aggregate functions that preserve convexity. In fact, the most widely used functions~\cite{papadias2005ann}, $max$ and $sum$, do preserve convexity. \citet{boyd2004convex} prove convexity preservation for a wide range of functions.

Specifically, the object minimizing $f(x)$ can be found by a single $O(\log\lambda)$ binary search for minimizing $x$, irrespective of the number of query vertices. Furthermore, subsequent candidates are retrieved through an $O(|Q|)$ left-right check, due to the convexity of the function. This allows the A$k$NN search to incrementally retrieve the objects by their minimum lower-bounds and enable termination at optimality without necessarily evaluating all $O(\lambda)$ elements in an $\ODL$, as in the naive approach. To perform the binary search, we must first determine the value of $x$ that minimizes the aggregation of the given absolute value functions (i.e., constants $C_i$). Lemma~\ref{lemma:sum}~\cite{schwertman1990simple} derives the minimizing value for $sum$. Similarly, Lemma~\ref{lemma:max} derives it for $max$, which follows from the one-center problem on a line in a straightforward manner.

\begin{lemma}\label{lemma:sum}
	Consider the aggregate function defined by the sum of a set of absolute value functions $f(x) = sum(|C_1-x|,\ldots,|C_n-x|)$. The minimizing $x^*$ of $f(x)$ is the median value of the constants $C_1,\ldots,C_n$.
\end{lemma}

\begin{proof}
	Let constants be sorted such that $C_1 \leq C_2 \leq \ldots \leq C_n$. Let $x^*$ be the median of these constants. We show that $f(x^*) \leq f(x')$ for all $x'$. Let $d=|x^*-x'|$. Without loss of generality, assume $x' < x^*$. For each $C_k<x^*$,  the difference between $|C_k - x^*|$ and $|C_k - x'|$ is at most $d$, i.e., $|C_k - x^*| - |C_k - x'| \leq d$. On the other hand, for each $C_j\geq x^*$, the difference between $|C_j-x'|$ and $|C_j-x^*|$ is exactly $d$, i.e., $|C_j - x'| - |C_j - x^*| = d$. Since $x^*$ is the median of the constants, the number of constants $C_j\geq x^*$ is at least $\lceil \frac{n}{2} \rceil$. In other words, for at least half of the constants, $|C_i-x'| - |C_i-x^*| = d$ and for each of the remaining constants, $|C_i-x^*| - |C_i-x'| \leq d$. Thus, $f(x^*) \leq f(x')$.
\end{proof}

\begin{lemma}\label{lemma:max}
	Consider the aggregate function defined by the maximum of a set of absolute value functions $f(x) = \max(|C_1-x|,\ldots,|C_n-x|)$. The minimizing $x^*$ of $f(x)$ is $\frac{C_{min}+C_{max}}{2}$, i.e., the average of the minimum and maximum constants.
\end{lemma}

\begin{proof}
	For sorted constants, let $C_1 = C_{min}$ and $C_n = C_{max}$. Note that $f(x)= \max(|C_1-x|,\ldots,|C_n-x|) = \max(|C_1-x|,|C_n-x|)$ as every intermediate $|C_i-x|$ must be no larger than one of those. The minimum value of $f(x)$ is $\frac{C_n-C_1}{2}$ as the maximum of two values is greater than or equal to their average, and $|C_1-x| + |C_n-x| \geq C_n - C_1$. Thus, by evaluating when the two sides are equal, the minimizing $x^*$ of $f(x)$ is $\frac{C_{min}+C_{max}}{2}$.
\end{proof}

\smallHead{Other Functions:} Object retrieval from $\ODL$s can be supported for any monotonic aggregate function (recall that monotonicity is required in existing work~\cite{yiu2006aknn}). First, one binary search for each query vertex locates the minimum of its absolute value function ($|C_i-x|$). Then, leveraging monotonicity, the Threshold Algorithm (TA)~\cite{fagin2003optimal} can be applied with left-right expansions until the minimum lower-bound object is found. TA is generally efficient as it allows early termination, but the first object may cost $O(|Q|\lambda)$ in the worst case compared to $O(\log\lambda)$ when the aggregate function is convexity-preserving. Nonetheless, we remark that $\lambda$ is still a relatively small constant (e.g., $\lambda{=}256$ for the US road network in Table~\ref{tab:datasets}). Without deriving its minimizer (cf.~$sum$ in Lemma~\ref{lemma:sum}), the TA approach for $median$ only provides the $O(|Q|\lambda)$ worst-case guarantee. Better guarantees without convexity preservation or a minimizer may be possible for some functions, e.g., $min$ costs $O(|Q|\log\lambda)$ for the first object as thresholding is not required because the minimum lower-bound object must be at one of the binary-searched $\ODL$ positions. 

\begin{algorithm}[htbp]
	\caption{Get A$k$NNs by COL-Tree for query vertex set $Q$}
	\begin{algorithmic}[1]
		\Function {GetAkNNs}{$k,Q,agg,\text{COL-Tree},\text{SUL-Tree}$}
		\State Initialize min queue $\mathcal{PQ} \leftarrow \emptyset$ and result set (max queue) $R \leftarrow \emptyset$
		\State \textproc{Insert}{($\mathcal{PQ},\text{COL-Tree}.root,0$)}, $D_k \leftarrow \infty$
		\While {\textproc{MinKey}{($\mathcal{PQ}$)}  $<D_k$ and $\mathcal{PQ}$ not empty}
			\State $c \leftarrow $ \textproc{Extract-Min}{($\mathcal{PQ}$)}
			\If {$c$ is an object}\label{line:hpolaknn:obj:s}
				\State Compute $d(q_i,c)\; \forall q_i \in Q$ then apply $agg$ to get $d_{agg}(Q,c)$ 
				\If {$d_{agg}(Q,c) < D_k$}\label{line:hpolaknn:obj:e:start}
					\State \textproc{Insert}{($R,c,d_{agg}(Q,c)$)}, if $|R| > k$ then \textproc{Extract-Max}$(R)$, if $|R| = k$ then $D_k \gets$ \textproc{MaxKey}$(R)$\label{line:hpolaknn:obj:e}
				\EndIf
			\ElsIf {$c$ is a non-leaf node}
				\ForAll {child node $e$ of $c$}\label{line:hpolaknn:nodenool:s}
					\State Compute best $LB_{agg}(Q,e)$ by Lemma \ref{lemma:agg} and node lower-bounds in Section~\ref{index:COL-Treellbs}\label{line:hpolaknn:nodenool:m}
					\State \textproc{Insert}{($\mathcal{PQ},e,LB_{agg,max}(Q,e)$)}\label{line:hpolaknn:nodenool:e}
				\EndFor
			\ElsIf {$c$ is a leaf node}
				\If {$c$ not seen before}\label{line:hpolaknn:nodeolfirst:s}
					\State Choose $\ODL_{c.l}$ by landmark $c.l$ in $c$ with max average distance from $Q$ by LB\label{line:hpolaknn:landmark}
					\State Compute $d(q,c.l)\; \forall q \in Q$ and save for reuse
					\State Binary search element $i$ in $\ODL_{c.l}$ with min $LB_{agg}(Q,c)$ by Lemma \ref{lemma:sum} or \ref{lemma:max}
					\State Let $c.\RP \gets i$ and $c.\LP  \gets i-1$\label{line:hpolaknn:nodeolfirst:e}
				\EndIf
				\State \textproc{RetrieveObjectsODL}{($\mathcal{PQ},Q,c$)}
				\State Compute $LB_{agg}(Q,c)$ for remaining $\ODL_{c.l}$ objects (will be by $c.\RP$ or $c.\LP$)
				\State \textproc{Insert}{($\mathcal{PQ},c,LB_{agg}(Q,c)$)} (unless $\ODL_{c.l}$ is fully searched)
			\EndIf
		\EndWhile
		\State \Return $R$
		\EndFunction
		\Function {RetrieveObjectsODL}{$\mathcal{PQ},Q,c$}
			\While{$\min(LB_{agg,c.l}(Q,\ODL_{c.l}[c.\LP]),LB_{agg,c.l}(Q,\ODL_{c.l}[c.\RP])) \leq$ \textproc{MinKey}{($\mathcal{PQ}$)}\label{line:hpolaknn:odlwhile}
			}
				\State{$p \gets$ object at $c.\LP$ or $c.\RP$ with smaller lower-bound by $\ODL_{c.l}$ \label{line:hpolaknn:leftorright}}	
				\State{Increment $c.\RP$ or decrement $c.\LP$ used to set $p$ at Line~\ref{line:hpolaknn:leftorright}}
                \State Compute best $LB_{agg}(Q,p)$ by $c.l$ or $\text{SUL-Tree}$ root landmarks
				\State \textproc{Insert}{($\mathcal{PQ},p,LB_{agg,max}(Q,p)$)}\label{line:hpolaknn:plbmax}			
		\EndWhile
		\EndFunction
	\end{algorithmic}
	\label{alg:hpolaknn}
\end{algorithm}

\subsection{Search Algorithm}

Algorithm~\ref{alg:hpolaknn} uses hierarchical traversal of the COL-Tree by node lower-bounds to guide us towards $\ODL$s in leaf nodes most likely to contain A$k$NN results. While SUL-Trees are an input to Algorithm~\ref{alg:hpolaknn}, only the root level is needed to compute lower and upper-bounds for Eq.~(\ref{eq:nodellb_relax}). The algorithm maintains a priority queue $\mathcal{PQ}$ containing objects and COL-Tree nodes keyed by their aggregate lower-bound distances from $Q$. The loop iteratively extracts the minimum lower-bound queue element. If an object is extracted, its exact aggregate distance is computed and the result set $R$ and $D_k$ are updated (Lines \ref{line:hpolaknn:obj:s} to \ref{line:hpolaknn:obj:e}). The result set $R$ maintains up to $k$ objects with the smallest aggregate distances found so far by the algorithm and $D_k$ corresponds to the largest aggregate distance among these $k$ objects. 

If a non-leaf node is extracted then an aggregate lower-bound score is computed according to Eq.~(\ref{eq:nodellbmax}) and Eq.~(\ref{eq:nodellb_relax}) for each of its child nodes (Lines \ref{line:hpolaknn:nodenool:s} to \ref{line:hpolaknn:nodenool:e}). If it is a leaf node, it is initialized if encountered for the first time (Lines \ref{line:hpolaknn:nodeolfirst:s} to \ref{line:hpolaknn:nodeolfirst:e}). An $\ODL$ to process is chosen (Line~\ref{line:hpolaknn:landmark}), the constants in the absolute value functions are computed, and a binary search is performed to find the minimizing list index given by Lemma \ref{lemma:sum} or \ref{lemma:max}. Pointer $\RP$ is initialized to the index of the minimizing value and $\LP$ to the index preceding it. Then objects are retrieved from the list using $\textproc{RetrieveObjectsODL}$ and, if the list is not completely searched, the node is re-inserted into $\mathcal{PQ}$ with minimum lower-bound computed by $\RP$ or $\LP$. $LB_{agg,max}(Q,p)$ on Line~\ref{line:hpolaknn:plbmax} computes the maximum (best) lower-bound aggregate distance for object $p$ over all landmarks for which network distances to $p$ are already computed (e.g., landmarks in the root SUL-Tree node or $c.l$ in the leaf). We use the lower-bound for $p$ to avoid unnecessary network distance computations (if $p$ is a result it will be processed at Line~\ref{line:hpolaknn:obj:e:start}). Note that the $LB_{agg,c.l}()$ term on Line~\ref{line:hpolaknn:odlwhile} refers to the lower-bound for all yet-to-be-retrieved objects in the $\ODL$ of leaf node $c$ corresponding to landmark $c.l$ (i.e., subject to $c.\RP$ and $c.\LP$).

Recall that leaf nodes contain $\ODL$s for each of their $m$ landmarks so that we can choose the $\ODL$ most likely to support the best candidate retrieval. For $k$NN queries, the landmark farthest from the query vertex is likely to be ``behind'' more objects and thereby result in better lower-bounds for more objects (Section~\ref{index:landmarksel}). Since A$k$NN queries involve multiple query vertices, at Line~\ref{line:hpolaknn:landmark}, we choose $c.l$ as the landmark with the largest average distance from the query vertices $Q$. Since we compute the lower-bound distance from query vertices to each landmark during node lower-bound computation using the SUL-Tree root's $\SDL$s (Line~\ref{line:hpolaknn:nodenool:m}), for leaf nodes we can also simply keep track of the landmark with largest average lower-bound distance. This provides an approximation on the farthest landmark essentially for free. Note $\ODL$ selection does not affect correctness, but some $\ODL$s may generate fewer false candidates.

\smallHeadIndent{Proof Sketch for Correctness:}
The intuition behind the correctness of Algorithm~\ref{alg:hpolaknn} is that the hierarchical tree search maintains a lower-bound aggregate distance for every unseen object in the priority queue at all times. That is, there is a queue element for either (a) the object itself; or (b) a tree node (i.e., subgraph) containing the object. Thus, when the algorithm terminates, lower-bounds for all remaining objects are at least $D_k$, the true aggregate distance for the $k$ result objects found so far. In other words, no other object can have an aggregate distance smaller than the current $k$ result objects. Non-result objects extracted from the queue (because the lower-bound was less than $D_k$) whose exact distance is at least $D_k$ cannot improve the current result set and can be safely discarded.

\section{Query Algorithm for \textit{k}FN Search}\label{query:fkn}

Euclidean $k$ Farthest Neighbors search can be performed through a simple variation on $k$NN search. For example, $k$NN search can be conducted using a lower-bound to all objects contained in an R-tree node by computing the minimum distance to its MBR. An upper-bound can be computed with similar ease, allowing Euclidean $k$FN algorithms to retrieve candidate objects by upper-bound and terminate when the next best (i.e., largest) upper-bound cannot improve the result set. \citet{papadias2003ine} observed that Euclidean distance can be used as a lower-bound for $k$NN search in road networks. Their technique, IER, was later shown to perform extremely well even when edge weights were travel-time~\cite{abey2016knn} (a lower-bound being obtained by dividing Euclidean distance by the maximum speed in the road network).

However, Euclidean distance \textit{cannot} similarly be used for $k$FN search in spatial networks. Although an upper-bound on the farthest Euclidean distance to an object can be obtained, the Euclidean distance can still only be used to compute a \textit{lower-bound} on network distance. Therefore, it cannot readily be used for $k$FN search as it is not admissible for this maximization problem. Moreover, existing network-based $k$FN methods require either evaluating all objects or an expensive object-set-specific pre-processing~\cite{wang2014afn}. An expansion-based approach would involve incrementally retrieving all objects until the farthest one is reached, meaning it is more efficient to simply evaluate all objects in a loop. Fortunately, COL-Trees can support incremental retrieval of the best candidate by \textit{upper-bound}, as we demonstrate next.

\subsection{Search Algorithm}

In contrast to A$k$NN search, $k$FN search in Algorithm~\ref{alg:fkn} relies on a maximum priority queue $\mathcal{PQ}$ keyed by the upper-bound distance to tree nodes and objects from the query $q$ to guide the hierarchical search. As in Algorithm~\ref{alg:hpolaknn}, priority queue elements that are objects, non-leaf nodes, and leaf nodes are handled separately but adapted for $k$FN search. For example, at Line~\ref {line:fkn:lub}, we utilize Eq.~(\ref{eq:nodelubmax}) to compute an upper-bound distance to all objects contained in a COL-Tree node. Another difference is that the best candidate object can be initialized in $O(1)$ time without a binary search as the final element in the $\ODL$ (Line~\ref{line:fkn:initialpoint}) and then we need only keep track of $c.\RP$ (Line~\ref{line:fkn:trackright}), the right-most element of the $\ODL$ not yet evaluated. This is because the $\ODL$ is sorted in increasing order of distance from the landmark and the object with largest distance from the landmark will clearly give the maximum (i.e., best) upper-bound by Eq.~(\ref{eq:lub}). 

The last major point of difference is that we choose the $\ODL$ associated with the landmark closest to the query, as it is more likely to lie on or be close to the shortest path from the query vertex to the objects, leading to more accurate upper-bounds~\cite{potamias2009distest}. An approximation on this landmark can be conveniently saved during node upper-bound computation for leaf nodes (Line~\ref{line:fkn:lub}) using SUL-Tree root $\SDL$s, by keeping track of the landmark with the smallest upper-bound distance to $q$. Algorithm~\ref{alg:fkn} retrieves the object with the next best (maximum) upper-bound until the result set $R$ cannot be improved, when it terminates returning the $k$ objects with the largest distances from $q$. 

Note that the $UB_{c.l}()$ term on Line~\ref{line:fkn:odlwhile} refers to the upper-bound for all yet-to-be-retrieved objects in the $\ODL$ of leaf node $c$ corresponding to landmark $c.l$ (i.e., subject to $c.\RP$).

\begin{algorithm}[htbp]
	\caption{Get $k$FNs by COL-Tree for query vertex $q$}
	\begin{algorithmic}[1]
		\Function {GetKFNs}{$k,q,\text{COL-Tree},\text{SUL-Tree}$}
		\State Initialize max queue $\mathcal{PQ} \leftarrow \emptyset$ and result set (min queue) $R \leftarrow \emptyset$
		\State \textproc{Insert}{($\mathcal{PQ},\text{COL-Tree}.root,\infty$)}, $D_k \leftarrow 0$
		\While {\textproc{MaxKey}{($\mathcal{PQ}$)}  $>D_k$ and $\mathcal{PQ}$ not empty}
			\State $c \leftarrow $ \textproc{Extract-Max}{($\mathcal{PQ}$)}
			\If {$c$ is an object}
				\State Compute $d(q,c)$
				\If {$d(q,c) > D_k$}
					\State \textproc{Insert}{($R,c,d(q,c)$)}, if $|R|>k$ then \textproc{Extract-Min}$(R)$, if $|R| = k$ then $D_k \gets $ \textproc{MinKey}$(R)$
				\EndIf
			\ElsIf {$c$ is a non-leaf node}
				\ForAll {child node $e$ of $c$}
					\State Compute best $UB(q,e)$ by node upper-bounds in Section~\ref{index:COL-Treellbs}\label{line:fkn:lub}
					\State \textproc{Insert}{($\mathcal{PQ},e,UB_{min}(q,e)$)}
				\EndFor
			\ElsIf {$c$ is a leaf node}
				\If {$c$ not seen before}
					\State Choose $\ODL_{c.l}$ by landmark $c.l$ in $c$ closest to $q$ by UB\label{line:fkn:selectlmk}
					\State Compute $d(q,c.l)$ and save for reuse
					\State Initialize $c.\RP$ to final element in $\ODL_{c.l}$\label{line:fkn:initialpoint}
				\EndIf
				\State \textproc{RetrieveObjectsODL}{($\mathcal{PQ},q,c$)}
				\State Compute $UB(q,c)$ for remaining $\ODL_{c.l}$ objects by $c.\RP$
				\State \textproc{Insert}{($\mathcal{PQ},c,UB(q,c)$)} (unless $\ODL_{c.l}$ is fully searched)
			\EndIf
		\EndWhile
		\State \Return $R$
		\EndFunction
		\Function {RetrieveObjectsODL}{$\mathcal{PQ},q,c$}
			\While{$UB_{c.l}(q,\ODL_{c.l}[c.\RP]) \geq$ \textproc{MaxKey}{($\mathcal{PQ}$)}}\label{line:fkn:odlwhile}
				\State $p \gets$ object at $c.\RP$ and decrement $c.\RP$\label{line:fkn:trackright}
                \State Compute best $UB(q,p)$ by $c.l$ or $\text{SUL-Tree}$ root landmarks
				\State \textproc{Insert}{($\mathcal{PQ},p,UB_{min}(q,p)$)}
		      \EndWhile
		\EndFunction
	\end{algorithmic}
	\label{alg:fkn}
\end{algorithm}

\section{Query Algorithm for Range Search}\label{query:range}

While the benefit of COL-Tree's hierarchical approach can be clearly intuited in the case of A$k$NN and $k$FN queries, whether this translates to purely proximity-based search problems such as $k$NN and range queries is less clear. The state-of-the-art $k$NN search heuristic~\cite{abey2017knn} utilizes an adjacency property of Network Voronoi Diagrams (NVD) to incrementally retrieve the next nearest object from a pool of candidates likely to contain it. Such greedy incremental approaches translate well for $k$NN queries, where our goal is to find the $k$ results and terminate as soon as possible. While range queries also retrieve objects by proximity, they do not have the benefit of stopping after a fixed number of results, instead retrieving all objects within a radius $r$. Since a COL-Tree can be used to compute lower- and upper-bounds to all objects in entire subgraphs, this can be leveraged to determine whether multiple objects are entirely within or outside radius $r$ in the best case. A COL-Tree approach is further motivated by the significant $O(|V|\log|V|)$ pre-processing time of NVDs~\cite{abey2017knn} compared to the $O(m|P|\log|P|)$ of COL-Trees (Table~\ref{tab:complexity}). In this section we present an algorithm to use COL-Trees to efficiently answer range queries based on network distance. 

\subsection{Contiguous Sublist of Object Distance Lists}

Similar to using convexity preservation to identify candidate objects from $\ODL$s for A$k$NN search, the candidates for range search can be seen to form a contiguous sublist of the $\ODL$ as demonstrated in Lemma~\ref{lemma:rangesublist}.

\begin{lemma}\label{lemma:rangesublist}
Consider a range query from query vertex $q$ with radius $r$, and the distance list $\ODL_l$ for landmark $l$ in any leaf node. Candidate objects, i.e., those that can be neither pruned nor confirmed using $LB_l(q,p)$ and $UB_l(q,p)$, form a contiguous sublist of $\ODL_l$. Specifically, they are exactly the objects $p$ whose distance $d(l,p)$ lies in the interval
\[
I =
\begin{cases}
(r - d(l,q), d(l,q) + r] & \text{if } d(l,q) \leq r \\
[d(l,q) - r, d(l,q) + r] & \text{if } d(l,q) > r
\end{cases}
\]
\end{lemma}

\begin{proof}
By Eq.~(\ref{eq:llb}), $p$ is pruned by $l$ if and only if $LB_l(q,p) = |d(l,q) - d(l,p)| > r$. Similarly, by the single-landmark form of Eq.~(\ref{eq:lub}), $p$ is confirmed by $l$ if and only if $UB_l(q,p) = d(l,q) + d(l,p) \leq r$.

\textit{Case $d(l,q) \leq r$:} By rearranging the equations, object $p$ is confirmed if and only if $d(l,p) \leq r - d(l,q)$, and pruned if and only if $d(l,p) > d(l,q) + r$.  Note that the other expansion of the absolute value, $d(l,p) < d(l,q) - r$, cannot be true, since $d(l,q) - r \leq 0$ and distances are non-negative. Thus, $p$ is a candidate object if and only if it meets neither condition:
\[
r - d(l,q) < d(l,p) \leq d(l,q) + r
\]
This means $d(l,p)$ must lie in the interval $I = (r - d(l,q), d(l,q) + r]$ when $d(l,q) \leq r$.

\textit{Case $d(l,q) > r$:} No object $p$ in $\ODL_l$ can ever be confirmed, since this requires $d(l,p) \leq r - d(l,q) < 0$, which cannot be true for non-negative distances. As in the first case, object $p$ is pruned if and only if $d(l,p) < d(l,q) - r$ or $d(l,p) > d(l,q) + r$. Therefore, $p$ is a candidate object if and only if it meets neither prune condition:\[
d(l,q) - r \leq d(l,p) \leq d(l,q) + r
\]
This means $d(l,p)$ must lie in the interval $[d(l,q) - r, d(l,q) + r]$ when $d(l,q) > r$.

As $\ODL_l$ is sorted by $d(l,p)$ and $I$ is an interval, candidate objects form a contiguous sublist in both cases. 
\end{proof}

By identifying the case for a given $\ODL$ and conducting a binary search for the start of
interval $I$, all candidates of the $\ODL$ can be retrieved sequentially. Note that if $d(l,q) \leq r$, it becomes easier to both confirm and prune objects as $I$ shortens as the landmark approaches the query vertex. If $d(l,q) > r$, with increasing distance between the landmark and query, $I$ moves towards the end of the $\ODL$ and eventually falls outside of it at which point all objects can be pruned. The same principles can be applied to non-leaf nodes and their node bounds. 

\subsection{Search Algorithm}

Algorithm~\ref{alg:range} describes how to use COL-Tree for range queries. Unlike all competing methods, COL-Tree range search \textit{does not} require a priority queue, and exploration can be implemented using a stack with reduced insertion and extraction costs (Line~\ref{line:range:initialization}). Note that either DFS or BFS of the COL-Tree can be used as both will visit the same nodes. If the node $c$ popped from the stack is a non-leaf node, we evaluate each of its child nodes $e$ by computing their lower- and upper-bounds. If the upper-bound is less than or equal to the radius $r$ (Line~\ref{line:range:nonleafub}), then all objects must be results. An inner stack can be used to conduct a DFS from $e$ to immediately retrieve all objects contained in the leaf node descendants of $e$ (note that $e$ may be a leaf node) and insert them in result set $R$ (Line~\ref{line:range:allobjresults}). However, $e$ is pushed to the stack for further exploration if its upper-bound is greater than $r$ and the lower-bound is less than or equal to $r$ (Line~\ref{line:range:lboverlap}), since it may have results. $e$ is implicitly pruned if the lower-bound is greater than $r$ as it cannot have results. 

Alternatively, if $c$ is a leaf node, we evaluate its objects using \textproc{RetrieveObjectsODL}. First we select the $\ODL_{c.l}$ (Line~\ref{line:range:odllmk}) and determine the candidate sublist in $\ODL_{c.l}$ by binary search (Line~\ref{line:range:start}). If $d(c.l,q) \leq r$, all objects before the start of the sublist ($\LP$) are confirmed and inserted into $R$ (Line~\ref{line:range:confirm}), otherwise they are (implicitly) pruned. The sublist is then scanned and candidate objects are decided exactly (Line~\ref{line:range:scan}). Any objects beyond $\RP$ are implicitly pruned (recall that the $\ODL$ models a convex function). Note that inexpensive lower- and upper-bounds computed using SUL-Tree root landmarks can be used to further filter objects and potentially avoid expensive network distance computations (Line~\ref{line:range:verify}).

To elaborate on $\ODL_{c.l}$ choice at Line~\ref{line:range:odllmk}, the cost of evaluating an $\ODL$ is the number of candidates in its sublist within $I$, as identified by Lemma~\ref{lemma:rangesublist}. To avoid computing expensive $d(c.l,q)$ for each landmark $c.l$ of $c$, we approximate $I$ using the bounds $LB(c.l,q) \leq d(c.l,q) \leq UB(c.l,q)$ for each landmark, which are obtained essentially for free during leaf node bound computation (Line~\ref{line:range:nodebound}). By straightforward substitution into Lemma~\ref{lemma:rangesublist}'s conditions, $I$ must be within $(r - UB(c.l,q), UB(c.l,q) + r]$ if $UB(c.l,q) \leq r$ or within $[LB(c.l,q) - r, UB(c.l,q) + r]$ if $LB(c.l,q) > r$, and in $[0, UB(c.l,q) + r]$ otherwise. An approximate number of candidates for each $\ODL_{c.l}$ can then be obtained using only two binary searches per $\ODL_{c.l}$ (thus $O(\log\lambda)$) and we choose the landmark $c.l$ with the fewest.

\begin{algorithm}[htbp]
	\caption{Get objects within radius $r$ from query vertex $q$ by COL-Tree}
	\begin{algorithmic}[1]
		\Function {RangeSearch}{$r,q,\text{COL-Tree},\text{SUL-Tree}$}
		\State Initialize stack $\mathcal{S} \leftarrow \text{COL-Tree}.root$ and result set $R \leftarrow \emptyset$\label{line:range:initialization}
		\While {$\mathcal{S}$ not empty}
			\State $c \leftarrow $ \textproc{Pop}{($\mathcal{S}$)}
			\If {$c$ is a non-leaf node}
				\ForAll {child node $e$ of $c$}
					\State Compute best $UB(q,e)$ and $LB(q,e)$ by node bounds in Section~\ref{index:COL-Treellbs}\label{line:range:nodebound}
                    \If {$UB_{min}(q,e) \leq r$}\label{line:range:nonleafub}
                        \State Retrieve all objects $p \in e$ by DFS and insert into $R$\label{line:range:allobjresults}
                    \ElsIf {$LB_{max}(q,e) \leq r$}
                        \State \textproc{Push}{($\mathcal{S},e$)}\label{line:range:lboverlap}
                    \EndIf
				\EndFor
			\ElsIf {$c$ is a leaf node}
				\State \textproc{RetrieveObjectsODL}{($R,q,c$)}
			\EndIf
		\EndWhile
		\State \Return $R$
		\EndFunction
		\Function {RetrieveObjectsODL}{$R,q,c$}
			\State Choose $\ODL_{c.l}$ with smallest approximate number of candidates by bound-based $I$\label{line:range:odllmk}
			\State Compute $d(c.l,q)$ and exact $I$ by Lemma~\ref{lemma:rangesublist}
			\State Find the index range $[\LP, \RP]$ of the objects of $\ODL_{c.l}$ in $I$ via binary search\label{line:range:start}
			\If {$d(c.l,q) \leq r$}
                \State Insert all objects $p$ of $\ODL_{c.l}$ before index $\LP$ into $R$\label{line:range:confirm}
			\EndIf
            \ForAll {object $p$ in $\ODL_{c.l}[\LP], \ldots, \ODL_{c.l}[\RP]$}\label{line:range:scan}
				\If {$UB_{min}(q,p) \leq r$}\label{line:range:verify}
					\State Insert $p$ into $R$
				\ElsIf {$LB_{max}(q,p) \leq r$}
					\State Compute $d(q,p)$ and insert $p$ into $R$ if $d(q,p) \leq r$
				\EndIf
			\EndFor
		\EndFunction
	\end{algorithmic}
	\label{alg:range}
\end{algorithm}

\section{Complexity Analysis}\label{sec:complexity}

SUL-Trees and COL-Trees enable hierarchical search at the expense of increased pre-processing costs over existing landmark-based indexes such as ALT (which, in any case, cannot support efficient object search). Although additional pre-processing cost is amortized with faster querying, we attempt to minimize this cost as much as possible, e.g., by limiting tree depth with $\lambda$. In this section we analyze the time and space complexity to show that the increased pre-processing cost is small or comparable.

\smallHead{COL-Tree:} A COL-Tree is created from a SUL-Tree, which is a complete $b$-ary tree. In a COL-Tree for object set $P$, there are at most $O(\frac{|P|}{\lambda})$ leaf nodes. The top-down conversion of the SUL-Tree may result in the COL-Tree being unbalanced as objects may be disproportionately distributed in SUL-Tree nodes. But during the conversion, empty SUL-Tree nodes are pruned, resulting in the total number of nodes in COL-Tree being the same as a complete $b$-ary tree, i.e., $O(\frac{|P|}{\lambda})$ (branching factor $b$ is a small constant). However, COL-Tree's space complexity is dominated by the $m$ Object Distance Lists ($\ODL$s) of length $\lambda$ stored in each leaf node. This results in $O(m|P|)$ total space as each object is only stored in one leaf node. Compression of COL-Tree's height via implicit merging of sibling-less child nodes into parent nodes results in an average COL-Tree depth of $O(\log|P|)$. Given the top-down conversion and $O(1)$ look-ups (e.g., if implemented as hash-tables) of SUL-Tree Subgraph Distance Lists ($\SDL$s) to construct $\ODL$s, propagating $|P|$ objects to build $m$ $\ODL$s and computing $M^-$/$M^+$ takes $O(m|P|\log|P|)$ time. Sorting all $\ODL$s takes $O(m|P|\log\lambda)$ time where $\lambda$ is a small constant. Thus, the total time complexity is $O(m|P|\log|P|)$. Note that $m$, the number of landmarks per node, is a small constant.

\smallHead{SUL-Tree:} The height of the balanced SUL-Tree, with branching factor $b$ and $\frac{|V|}{\alpha}$ leaf nodes, is $O(\log_{b} \frac{|V|}{\alpha})$, i.e., $O(\log|V|)$ since $b$ and $\alpha$ are small constants. Given the disjoint partitioning, each vertex belongs to at most one tree node at each level, meaning there are $O(|V|)$ vertices per level. Then all $\SDL$s associated with the SUL-Tree will take $O(m|V|\log|V|)$ space, since we store the distance from every vertex to $m$ landmarks at each level. This, of course, dominates the $O(|V|)$ space occupied by the tree nodes. For level $i$, there will be at most $b^i$ nodes, each with $\frac{|V|}{b^i}$ vertices. If a Dijkstra's search from each of the $m$ landmarks is used to compute the distances to the node's vertices, then this may take $O(b^i m\frac{|V|}{b^i}\log\frac{|V|}{b^i})$ time for all nodes at level $i$. However, the Dijkstra's search may visit vertices outside the SUL-Tree node's subgraph, e.g., if a shortest path from a landmark to some subgraph vertex passes outside the subgraph. We consider $\gamma \times \frac{|V|}{b^i}$ to represent the total number of vertices visited by the Dijkstra search. We find that $\gamma$ is a small constant in practice (and can be further reduced as shown in Section~\ref{sec:optimizations}). Multiplying by the height, the time complexity for the whole tree is $O(m|V|\log^2{|V|})$. 

\begin{table}[htb]
\centering
\begin{tabular}{|c||*{4}{c|}} \hline
\textbf{Complexity} & \textbf{ALT~\cite{goldberg2005alt}} & \textbf{SUL-Tree (Ours)} & \textbf{OL~\cite{abey2017knn}} & \textbf{COL-Tree (Ours)} \\ \hline\hline
Space & $O(m|V|)$        & $O(m|V|\log|V|)$       & $O(m|P|)$        & $O(m|P|)$        \\ \hline
Time  & $O(m|V|\log|V|)$ & $O(m|V|\log^2{|V|})$   & $O(m|P|\log|P|)$ & $O(m|P|\log|P|)$ \\ \hline
\end{tabular}
\caption{Comparison of space and time complexity for the proposed methods versus previous landmark-based indexes}
\label{tab:complexity}
\end{table}

Table~\ref{tab:complexity} shows the theoretical space and time complexities of our techniques compare favorably with existing landmark-based graph and object indexes, ALT and Object Lists (OL)\footnote{We derive the time complexity of OL as it was not listed in~\cite{abey2017knn} (which is dominated by sorting $m$ lists of $|P|$ objects).} (resp.), which do not support hierarchical object search. COL-Tree offers the same space and time as OL, while SUL-Tree's space and time cost only increases by a factor of $O(\log|V|)$ over ALT.

\section{Pre-Processing Optimizations}\label{sec:optimizations}
 
\subsection{Subgraph Vertex Ordering}\label{sec:optimizations:ordering}
In Section~\ref{sec:coltreeconstruction}, the top-down conversion of a SUL-Tree to a COL-Tree for object set $P$ requires the root-to-leaf propagation of each object $p \in P$. To do this efficiently, we must determine the child node $p$ belongs to in $O(1)$ time. We must also retrieve the distance from each landmark to $p$ from the $\SDL$ of each relevant SUL-Tree node in $O(1)$ time. However, in practice, achieving $O(1)$ containment and distance retrieval has unfavorable implications for memory (both in terms of total usage and access latency). First, this would necessarily require storing the same vertex at each level of the SUL-Tree (i.e., each node/subgraph a vertex is associated with). Second, $O(1)$ operations require using highly memory-inefficient data structures like hash-tables. While hash-tables are an elegant theoretical construct, actual implementations of hash-tables such as \texttt{std::unordered\_map} in C\texttt{++} take up significantly more memory than just their elements. Hash-tables also involve increased memory latency~\cite{abey2016knn}, which can often significantly reduce algorithm performance~\cite{sidlauskas2014impl}. To avoid such issues one could attempt to only store vertices in the leaf nodes of the SUL-Tree. With only implicit storage of vertices in non-leaf nodes, the propagation of object $p \in P$ would need to be performed in a bottom-up manner taking $O(\log|V|)$ time. This would drastically affect COL-Tree construction time (Section~\ref{sec:complexity}), making it dependent on $|V|$ rather than the generally far smaller $|P|$. 

However, these issues can be entirely avoided by using a novel \textit{subgraph vertex ordering} which groups vertices of child subgraphs recursively as shown in Figure \ref{diag:stree_order} and numbering vertex IDs based on this order. By simply storing the first and last ID of vertices contained by a SUL-Tree node similar to an interval tree, we can achieve the necessary $O(1)$ complexity operations without resorting to highly memory-inefficient data structures like hash-tables, which are problematic as discussed above. In fact, as an added benefit of the subgraph vertex ordering, this significantly reduces the space cost of the SUL-Tree (i.e., more than just memory usage). The subgraph vertex IDs do not need to be stored \textit{at all} (the vertex IDs can be inferred from the $\SDL$ index). The distances can simply be stored in an array in vertex ID order with access to the distance for a specific vertex being achieved through simple arithmetic using the first vertex ID of the containing node. Although it does not change the complexity, for large-scale road networks, this can reduce real-world index size by gigabytes. We evaluate all improvements experimentally in Section~\ref{sec:exp:improvements}.

{
	\begin{figure}[htbp]
		\centering
		\includegraphics[width=0.6\linewidth]{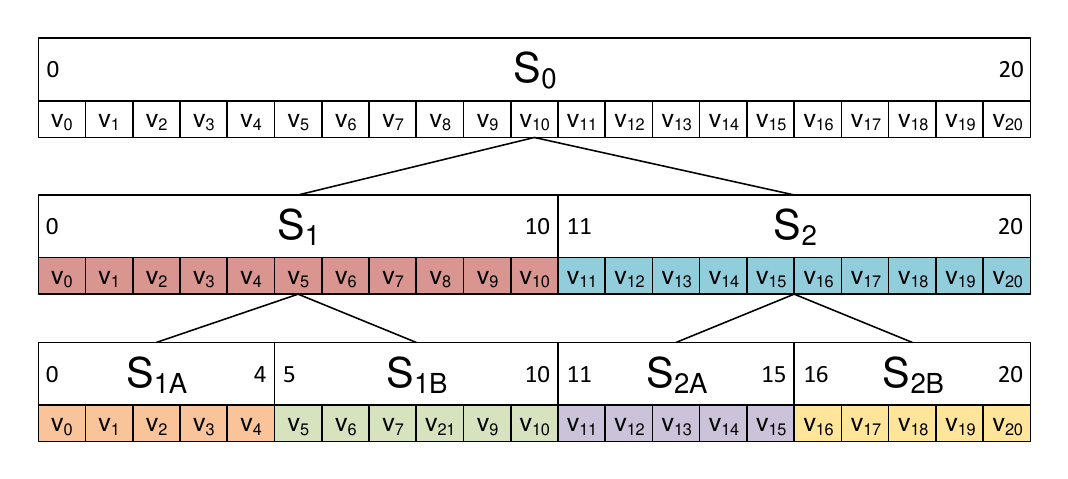}
		\caption{Hierarchical Subgraph Vertex Ordering}
		\label{diag:stree_order}
	\end{figure}
}

\smallHeadIndent{Improved Dijkstra Implementation:} Array-based data structures offer significantly better practical performance due to their locality of reference compared to hash-tables, enabling speedup through efficient utilization of CPU cache~\cite{abey2016knn}. As another benefit of replacing hash-tables, the subgraph vertex ordering allows substantial improvement in the implementation efficiency of SUL-Tree construction. Specifically, $\SDL$ computation involves a restricted Dijkstra search which terminates when all subgraph vertices are settled. Checking whether a settled vertex is a subgraph vertex can be performed by simply checking if the vertex ID is in the range for that subgraph. Moreover, the distance itself can be stored directly in an array-based $\SDL$. $\SDL$ computation time can be reduced by 34\% from simply avoiding hash-tables (Table~\ref{tab:improvements}).

\subsection{Subgraph Restricted Dijkstra Search}\label{sec:optimizations:search}
The most intensive part of SUL-Tree construction is the creation of $\SDL$s for each SUL-Tree node $n_T$. This involves a Dijkstra search from each of the $m$ landmarks to find the shortest path distances to all vertices contained in the subgraph associated with $n_T$. Although the recursive partitioning step ensures subgraphs are disjoint vertex sets, there is still potential for repeated work as there is no guarantee that shortest paths will \textit{only} involve subgraph vertices. It is possible that a shortest path from a landmark to a subgraph vertex leaves and re-enters the subgraph. To ensure correct distances, the Dijkstra search cannot terminate until all subgraph vertices are settled. In Section~\ref{sec:complexity}, we used $\gamma$ to represent the average factor by which subgraph Dijkstra search exceeded its number of vertices. Ideally, $\gamma=1$, meaning only subgraph vertices were settled. While it is small in practice ($\gamma \approx 6.5$), it can be further reduced using a novel variant of Dijkstra inspired by COL-Tree search, as we describe next.

{
	\begin{figure}[htbp]
		\centering
		\includegraphics[width=0.4\linewidth]{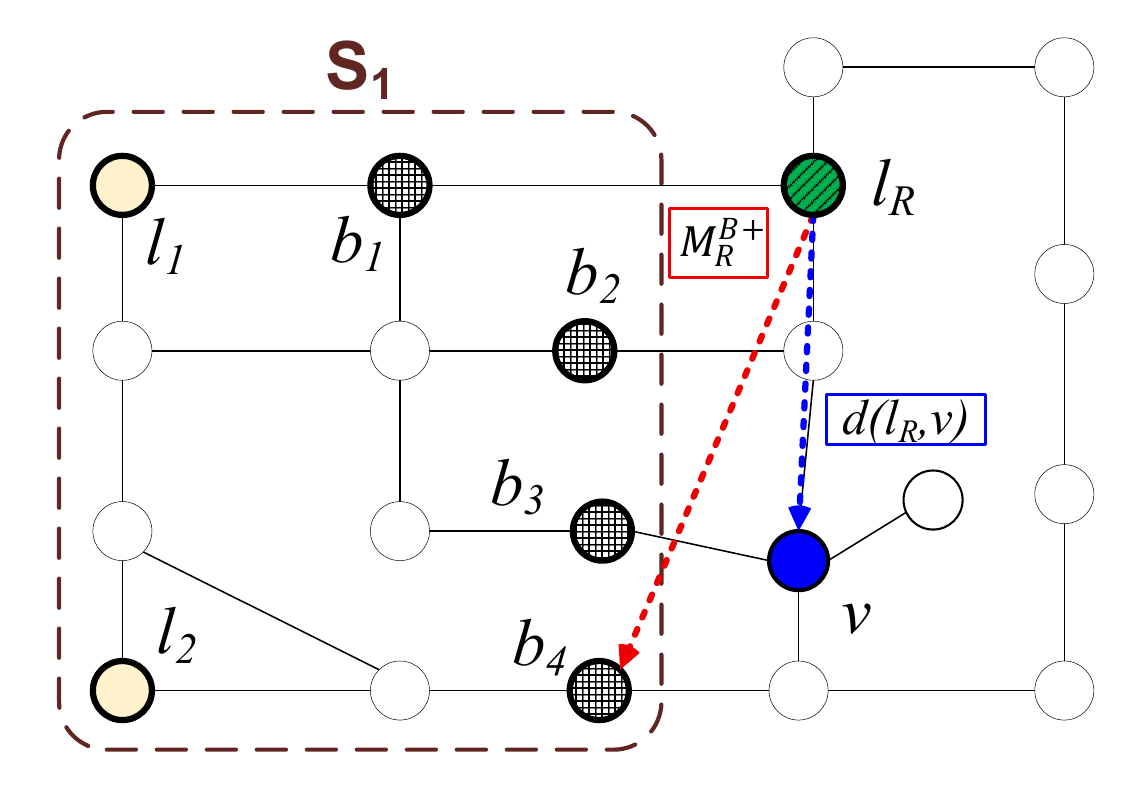}
		\caption{Lower-Bound Distance to Subgraph Borders}
		\label{diag:stree_border_set}
	\end{figure}
}

\smallHeadIndent{Border Set Lower-Bounds:} We observe that the COL-Tree object search mechanism can be used to improve Dijkstra subgraph search as well. First, we ensure that SUL-Tree nodes are processed in partitioning order, so that root node $\SDL$s can be accessed during $\SDL$ computation for any descendant node. After a subgraph is generated from the partitioning step, we identify the border vertex set $B$, which contains the subgraph vertices that are adjacent to a non-subgraph (i.e., external) vertex, e.g., $B = \{b_1,b_2,b_3,b_4\}$ in Figure~\ref{diag:stree_border_set}. We also keep track of the max $M^{B+}_R$ and min $M^{B-}_R$ distance from each root landmark $l_R$ to any border vertex by retrieving landmark-border distances from the $\SDL$ of the SUL-Tree root node, which comes at negligible cost during $\SDL$ population. 

The optimization uses the fact that any shortest path between subgraph vertices through an external vertex must re-enter the subgraph through some border vertex. Thus, we can compute a lower-bound from an external vertex $v$ to the subgraph by using $M^{B+}_R$ and $M^{B-}_R$ for root landmark $l_R$ using Eq.~(\ref{eq:nodellb}) and choose the best lower-bound by Eq.~(\ref{eq:nodellbmax}). This is similar to the way COL-Trees are used to compute lower-bounds to objects contained in a subgraph (in this case, the borders are essentially the objects). $d(l_R,v)$ for use in Eq.~(\ref{eq:nodellb}) is already available and can be retrieved in $O(1)$ time from the root $\SDL_R$ for $l_R$.  For example, in Figure~\ref{diag:stree_border_set}, the shortest path from subgraph landmark $l_2$ to some subgraph vertices may leave through $b_4$, pass through vertex $v$, and re-enter via $b_3$. A lower-bound from $v$ to any border of subgraph $S_1$ can be computed using already available $d(l_R,v)$ and $M^{B+}_R$ for root landmark $l_R$, potentially pruning the Dijkstra search at $v$ as the lower-bound may result in it never being expanded. Note that the heuristic may be inconsistent, as vertex $v$ farther from source $l$ may have a smaller lower-bound $d(l,v)+LB(v,B)$ to any subgraph vertex. While still admissible (as a lower-bound on the true distance), we employ the widely used pathmax heuristic~\cite{mero1984pathmax} to ensure the lower-bound to the subgraph $d(l,v)+LB(v,B)$ does not decrease, allowing use of faster priority queues such as radix heaps. This is still correct because $d(l,v)+d(v,B) \geq d(l,u)+LB(u,B)$ must be true if the shortest path to $v$ goes via $u$, i.e., there cannot be a strictly shorter path back to the subgraph via $u$ by moving away from it.

\smallHeadIndent{Border Set Completion:} Even with border set lower-bounds, the search must continue until all subgraph vertices are settled. We observe that the subgraph vertex farthest from a source landmark is often farther than every border vertex, e.g., caused by a long dead-end road. This penalizes the entire Dijkstra search, as it must continue to settle vertices until the farthest subgraph vertex is found. However, once all border vertices in $B$ are settled by the search, external vertices no longer need to be expanded (including those already relaxed and inserted in the search queue). Once all borders in border vertex set $B$ are settled, $d(l,b)$ is final for all $b \in B$. Any path from $l$ that re-enters the subgraph via an unsettled external vertex $v$ does so through some $b \in B$ and cannot be shorter than $d(l,b)$, i.e., no shorter path can be found back to the subgraph via $v$. After settling the last border, which can be detected via a simple count, external vertices are no longer expanded and any unsettled subgraph vertices can be found by only expanding within the subgraph. This approach works in combination with border set lower-bounds, which delay expansion of external vertices allowing the final border to be settled sooner. Furthermore, both improvements are applicable to Dijkstra search in landmark selection policies such as \policy{Farthest} and \policy{Minmax}, and we use them in our implementation.

\section{Experiments}\label{sec:exp}

\subsection{Experimental Settings}\label{exp:settings}

\smallHead{Environment:} 
We conduct experiments on a PC with an AMD Ryzen 7 9700X CPU with 64GB of RAM running Ubuntu 26.04 LTS (64-bit). Out of an abundance of caution, we disable CPU features such as SMT and PBO to avoid any possible inconsistencies, e.g., due to unexpected changes in clock speed. Code was written in C\texttt{++} and compiled by g\texttt{++} v15.2.0 with -O3 and -march=native flags and is publicly available as a GitHub repository\footnote{\url{https://github.com/tenindra/rn-obj-search}}. The repository includes scripts to reproduce all figures in the paper and details on extending the experimental framework. All experiments were conducted using memory-resident indexes for fast query processing.

\smallHead{Graph Datasets:} We evaluate four real-world spatial networks listed in Table~\ref{tab:datasets}. Our main dataset is a large-scale road network for the continental United States (US) with travel time edge weights, which we use for experiments unless otherwise stated. The road network for the Western US (W-US), with both travel time and distance edge weights, is used to investigate the impact of the edge weight metric on heuristic performance (Section~\ref{sec:exp:rndist}). To demonstrate the applicability of our techniques beyond road networks, we evaluate two other types of spatial networks (Sections~\ref{sec:exp:naca0015} and~\ref{sec:exp:srtmgl3}): NACA0015, a 2D mesh network representing an airfoil, and SRTMGL3, a terrain graph for a $1^\circ\times1^\circ$ tile in Colorado covering part of the Rocky Mountains. The road networks are curated versions of the datasets from the 9th DIMACS Challenge\footnote{\url{http://www.dis.uniroma1.it/\%7Echallenge9/}} and the mesh is from the 10th DIMACS Challenge\footnote{\url{https://sites.cc.gatech.edu/dimacs10/index.shtml}}. SRTMGL3 is constructed from a NASA Shuttle Radar Topography Mission (SRTM) 3-arcsecond elevation tile for N39W107~\cite{nasa_jpl_2013}, where each vertex is the center of a 72m$\times$93m grid cell and edges connect to its 8 neighbors with distance edge-weights accounting for the change in elevation. All four datasets are available via a Zenodo repository~\cite{abey2025rndataset}. Per dataset SUL-Tree and COL-Tree parameters take the values indicated in the relevant column of Table~\ref{tab:datasets} unless otherwise stated.

\smallHead{Object Datasets:}
For the US road network, we use real-world locations for $8$ types of POIs obtained from OpenStreetMap\footnote{\url{http://www.openstreetmap.org}}, as listed in Table~\ref{tab:pois}, which are also available via the Zenodo repository~\cite{abey2025rndataset}. For all graphs, including US, we also generate synthetic object sets by choosing vertices uniformly at random for a given density $d{=}\frac{|P|}{|V|}$. Synthetic object sets are used for all sensitivity analyses.

\begin{table}[htb]
\centering
\begin{tabular}{|c||c|*{2}{r|}c|c|c|c|} \hline
\textbf{Name} & \textbf{Type} & $\boldsymbol{|V|}$ & $\boldsymbol{|E|}$ & \textbf{Edge Weights} & $\boldsymbol{b}$ & $\boldsymbol{\alpha}$  & $\boldsymbol{\lambda}$  \\ \hline\hline
US       & Road Network & 23,947,347 & 57,708,624 & Time            & 8  & 1024  & 256  \\ \hline
W-US     & Road Network & 6,262,104  & 15,119,284 & Time, Distance  & 8  & 512   & 128  \\ \hline
NACA0015 & 2D Mesh      & 1,039,183  & 6,229,636  & Distance        & 4  & 256   & 64   \\ \hline
SRTMGL3  & Terrain      & 1,442,401  & 11,524,800 & Distance        & 4  & 256   & 64   \\ \hline
\end{tabular}
\caption{Datasets tested and their features/statistics}
\label{tab:datasets}
\end{table}

\begin{table}[htb]
\begin{minipage}{0.49\linewidth}
\centering
\begin{tabular}{|c||*{2}{r|}} \hline
\textbf{Object Set} & \textbf{Size} & \textbf{Density} \\ \hline\hline
Schools      & 160,525 & 0.007   \\ \hline
Parks        & 69,338  & 0.003   \\ \hline
Fast Food    & 25,069  & 0.001   \\ \hline
Post Offices & 21,319  & 0.0009  \\ \hline
Hospitals    & 11,417  & 0.0005  \\ \hline
Hotels       & 8,742   & 0.0004  \\ \hline
Universities & 3,954   & 0.0002  \\ \hline
Courthouses  & 2,161   & 0.00009 \\ \hline
\end{tabular}
\caption{Real-world POI sets for US road network}
\label{tab:pois}
\end{minipage}
\begin{minipage}{0.49\linewidth}
\centering
\begin{tabular}{|c||l|} \hline
\textbf{Variable} & \textbf{Values} \\ \hline\hline
$k$             & 1, 5, \textbf{10}, 25, 50              \\ \hline
$d$             & 0.1, 0.01, \textbf{0.001}, 0.0001   \\ \hline
A (\%)          & 1, 5, \textbf{15}, 50, 100             \\ \hline
$|Q|$           & 2, 4, \textbf{8}, 16, 32               \\ \hline
Agg. Func.      & \textit{sum}, \textbf{\textit{max}}    \\ \hline
$r$ (\% of $D$) & 0.5, 1, \textbf{2.5}, 5, 10            \\ \hline
\end{tabular}
\caption{Experimental variables (defaults in \textbf{bold})}
\label{tab:exp:variables}
\end{minipage}
\end{table}

\smallHead{Experimental Variables:} Table \ref{tab:exp:variables} lists all variables evaluated. For A$k$NN and $k$FN queries, we perform sensitivity analysis on query performance by varying the number of results $k$ and object set density $d$. For A$k$NN queries we additionally evaluate varying number of query vertices $|Q|$, aggregate functions, and variable $A$ representing how ``local'' the set of query locations are to each other. Similar to past studies~\cite{yiu2006aknn}, $A$ refers to a connected subgraph of the input graph $G$ with $A$\% of the total vertices $|V|$, from which $|Q|$ random query vertices are chosen. Note that we set a low default of $A{=}15$\% to ensure fairness, as our method is likely to perform better when query locations are farther apart with higher $A$. Query vertices are chosen uniformly at random for other queries as in previous studies~\cite{abey2016knn,zhong2015gtree}. For range queries, in addition to $d$, we evaluate radius $r$ expressed as a percentage of the graph diameter $D$ (which we approximate by the double-sweep method). Unless otherwise stated, the variables take on their default values indicated in bold in Table~\ref{tab:exp:variables}.

\smallHead{Performance Metrics:} We report average running time of $1{,}000$ sequentially executed queries. For sensitivity analysis, this consists of $20$ object sets and $50$ query vertices ($k$FN and range) or vertex sets (A$k$NN). In addition to query time, we measure heuristic efficiency using two machine-independent metrics, the total numbers of (a) network distances computed to objects; and (b) candidate objects retrieved from the data structure. The former measures how accurate the lower-bounds are at avoiding expensive network distance computations to objects that are not results and the latter measures how effective each data structure is at implicitly pruning non-result objects (i.e., by never even considering them individually).

\smallHead{Competitors:} 
We compare the performance of each query algorithm that utilizes COL-Trees against state-of-the-art heuristics for POI search including the hierarchical search-based \textit{Incremental Euclidean Restriction} (IER)~\cite{papadias2003ine} and expansion-based \textit{Network Voronoi Diagrams} (NVD). IER has been applied to solve A$k$NN queries~\cite{yiu2006aknn}. We significantly improve existing NVD-based A$k$NN search~\cite{zhu2010vn3aknn} by adapting Incremental Lower-Bound Restriction~\cite{abey2017knn} that uses ALT to avoid expensive network distance computation. Note that IER cannot solve $k$FN queries and NVD will always evaluate all objects (see Section~\ref{query:fkn}), so we implemented an optimized brute-force baseline, AUB-PHL, that evaluates all objects but filters candidates by ALT upper-bounds to reduce network distance computation. Comparisons for each query are listed in Table~\ref{tab:techniques}. Note that while COL-Trees are not designed with $k$NN queries in mind, we provide experimental results in~\ref{app:knn} as there are some cases in which they can be highly useful (e.g., where NVD indexing cost cannot be amortized).

\begin{table}[htb]
\centering
\begin{tabular}{|c||*{4}{c|}} \hline
\textbf{Technique} & \textbf{A$k$NN} & \textbf{$k$FN} & \textbf{Range} & \textbf{$k$NN} \\ \hline\hline
COLT-PHL (Ours)            & \ding{51} & \ding{51} & \ding{51} & \ding{51} \\ \hline
IER-PHL~\cite{yiu2006aknn} & \ding{51} & \ding{55} & \ding{51} & \ding{51} \\ \hline
NVD-PHL~\cite{abey2017knn} & \ding{51} & \ding{55} & \ding{51} & \ding{51} \\ \hline
AUB-PHL                    & \ding{55} & \ding{51} & \ding{55} & \ding{55} \\ \hline
\end{tabular}
\caption{Compared methods for each query}
\label{tab:techniques}
\end{table}

Each technique uses Pruned Highway Labeling (PHL)~\cite{kawata2014phl} for exact network distance computation, hence techniques are named IER-PHL, NVD-PHL, COLT-PHL, etc. Using the same network distance technique ensures a level playing field. Furthermore, as PHL is one of the fastest techniques and network distance is the most costly subroutine, it minimizes the bottleneck and better exposes overheads associated with each technique. NVD-based methods use the ALT index~\cite{goldberg2005alt} with $m{=}16$ random landmarks to compute LLBs and NVDs are represented as linear quadtrees~\cite{demiryurek2012nvd,abey2020kspin} for reduced space cost and improved memory locality. We also utilize the same ALT index in query algorithms based on COL-Tree to compute point-to-point lower- and upper-bounds, as it is effectively a SUL-Tree root node, for simplicity and a means of like-for-like comparison. We select COL-Tree parameters for each query by parameter testing detailed in Section~\ref{sec:exp:parametertesting}.

\subsection{A\texorpdfstring{$k$}{k}NN Queries}\label{exp:aknn}
Figure \ref{exp:aknn:rw} depicts A$k$NN query time on the real-world POI datasets, with the number of objects increasing from left to right. COLT-PHL significantly outperforms the other methods across the board, with up to 2 orders of magnitude improvement. COLT-PHL tends to improve more on larger POI sets, where it is more difficult to distinguish between objects. Next, we perform sensitivity analysis to delve deeper into this and other nuances of query performance.

\begin{figure}[t]
    \centering
    \includegraphics[width=0.65\linewidth]{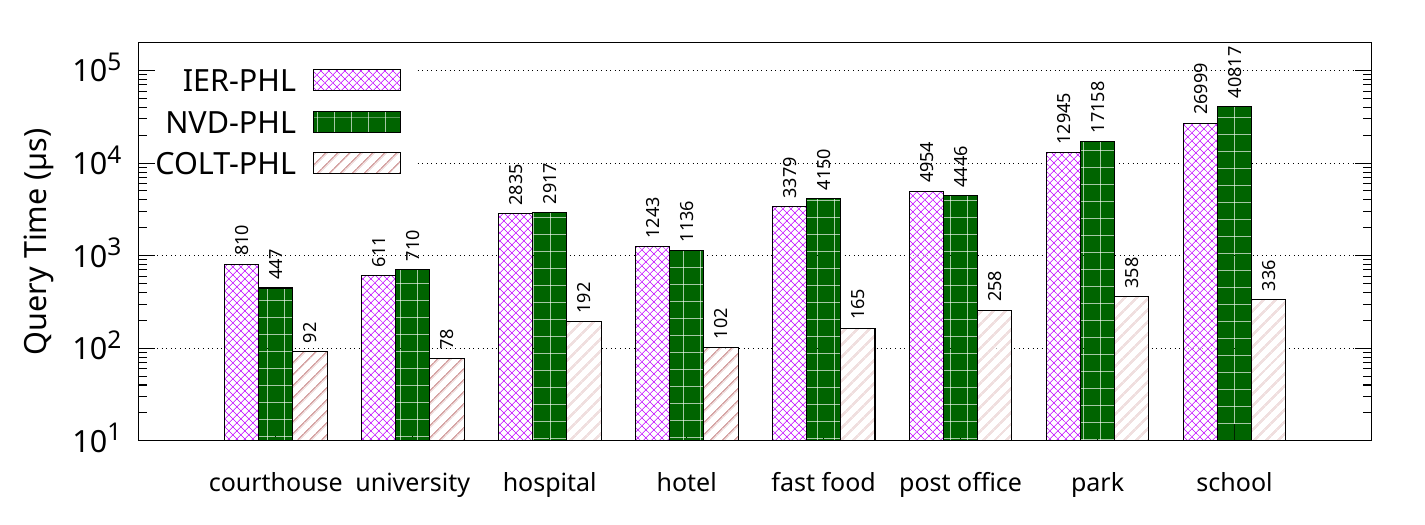}
    \caption{A$k$NN max query performance on different real-world POI sets 
    }
    \label{exp:aknn:rw}
\end{figure}

\begin{figure}[htbp]
    \centering
    \begin{subfigure}[b]{0.24\linewidth}
        \includegraphics[width=\linewidth]{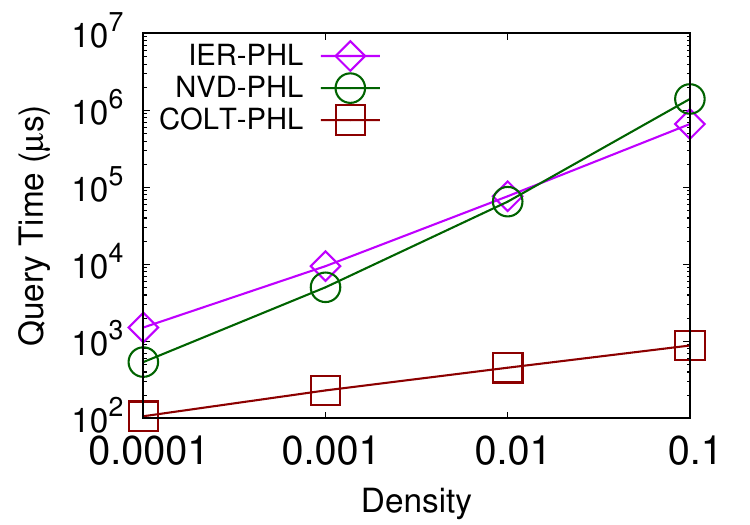}
        \caption{Varying $d$}
        \label{exp:aknn:max:time:varyd}
    \end{subfigure}
    \hfill
    \begin{subfigure}[b]{0.24\linewidth}
        \includegraphics[width=\linewidth]{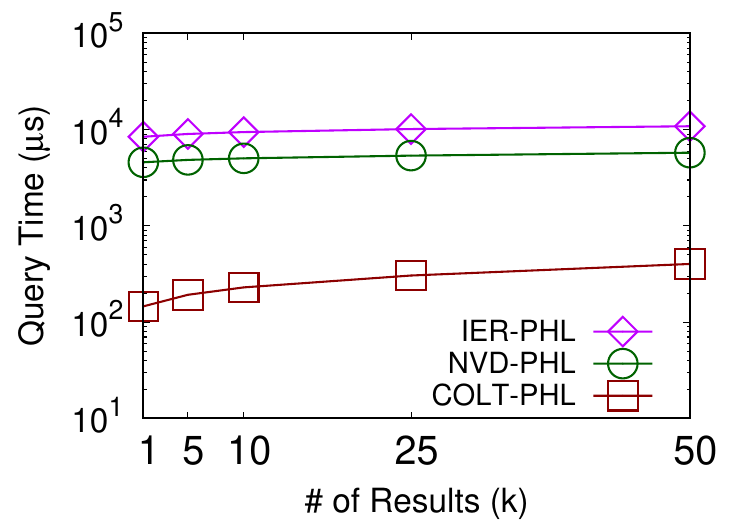}
        \caption{Varying $k$}
        \label{exp:aknn:max:time:varyk}
    \end{subfigure}
    \hfill
    \begin{subfigure}[b]{0.24\linewidth}
        \includegraphics[width=\linewidth]{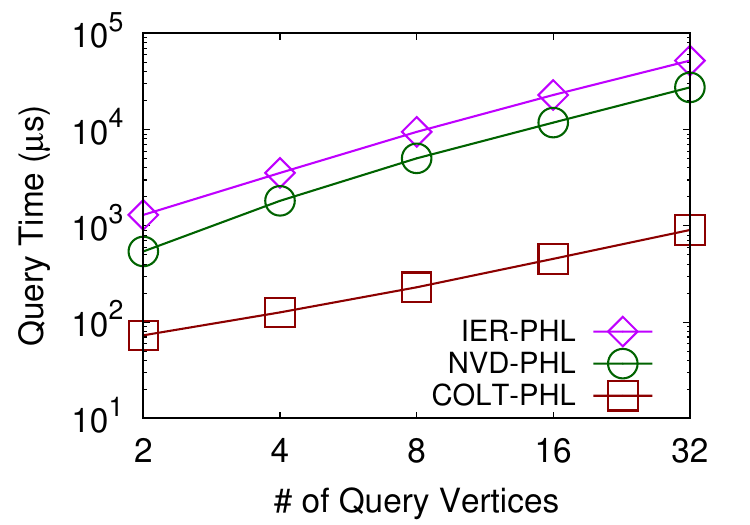}
        \caption{Varying $|Q|$}
        \label{exp:aknn:max:time:varyq}
    \end{subfigure}
    \hfill
    \begin{subfigure}[b]{0.24\linewidth}
        \includegraphics[width=\linewidth]{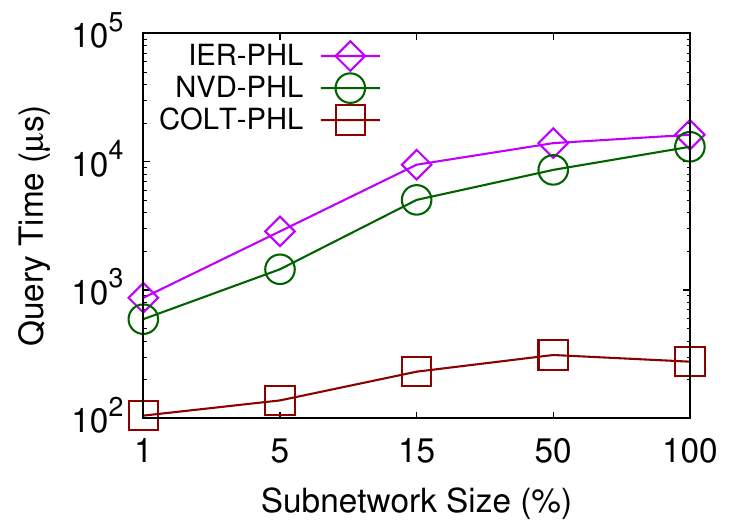}
        \caption{Varying $A$}
        \label{exp:aknn:max:time:varya}
    \end{subfigure}
	\caption{Performance for $max$ function 
	}
	\label{exp:aknn:max:time}
\end{figure}

\subsubsection{Sensitivity Analysis for A\texorpdfstring{$k$}{k}NN Queries}

\smallHead{Effect of $\boldsymbol{d}$:} The trend seen for real-world POIs is confirmed by increasing object density $d$ in Figure \ref{exp:aknn:max:time:varyd}. Both IER and NVD scale poorly with increasing $d$. With more objects, NVD-PHL must expand more adjacent objects to find a common candidate for the same query vertices, while IER-PHL finds it harder to distinguish objects using its less accurate Euclidean lower-bound. In contrast, COLT-PHL benefits from the tighter lower-bounds and hierarchical traversal afforded by the COL-Tree data structure, resulting in effective pruning of subgraphs and pinpointing likely candidates.  

\smallHead{Effect of $\boldsymbol{k}$:} COLT-PHL significantly outperforms the other methods for varying $k$ in Figure \ref{exp:aknn:max:time:varyk}. NVD-PHL and IER-PHL query times do not vary significantly compared to COLT-PHL (note the logarithmic scale). This suggests the same amount of work is done irrespective of $k$ and strongly implies the competing techniques cannot effectively identify good candidates and terminate quickly. For example, the fact that subsequent result objects are not much more costly to retrieve suggests the concurrent expansion in NVD-PHL is extremely large just to generate the first candidate.

\smallHead{Effect of $\boldsymbol{|Q|}$:} As the number of query vertices increases in Figure~\ref{exp:aknn:max:time:varyq}, additional lower-bounds and network distances to candidates will be computed; thus, query time increases for all methods. However, COLT-PHL scales better due to the convexity-preserving property only requiring a single binary search on the COL-Tree's $\ODL$s irrespective of the number of query vertices. On the other hand, NVDs require additional concurrent expansions as A$k$NN results are less likely to be close to any single query vertex. 

\smallHead{Effect of $\boldsymbol{A}$:} Recall that $A$ is the percentage of graph vertices in a subgraph from which query vertices $Q$ are chosen. With increasing $A$ query vertices become further apart, e.g., to represent a query by a logistics company placing depots nation-wide. For $A=1\%$, NVD-PHL is closer to COLT-PHL in Figure \ref{exp:aknn:max:time:varya}. This scenario resembles $k$NN queries where NVD-PHL is effective, as more query vertices share the same 1NN and concurrent expansions overlap earlier. IER does not benefit from this as its lower-bounds are still inaccurate. Queries become harder with increasing $A$ and COLT-PHL scales extremely well in that case.
\begin{figure}[htbp]
    \centering
    \begin{subfigure}[b]{0.24\linewidth}
        \includegraphics[width=\linewidth]{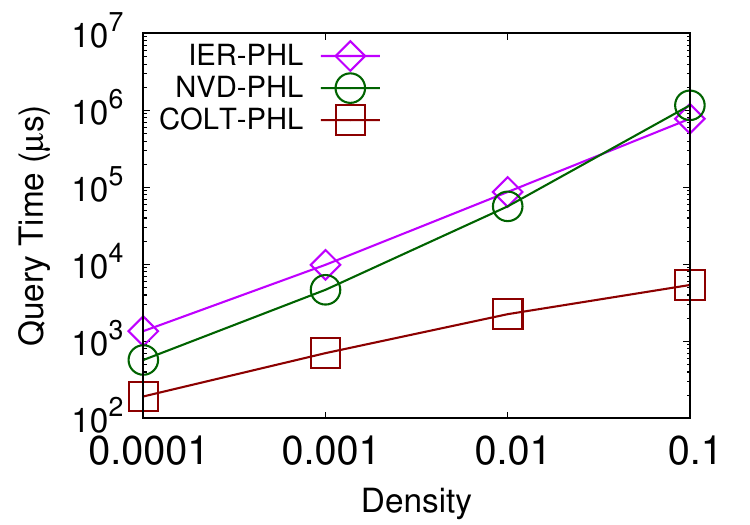}
        \caption{Varying $d$}
        \label{exp:aknn:sum:time:varyd}
    \end{subfigure}
    \hfill
    \begin{subfigure}[b]{0.24\linewidth}
        \includegraphics[width=\linewidth]{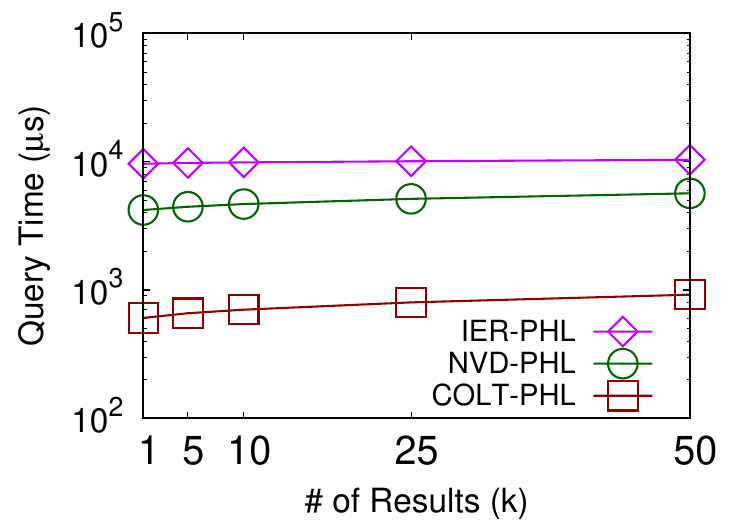}
        \caption{Varying $k$}
        \label{exp:aknn:sum:time:varyk}
    \end{subfigure}
    \hfill
    \begin{subfigure}[b]{0.24\linewidth}
        \includegraphics[width=\linewidth]{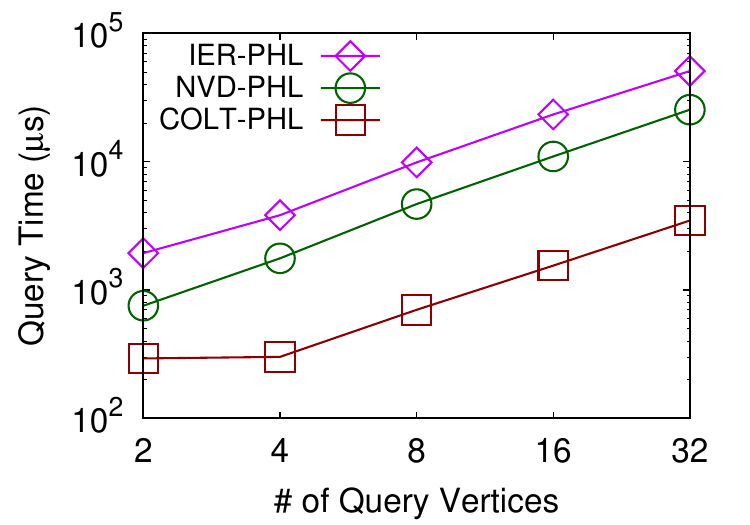}
        \caption{Varying $|Q|$}
        \label{exp:aknn:sum:time:varyq}
    \end{subfigure}
    \hfill
    \begin{subfigure}[b]{0.24\linewidth}
        \includegraphics[width=\linewidth]{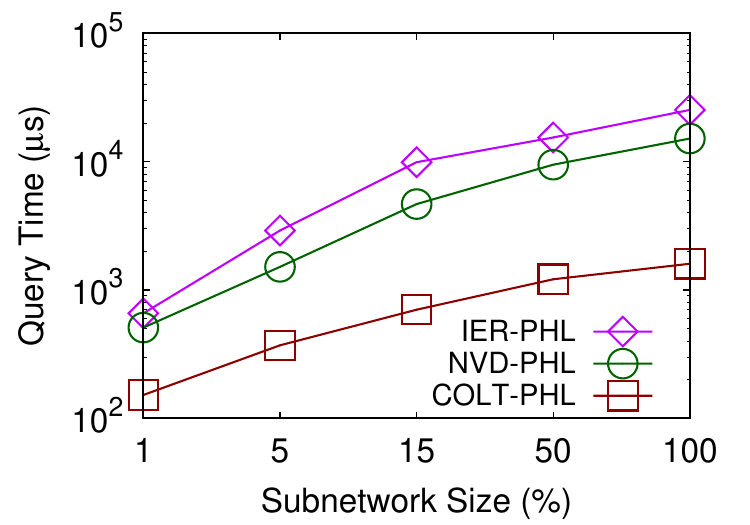}
        \caption{Varying $A$}
        \label{exp:aknn:sum:time:varya}
    \end{subfigure}
	\caption{Performance for $sum$ function 
	}
	\label{exp:aknn:sum:time}
\end{figure}

\smallHead{Effect of Aggregate Function:} We evaluate another popular aggregate function, $sum$, in Figure \ref{exp:aknn:sum:time}. Both IER-PHL and COLT-PHL have higher query times for $sum$ than $max$ because $sum$ also accumulates the error of the lower-bound. This effect is amplified by the hierarchical tree search used by both methods, explaining the relative improvement of NVD-PHL. However, tighter lower-bounds in the COL-Tree hierarchy make COLT-PHL more robust to this than IER. NVD-PHL efficiency is dominated by the concurrent expansions that must meet to generate candidates, which remains similar for $sum$. Regardless, COLT-PHL still significantly outperforms both methods. 

\smallHead{Heuristic Efficiency:} Figures \ref{exp:aknn:max:heur:varyqndist} and \ref{exp:aknn:sum:heur:varyqndist} show the heuristic performance for $max$ and $sum$, respectively. The significantly higher number of network distances computed to objects in Figures \ref{exp:aknn:max:heur:varyqndist:fh} and \ref{exp:aknn:sum:heur:varyqndist:fh} explains the poor query time of IER. However, NVD-PHL does not compute many more network distances than COLT-PHL. As both methods use landmark lower-bounds (LLBs), which are more accurate than IER's Euclidean-based lower-bound, they both avoid computing network distances. The poor query time of NVD-PHL can be explained by Figures \ref{exp:aknn:max:heur:varyqndist:lb} and \ref{exp:aknn:sum:heur:varyqndist:lb}, which show that NVD-PHL considers significantly more candidate objects. This is because the concurrent expansions from each query vertex must continue to retrieve candidates until they meet before even trying to compute network distances, which is not efficient for A$k$NN search. Note that we exclude IER from Figure~\ref{exp:aknn:max:heur:varyqndist:lb} (and other figures for candidates retrieved) because the metric is designed to differentiate the landmark-based techniques, NVD-PHL and COLT-PHL.

\begin{figure}[htbp]
    \begin{minipage}[t]{.49\linewidth}
        \centering
        \subfloat[Network Distances\label{exp:aknn:max:heur:varyqndist:fh}]{\includegraphics[width=0.49\linewidth]{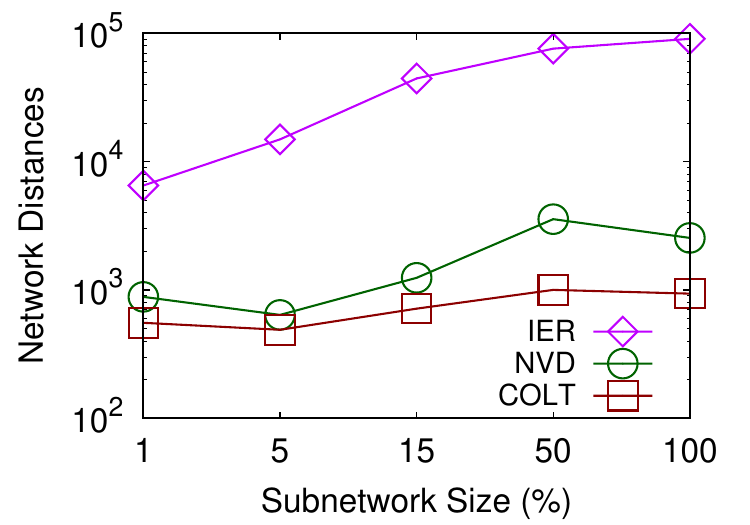}}
        \hfill
        \subfloat[Candidates Retrieved\label{exp:aknn:max:heur:varyqndist:lb}]{\includegraphics[width=0.49\linewidth]{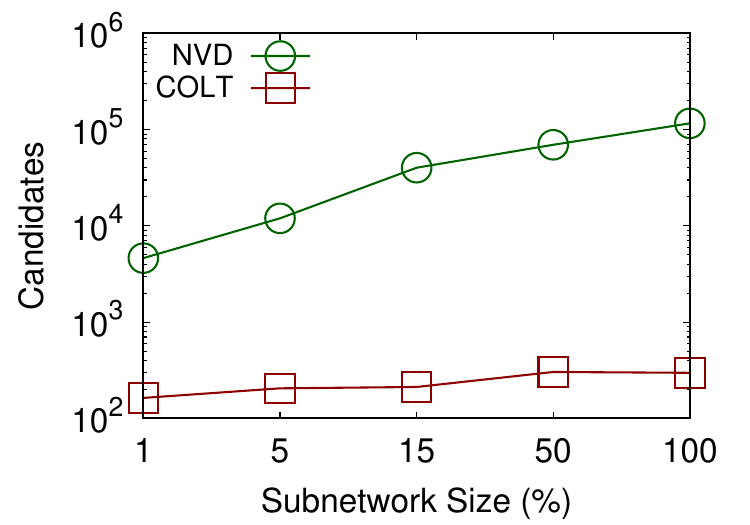}}
        \caption{Heuristic performance for A$k$NN with $max$ 
        }
        \label{exp:aknn:max:heur:varyqndist}
    \end{minipage}%
    \hfill
    \begin{minipage}[t]{.49\textwidth}
        \centering
        \subfloat[Network Distances\label{exp:aknn:sum:heur:varyqndist:fh}]{\includegraphics[width=0.49\linewidth]{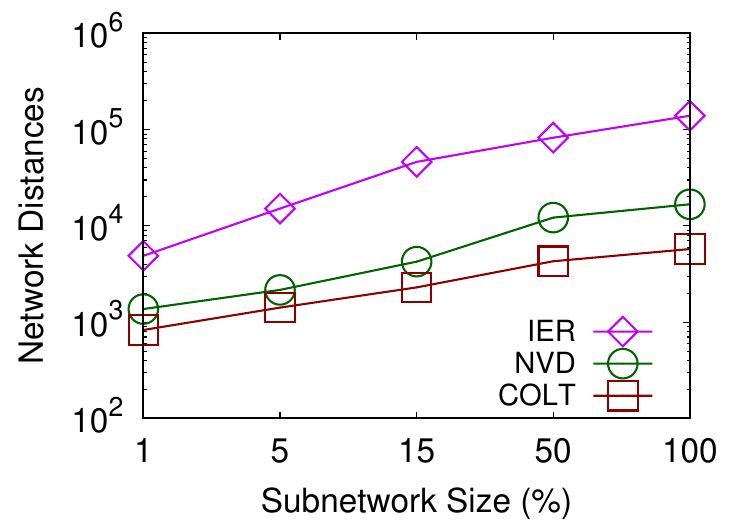}}
        \hfill
        \subfloat[Candidates Retrieved\label{exp:aknn:sum:heur:varyqndist:lb}]{\includegraphics[width=0.49\linewidth]{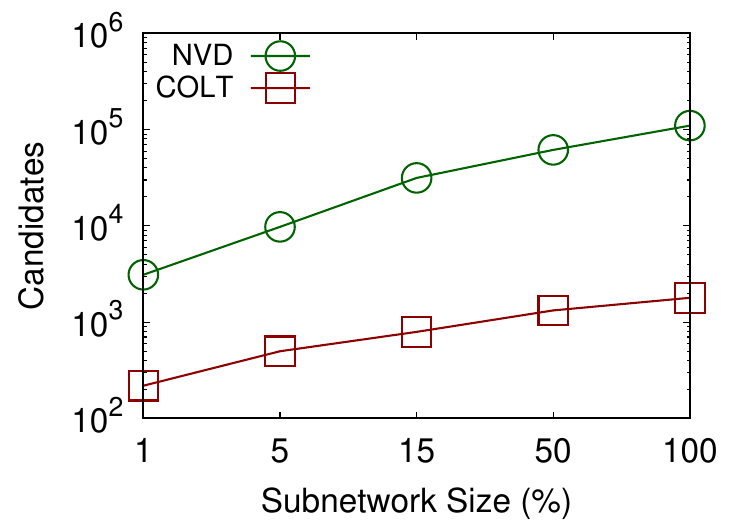}}
        \caption{Heuristic performance for A$k$NN with $sum$ 
        }
        \label{exp:aknn:sum:heur:varyqndist}
    \end{minipage}
\end{figure}

\subsection{\texorpdfstring{$k$}{k}FN Queries}\label{exp:kfn}
The effectiveness of COLT-PHL for $k$FN queries on real-world POI datasets is shown in Figure \ref{exp:fkn:rw}. Recall that IER-PHL cannot be used to solve $k$FN at all and concurrent expansion using NVDs can be no better than evaluating all objects; thus, we can only compare to the optimized brute-force baseline, AUB-PHL. Consequently, the COLT-PHL maximum speedup is even greater than that for A$k$NNs, being up to 187\texttimes~faster.

{
	\begin{figure}[htbp]
		\centering
		\includegraphics[width=0.65\linewidth]{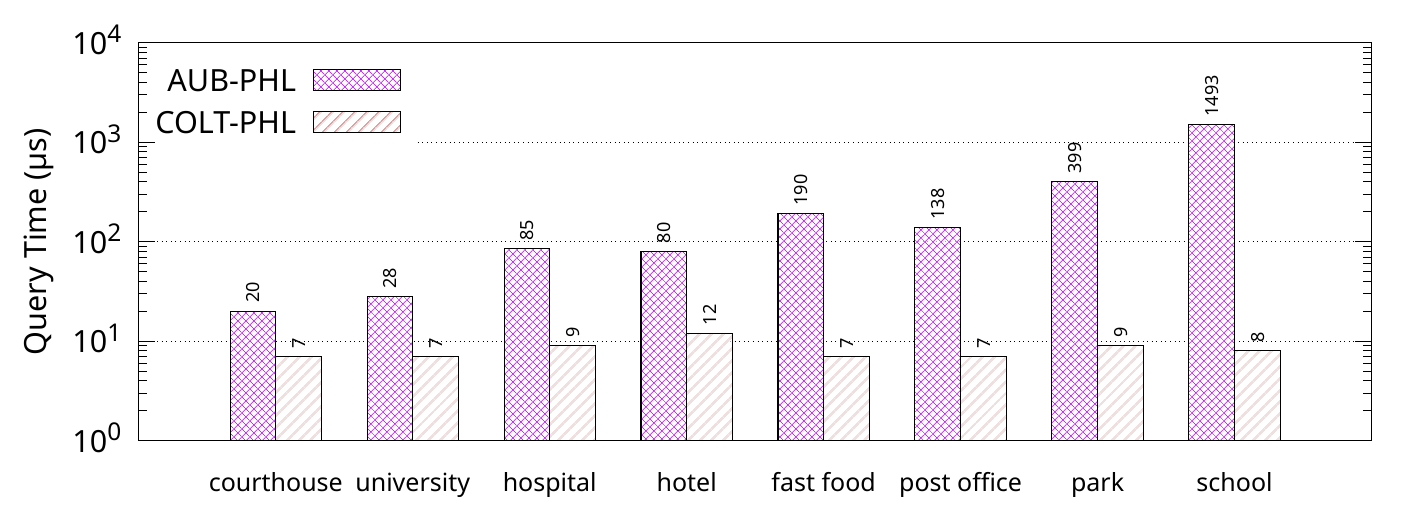}
		\caption{$k$FN query performance on different real-world POI sets  
        }
        \label{exp:fkn:rw}
	\end{figure}
}

\subsubsection{Sensitivity Analysis for \texorpdfstring{$k$}{k}FN Queries}

\smallHead{Effect of $\boldsymbol{d}$:} COLT-PHL remains largely unaffected by varying density $d$ in Figure~\ref{exp:fkn:varyd}. This indicates that the pruning ability is highly effective. With increasing density, the resulting COL-Tree will simply have a deeper hierarchy and effective bounds will prune the same branches and by extension also the extra levels. Since AUB-PHL must evaluate all objects, it is only competitive when object sets are extremely small.

\smallHead{Effect of $\boldsymbol{k}$:} Given that it must inspect all objects, AUB-PHL must essentially do the same amount of work irrespective of the number of results needed, as is evident with increasing $k$ in Figure~\ref{exp:fkn:varyk}. On the other hand, COLT-PHL scales more reasonably with increasing $k$. One noteworthy observation is the comparison between query times for $k$FN and its opposite, the $k$NN query. $k$FN queries using COLT-PHL are faster than $k$NN (Figure~\ref{exp:knn:varyk}) despite using the same data structure and an analogous search algorithm. The difference is that $k$FN traversal is by upper-bound instead of lower-bound, suggesting the former is more effective. 

\smallHead{Heuristic Efficiency:} The extremely low number of network distances in Figure~\ref{exp:fkn:heur:varyd:fh} suggests that COL-Tree upper-bounds are highly accurate, and mostly retrieve candidates that are one of the $k$FN results. Upper-bounds are known to be much closer to the true network distance~\cite{potamias2009distest} and being able to prune based on them shows that COL-Trees are highly suited for $k$FN queries. The small uptick for higher densities is due to objects being far closer together, which reduces the margin of error, affecting even accurate upper-bounds. This is confirmed in Figure~\ref{exp:fkn:heur:varyd:lb}, because the number of candidates does not similarly increase. 

\begin{figure}[htb]
    \centering
    \begin{subfigure}[b]{0.24\linewidth}
        \includegraphics[width=\linewidth]{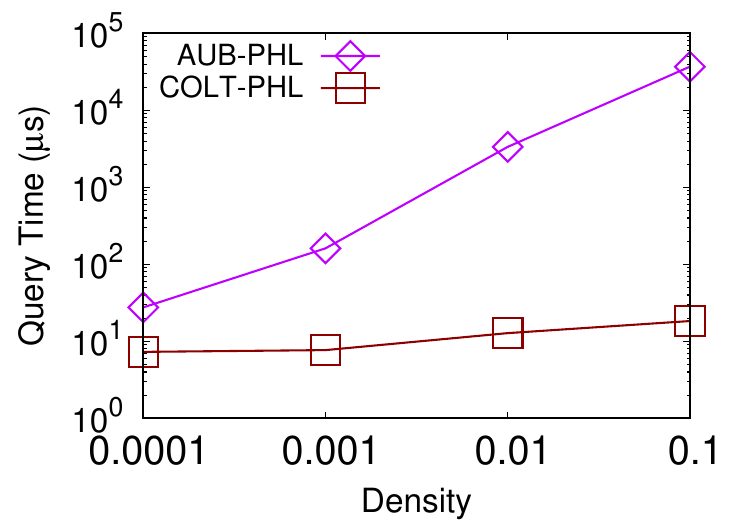}
        \caption{Varying $d$}
        \label{exp:fkn:varyd}
    \end{subfigure}
    \hfill
    \begin{subfigure}[b]{0.24\linewidth}
        \includegraphics[width=\linewidth]{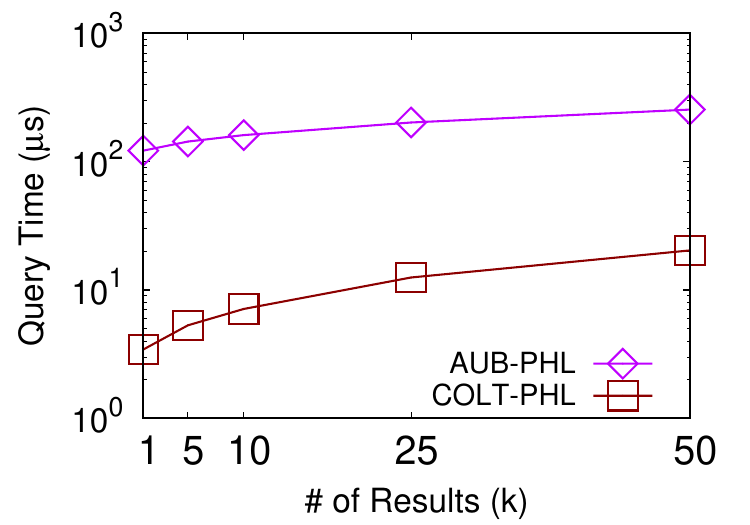}
        \caption{Varying $k$}
        \label{exp:fkn:varyk}
    \end{subfigure}
    \hfill
    \begin{subfigure}[b]{0.24\linewidth}
        \includegraphics[width=\linewidth]{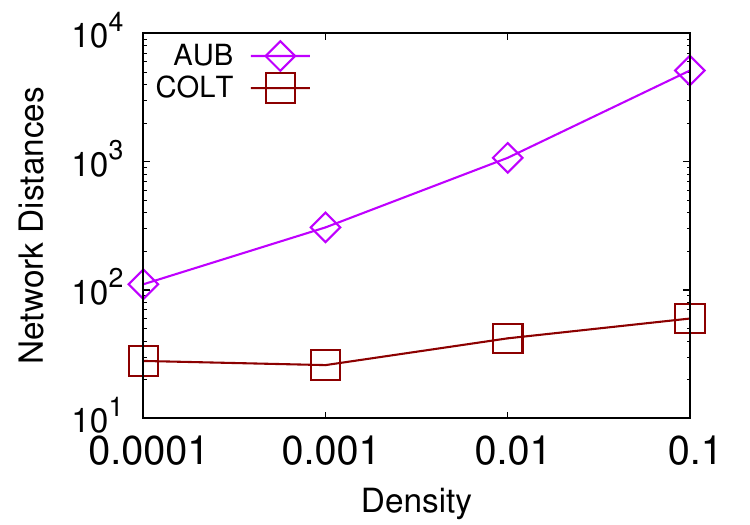}
        \caption{Network Distances}
        \label{exp:fkn:heur:varyd:fh}
    \end{subfigure}
    \hfill
    \begin{subfigure}[b]{0.24\linewidth}
        \includegraphics[width=\linewidth]{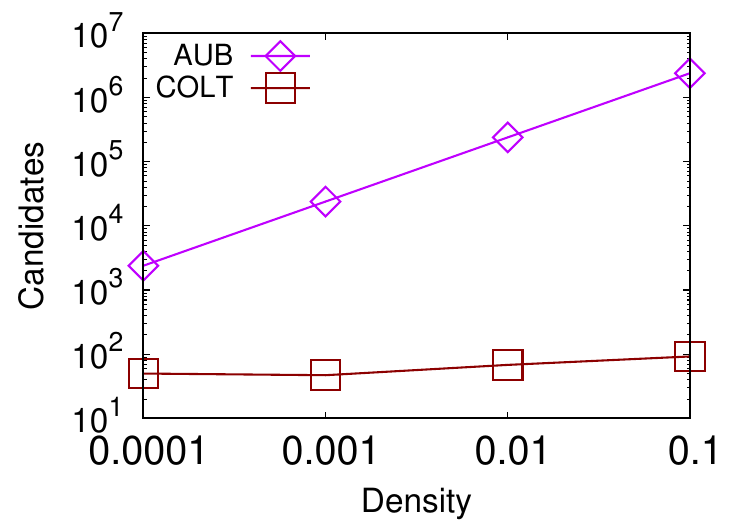}
        \caption{Candidates Retrieved}
        \label{exp:fkn:heur:varyd:lb}
    \end{subfigure}
	\caption{Performance for $k$FN queries 
	}
	\label{exp:fkn}
\end{figure}

\subsection{Range Queries}\label{sec:exp:range}
COLT-PHL range query times are surprisingly on par with NVD-PHL with the smallest object sets in Figure~\ref{exp:range:rw}, gradually improving with increasing size. On the largest set it achieves a significant $2.9$\texttimes~improvement over the next fastest competitor. Recall that the motivation for COL-Trees was based on POI search queries such as A$k$NN and $k$FN that are not aligned with the ``nearest first'' heuristics more suited for $k$NN and range queries. Intuitively, COLT-PHL should not work so well for range queries and we investigate why next. 

{
	\begin{figure}[htb]
		\centering
		\includegraphics[width=0.65\linewidth]{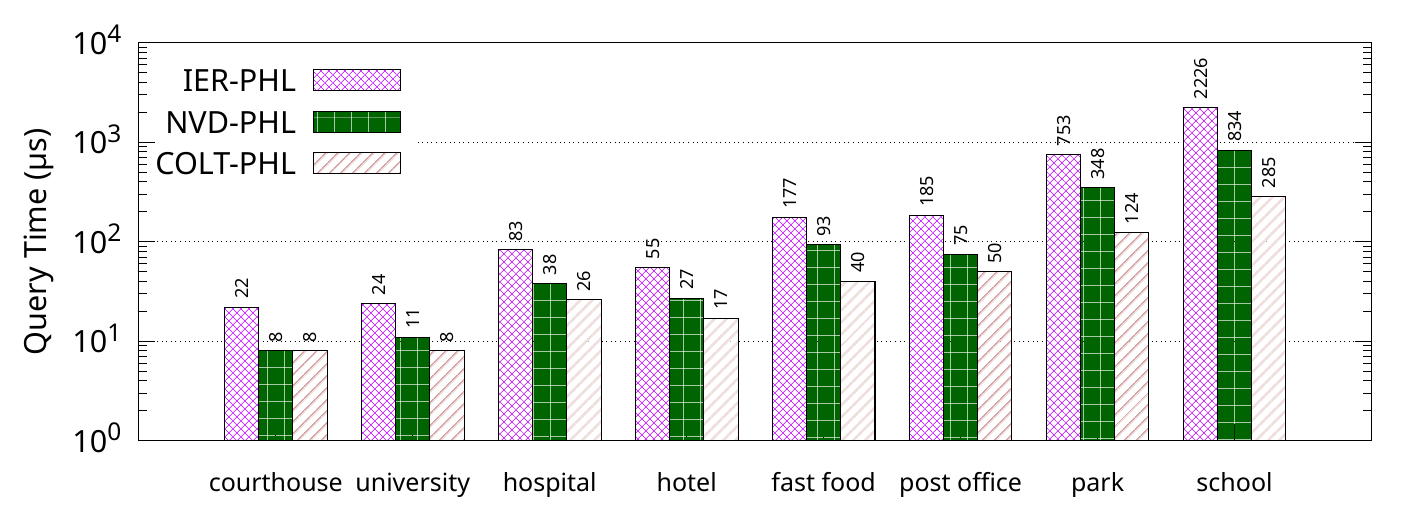}
		\caption{Range query performance on different real-world POI sets  
        }
        \label{exp:range:rw}
	\end{figure}
}

\subsubsection{Sensitivity Analysis for Range Queries}
For range search sensitivity analysis, we increase the default density to $d{=}0.01$ to make queries more challenging and retrieve more objects as results (our varying $d$ experiment still evaluates all densities).

\smallHead{Effect of $\boldsymbol{d}$:} 
The query time of COLT-PHL improves relative to NVD-PHL with greater density $d$ in Figure~\ref{exp:range:varyd}. With increasing density, more objects fall into the same range, and COLT-PHL scales better as more objects can simply be identified as results or pruned using only the COL-Tree node lower- and upper-bounds to entire groups of objects (i.e., without testing individual objects from leaf node $\ODL$s). COLT-PHL achieves significant improvement on higher densities. We must also take into account that building an NVD index is more than $1{,}500$\texttimes~slower than a COL-Tree (Table~\ref{tab:index}) and must be paid for every object set, making methods based on COL-Trees clearly preferable. 

\smallHead{Effect of $\boldsymbol{r}$:} 
COLT-PHL also improves relative to NVD-PHL with increasing radius $r$ in Figure~\ref{exp:range:varyr}. Similar to the case with increasing $d$, increasing $r$ leads to more objects being evaluated and returned as results, which COL-Tree traversal can handle well. The scalability of COLT-PHL also likely benefits from two advantages with increasing $r$. First, range search within COL-Trees can be implemented by a stack while NVD-PHL must use a priority queue (such as a binary heap). Intuitively, similar to increasing a circle's radius, a modest increase in range can dramatically increase the area covered and the potential number of objects within range. While not directly comparable, stacks will scale better for the corresponding significant increase in data structure operations than priority queues. Second, as already discussed in relation to $k$FN queries, we know that upper-bound pruning is extremely effective and increasing $r$ means more objects fall within it and can be immediately confirmed as results by upper-bounds to entire groups of objects.

\smallHead{Heuristic Efficiency:} 
Range queries do not strictly require computation of network distance to each object, as long as we correctly determine they are within range $r$. Figure~\ref{exp:range:heur:varyd:fh} shows that COLT-PHL is able to compute fewer network distances to objects with increasing $d$, suggesting that its upper-bounds are indeed working as more objects fall within the same range, as described above. Moreover, COLT-PHL is more agreeable with increasing $d$ compared to queries that retrieve $k$ results, since the order is not relevant for range queries. Figure~\ref{exp:range:heur:varyd:lb} shows that COLT-PHL scales better with increasing $d$ on the number of candidates considered (i.e., objects that can neither be implicitly pruned nor immediately reported as results), eventually overtaking NVD-PHL. Such behavior would penalize NVD-PHL's priority queue more severely. 

\begin{figure}[htb]
    \centering
    \begin{subfigure}[b]{0.24\linewidth}
        \includegraphics[width=\linewidth]{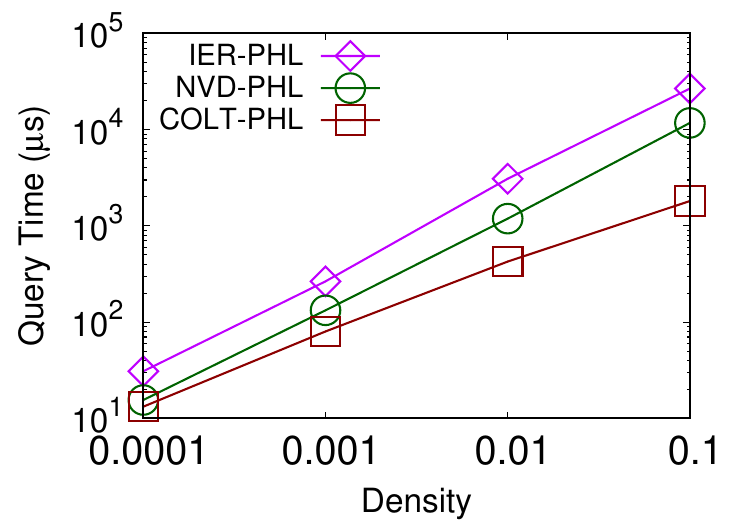}
        \caption{Varying $d$}
        \label{exp:range:varyd}
    \end{subfigure}
    \hfill
    \begin{subfigure}[b]{0.24\linewidth}
        \includegraphics[width=\linewidth]{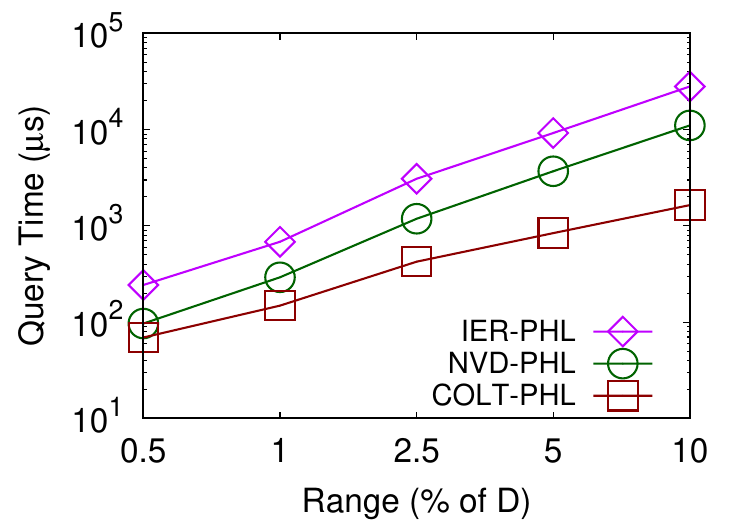}
        \caption{Varying $r$}
        \label{exp:range:varyr}
    \end{subfigure}
    \hfill
    \begin{subfigure}[b]{0.24\linewidth}
        \includegraphics[width=\linewidth]{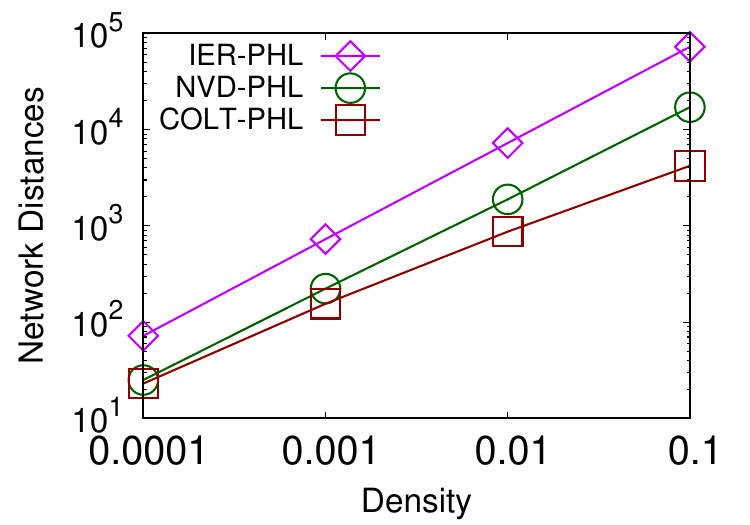}
        \caption{Network Distances}
        \label{exp:range:heur:varyd:fh}
    \end{subfigure}
    \hfill
    \begin{subfigure}[b]{0.24\linewidth}
        \includegraphics[width=\linewidth]{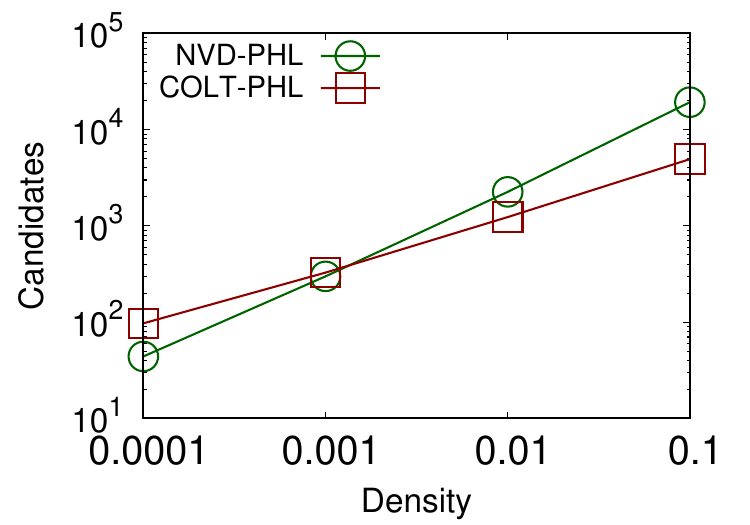}
        \caption{Candidates Retrieved}
        \label{exp:range:heur:varyd:lb}
    \end{subfigure}
	\caption{Performance for range queries 
	}
	\label{exp:range}
\end{figure}

\subsection{Pre-Processing Performance}\label{exp:preprocess}

Table~\ref{tab:index} details pre-processing costs in terms of time and space. Naturally, graph indexes are more costly (Table~\ref{tab:index:graph}). SUL-Tree indexing time represents a small overhead over the other landmark-based index, ALT, aligning with our theoretical analysis in Section~\ref{sec:complexity}. However, it is still very low and significantly lower than PHL, which possesses one of the fastest pre-processing times for high-performance road network indexes with travel-time edge-weights~\cite{kawata2014phl}. SUL-Tree indexes are even smaller than ALT, as fewer (but more localized) landmarks are needed to support efficient search, offsetting the $O(\log|V|)$ overhead over ALT. 

SUL-Tree and ALT are slightly faster on the W-US graph with distance edge-weights. This is likely because the Dijkstra search can settle vertices more easily due to there being fewer subpaths that initially appear longer but cannot be pruned as is the case for travel time edge-weights (e.g., paths using highways with higher speed limits may be shorter overall but may involve longer paths to reach them initially). PHL degrades significantly on distance edge-weights as it relies on the same highway-like hierarchies to prune the label construction phase, which are less prevalent in distance-based graphs. This was previously reported by~\citet{farhan2023h2cl}, and an alternative network distance computation such as their H2CL method may be more suitable for distance edge-weights. As a benefit of our approach, any network distance computation method can be used during COL-Tree query processing. For example, a practitioner with limited system memory may choose a network distance index with a smaller memory footprint, such as Contraction Hierarchies~\cite{geisberger2008ch}, but such concerns are unlikely to affect SUL-Trees and COL-Trees given their small index sizes.

The cost of object indexes, other than NVDs, is extremely small (Table~\ref{tab:index:object}). The pre-processing time is highly relevant to determine how easily a technique can be used in real-time settings. For example, in a ride-hailing application almost all objects move simultaneously, necessitating frequent rebuilds of the object index~\cite{abey2021ikm}. COL-Trees have comparable construction time to R-trees but offer significantly better query performance, making them ideally suited to such applications. While the indexing costs (in kilobytes) seem trivial, there may be many object sets for different types of POIs like those in Table~\ref{tab:pois} or even thousands of keywords~\cite{abey2020kspin}. NVD involves significant pre-processing cost, with its construction time and index size depending on the entire input graph, making it costly even on small object sets.

Note that for SUL-Trees, in theory, we require $\alpha{\leq}\lambda$ in case all vertices in a SUL-Tree leaf node are objects. However, in practice, we find that we can safely set $\alpha{=}4{\times}\lambda$ without affecting COL-Tree construction at all, even for the largest object sets ($d{=}0.1$). This reduces SUL-Tree construction time by avoiding computing $\SDL$s that are never used. If a case occurs where smaller $\alpha$ is needed, it can be simply handled on the fly via $\SDL$ computation on the corresponding (unprocessed) SUL-Tree node. Also note that we report SUL-Tree index costs for \policy{Minmax} landmarks as they involve the highest landmark selection cost and $m{=}1$ as this is the most commonly used number of landmarks (Table~\ref{tab:parameter_testing}). 

\begin{table}[htb]
\centering
\begin{subtable}{\linewidth}
\centering
\begin{tabularx}{0.8\linewidth}{|c||*{6}{R|}} \hline
\multirow{2}{*}{\textbf{Dataset}} & \multicolumn{3}{c|}{\textbf{Time (s)}} & \multicolumn{3}{c|}{\textbf{Space}} \\ \cline{2-7}
 & \textbf{SUL-T} & \textbf{ALT} & \textbf{PHL} & \textbf{SUL-T} & \textbf{ALT} & \textbf{PHL} \\ \hline\hline
US           & 128.1  & 25.1 & 722    & 577~MB   & 1.43~GB & 15.8~GB   \\ \hline
W-US (Time)  & 30.3   & 6.3  & 94     & 146~MB   & 382~MB  & 2.8~GB    \\ \hline
W-US (Dist.) & 27.3   & 6.0  & 653    & 146~MB   & 382~MB  & 10.8~GB   \\ \hline
NACA0015     & 9.6    & 2.2  & 4,083  & 31~MB    & 63~MB   & 13.9~GB   \\ \hline
SRTMGL3      & 12.6   & 3.6  & 8,470  & 45~MB    & 88~MB   & 23.0~GB   \\ \hline
\end{tabularx}
\caption{Graph indexes. ALT uses $m{=}16$. SUL-Trees use $m{=}1$ \policy{Minmax} landmarks, and branching factor $b$ and $\alpha$ are listed in Table~\ref{tab:datasets}.}
\label{tab:index:graph}
\end{subtable}

\vspace{2mm}

\begin{subtable}{\linewidth}
\centering
\begin{tabularx}{0.8\linewidth}{|c||*{6}{R|}} \hline
\multirow{2}{*}{\textbf{Dataset}} & \multicolumn{3}{c|}{\textbf{Time (ms)}} & \multicolumn{3}{c|}{\textbf{Space}} \\ \cline{2-7}
 & \textbf{COL-T} & \textbf{R-tree} & \textbf{NVD} & \textbf{COL-T} & \textbf{R-tree} & \textbf{NVD} \\ \hline\hline
US           & 3.3 & 2.0 & 5,030 & 310~KB  & 776~KB & 28.1~MB \\ \hline
W-US (Time)  & 0.7 & 0.6 & 1,123 & 78~KB  & 228~KB  & 6.9~MB \\ \hline
W-US (Dist.) & 1.0 & 0.6 & 1,087 & 78~KB  & 228~KB  & 6.8~MB \\ \hline
NACA0015     & 0.1 & 0.1 & 390 & 20~KB & 42~KB & 1.8~MB \\ \hline
SRTMGL3      & 0.2 & 0.1 & 348 & 30~KB  & 59~KB  & 3.5~MB \\ \hline
\end{tabularx}
\caption{Object indexes for uniform object sets with $d{=}0.001$. R-tree branching factor is $32$. COL-Trees use $m{=}1$ \policy{Minmax} landmarks, and branching factor $b$ and $\lambda$ are listed in Table~\ref{tab:datasets}.}
\label{tab:index:object}
\end{subtable}
\caption{Pre-processing costs of graph and object indexes}
\label{tab:index}
\end{table}

\subsubsection{Pre-Processing Improvements}\label{sec:exp:improvements}

We evaluate the SUL-Tree pre-processing improvements proposed in Section~\ref{sec:optimizations} with results detailed in Table~\ref{tab:improvements}. The first three rows report the impact of subgraph Dijkstra search improvements on populating $\SDL$s for $m{=}4$ random landmarks (i.e., not including other SUL-Tree construction tasks, such as partitioning). Using the new ordering to eliminate the use of hash-tables during Dijkstra alone results in a 34\% reduction in $\SDL$ population time. The next two rows show the impact of the restricted Dijkstra subgraph search using border set lower-bounds (Section~\ref{sec:optimizations:search}) in terms of reduction in time and $\gamma$. Recall that $\gamma$ is the average factor by which subgraph Dijkstra search exceeds the number of subgraph vertices. The $\SDL$ population time after the removal of hash-tables is reduced by a further 28.5\% using the new variant of Dijkstra search. The improvement of $\gamma$ is even more significant at 43.1\%. The running time reduction is lower than the reduction in $\gamma$ due to fixed costs like initialization being unaffected and the overhead added by computing lower-bounds. Note that we choose higher $m$ and random landmarks to reduce the impact of any specific landmark. The last two rows report the impact of the subgraph vertex ordering (Section~\ref{sec:optimizations:ordering}) on index size. Without the subgraph vertex ordering, there is a fixed 1,119~MB cost to maintain vertex node positions on the US road network; thus, the savings are relative to the number of landmarks, ranging from 32.8\% at $m{=}4$ to 66\% at the $m{=}1$ setting used by a majority of the experiments.  

\begin{table}[htb]
\centering
\begin{tabular}{|c||*{3}{r|}} \hline
\textbf{Optimization} & \textbf{Before} & \textbf{After} & \textbf{\% Improv.} \\ \hline\hline
No Hash-table Dijk. & 218~s                 & 144~s                  & 34.0\% \\ \hline
Border Set Dijk. (s)  & 144~s  & 103~s  & 28.5\% \\ \hline
Border Set Dijk. ($\gamma$) & 6.5 & 3.7 & 43.1\% \\ \hline
Subgraph Ordering & 3,410~MB              & 2,291~MB            & 32.8\% \\ \hline
Subgraph Ordering ($m{=}1$) & 1,696~MB              & 577~MB              & 66.0\% \\ \hline
\end{tabular}
\caption{Impact of SUL-Tree pre-processing optimizations described in Section~\ref{sec:optimizations} on the US road network ($m{=}4$ random landmarks unless stated otherwise). Times are for $\SDL$ computation only, so they exclude other costs included by timing for Table~\ref{tab:index:graph}, such as partitioning.}
\label{tab:improvements}
\end{table}

\subsection{Index Parameter Testing}\label{sec:exp:parametertesting}

We investigate query performance while varying the structure and properties of COL-Tree on four variables: the branching factor $b$, the maximum leaf node $\ODL$ size $\lambda$, the number of landmarks per node $m$, and the landmark selection policy. The tested values for each are listed in Table~\ref{tab:parameter_testing}. 

\smallHeadIndent{A Note on Parallelization:} While SUL-Tree construction is fast (taking no more than a few minutes), we use parallelization to make exhaustive parameter testing on a large number of SUL-Trees feasible. The SUL-Tree node $\SDL$ computation was easily parallelized using OpenMP, but we leave the full exploration of parallelization to future work. Note that we only use multiple cores during parameter testing for SUL-Tree construction (any reported time is single-threaded). We next present query performance of COLT-PHL over the different COL-Tree configurations.

\begin{table}[htb]
\centering
\begin{tabular}{|c||*{4}{c|}} \hline
\textbf{Variable} & \textbf{Tested Values} & \textbf{A$k$NN} & \textbf{$k$FN} & \textbf{Range} \\ \hline\hline
$b$           & $2$, $4$, $8$, $16$              & $8$    & $8$   & $8$    \\ \hline
$\lambda$     & $128$, $256$, $512$, $1024$      & $256$  & $256$ & $256$  \\ \hline
$m$           & $1$, $2$, $4$, $8$               & $1$    & $8$   & $1$    \\ \hline
Landmark Type & \policy{Random}, \policy{Minmax}, \policy{Farthest}, \policy{Slice}  & \policy{Minmax} & \policy{Random} & \policy{Minmax} \\ \hline
\end{tabular}
\caption{Tested values and finalized COL-Tree parameters on the US dataset for each query}
\label{tab:parameter_testing}
\end{table}

Figures \ref{exp:parametertest:b} to \ref{exp:parametertest:lmktype} illustrate query performance with varying COL-Tree parameters for each query type. While the parameter indicated in each figure is varied, the remaining parameters take on the default values for that query indicated in Table~\ref{tab:parameter_testing} (see Table~\ref{tab:exp:variables} for other defaults). The overarching observation is that COL-Tree performance is not beholden to extensive parameter tuning, and will still outperform competing methods in most settings with naive parameter values (e.g., a few random landmarks), and the improvement will still be significant in the case of A$k$NN and $k$FN search. However, useful performance improvements can be extracted by carefully choosing parameters, especially for dense object sets and in the case of range queries where the margin over competing methods is smallest. Note that figures display relative query performance to the default setting in Table~\ref{tab:parameter_testing} to better show the differences.

\begin{figure}[htb]
    \centering
    \includegraphics[width=\linewidth]{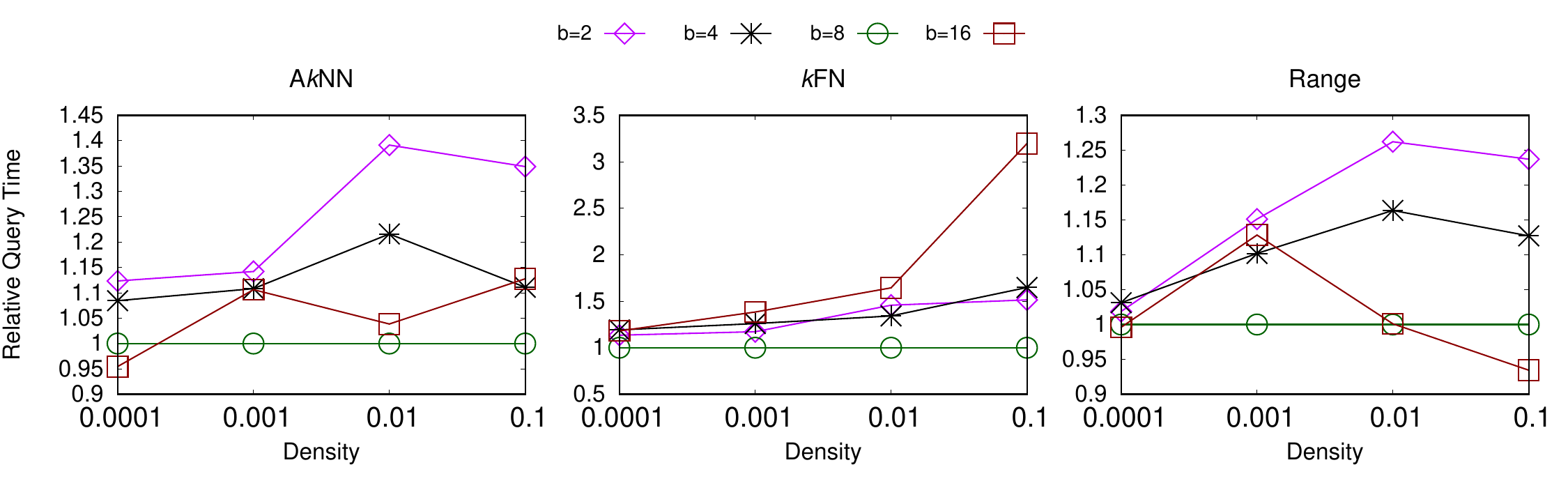}
	\caption{Query performance with varying branching factor $b$ (relative to each query's default in Table~\ref{tab:parameter_testing})}
	\label{exp:parametertest:b}
\end{figure}

\smallHead{Effect of $\boldsymbol{b}$:} The branching factor results in Figure~\ref{exp:parametertest:b} compare the trade-off between more branches and a shallower tree or fewer branches and a deeper tree. With few exceptions, $b{=}8$ is a clear winner across all query types. At low density, all trees are shallow simply from there being fewer objects to index, and all $b$ tend to perform similarly. Interestingly, $b{=}4$ and $b{=}16$ appear to diverge and converge with increasing density for A$k$NN queries. While $b$ determines the depth of the tree, the resulting number of leaf nodes can be similar (i.e., two trees with different depths can have the same number of leaf nodes). Notably, $b{=}16$ shows improvement with increasing density for range queries, while $k$FN shows the opposite. Generally, increasing density fills leaves and allows increasing $b$ to fill a deep hierarchy. Range queries can exploit the resulting finer-grained bounding allowed by more branches. For $k$FN, the pruning distances are large due to the farthest metric and a finer grain only creates more work to eliminate the same objects.

\begin{figure}[htb]
    \centering
    \includegraphics[width=\linewidth]{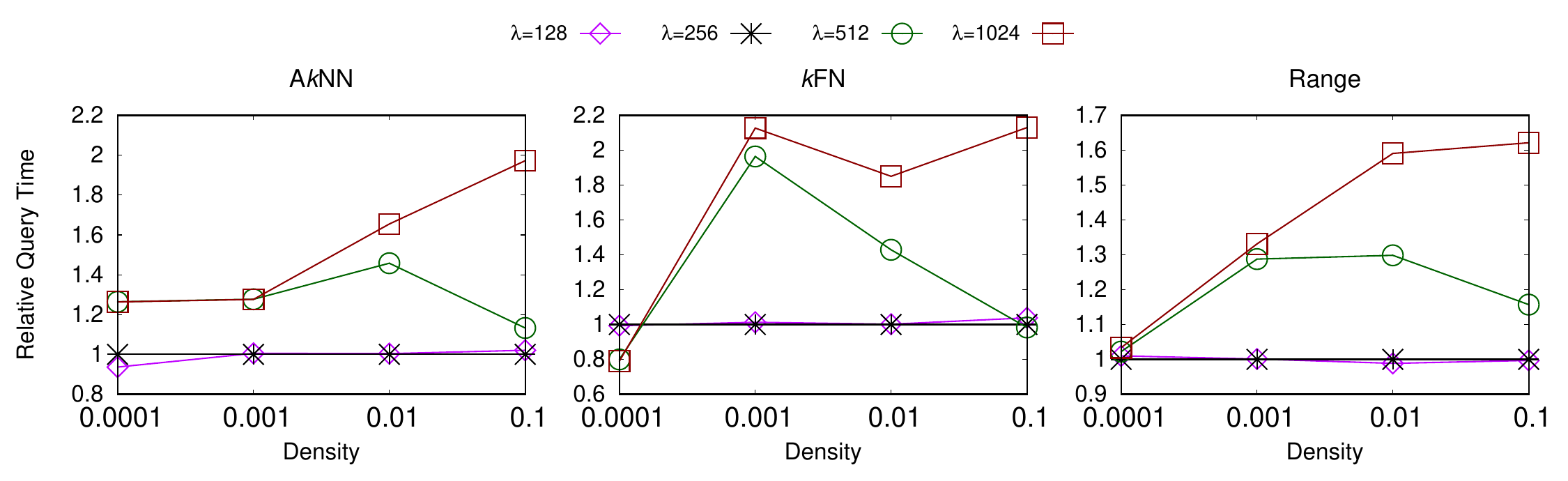}
	\caption{Query performance with varying object distance list size $\lambda$ (relative to each query's default in Table~\ref{tab:parameter_testing})}
	\label{exp:parametertest:lambda}
\end{figure}

\smallHead{Effect of $\boldsymbol{\lambda}$:} Figure~\ref{exp:parametertest:lambda} shows a clear preference for $\lambda{=}256$ or lower. With a larger $\lambda$, increasing object density results in more objects in each $\ODL$. Naturally, this makes the $\ODL$ bounds less effective as the same landmark must differentiate more objects. Interestingly, $\lambda{=}512$ degrades initially but begins to improve after a certain density, which is likely explained by the impact of branching factor $b$. Subgraphs become fuller with increasing density, which degrades performance. But if density increases enough to cause a leaf node to exceed $\lambda$ objects, the node is split and the tree depth will increase. The $\ODL$s of the new leaves will have on average $\frac{1}{b}$th the objects of the previous leaf, which improves the bounds and associated query performance. These observations support our intuition behind COL-Trees and the use of $\lambda$. The performance of $128$ and $256$ is very similar, indicating the point at which the benefit from this stops, and we choose the latter due to the smaller pre-processing cost that results from the larger SUL-Tree $\alpha$ allowed by larger COL-Tree $\lambda$.

\begin{figure}[htb]
    \centering
    \includegraphics[width=\linewidth]{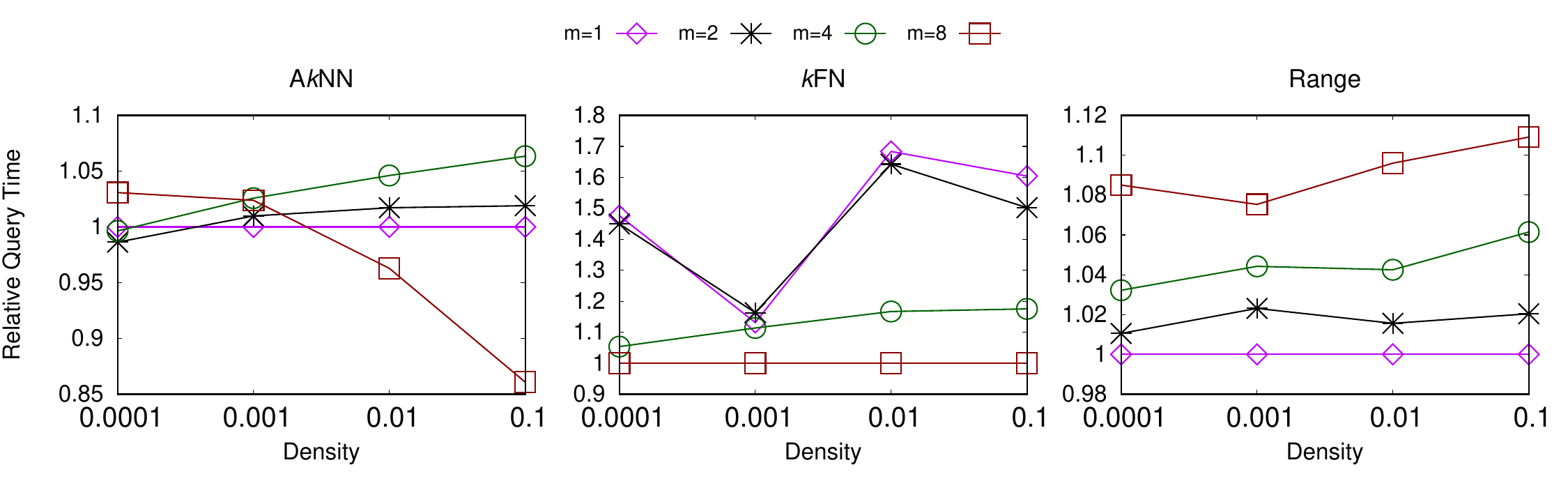}
	\caption{Query performance with varying number of landmarks $m$ (relative to each query's default in Table~\ref{tab:parameter_testing})}
	\label{exp:parametertest:m}
\end{figure}

\smallHead{Effect of $\boldsymbol{m}$:} For A$k$NN and range search, performance generally degrades across densities with increasing numbers of landmarks in Figure~\ref{exp:parametertest:m}. Recall that we must choose the tightest lower-bound over all landmarks as per Eq.~(\ref{eq:llbmax}). Hence, with increasing $m$, the cost of computing lower-bounds increases, which must be compensated for by increased lower-bound accuracy. This is a less viable trade-off for A$k$NN queries, because even if we have more landmarks, it is still difficult to find one that is suitable for all query vertices. Thus, we simply take longer to compute similar quality lower-bounds with increasing $m$. Range queries require both lower- and upper-bounds, and similar to the multiple query vertices in A$k$NN search, increasing the number of landmarks does not necessarily find a better trade-off. The one exception is $m{=}8$ for A$k$NN, which we provide an explanation for in the summary section.

\begin{figure}[htb]
    \centering
    \includegraphics[width=\linewidth]{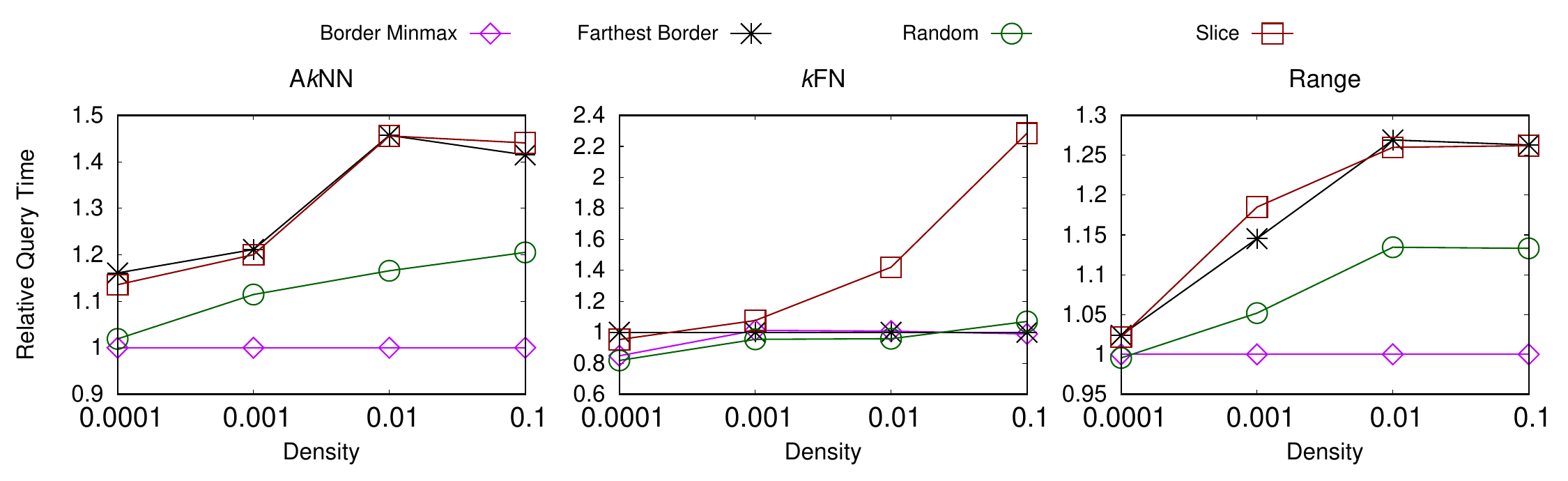}
	\caption{Query performance with varying landmark type (relative to each query's default in Table~\ref{tab:parameter_testing})}
	\label{exp:parametertest:lmktype}
\end{figure}

\smallHead{Effect of Landmark Selection Policy:} The best selection policy largely depends on the query type. \policy{Minmax} is clearly the most effective for A$k$NN and range queries, and the intuition is straightforward. \policy{Minmax} attempts to find an approximate subgraph center, the vertex that minimizes the maximum distance to any subgraph vertex. Both A$k$NN and range search benefit from a narrow spread of distances, i.e., reducing the gap between $M^-$ and $M^+$, which a more central landmark does. For A$k$NN search this results in better lower-bounds for multiple query vertices and for range search it reduces the size of the sublist. The superiority of $m{=}1$ for these queries is also closely related to \policy{Minmax}, since it is less likely that there are multiple central vertices that are not close to each other. For $k$FN, the selection policy has little effect as even a 20\% difference would at best reduce query time by a few microseconds and we choose \policy{Random} as the default for simplicity and faster pre-processing. A landmark gives a tight upper-bound only if it lies on the paths from the query into the subgraph and is also close to all of the subgraph's objects; border landmarks favor the former and interior landmarks the latter, so no policy can demonstrably optimize both. The notable exception is \policy{Slice} at higher densities, which optimizes neither, explained by upper-bounds being effective when landmarks are ``between'' a source and target (in contrast, $k$NN prefers this policy as lower-bounds benefit from landmarks being ``behind'' a source and target as is the case with \policy{Slice}). 

\smallHead{Summary:} One important takeaway from the parameter testing is that landmark selection, including the strategy and the number of landmarks, warrants a deep investigation in its own right. Bucking the general trend, $k$FN search generally improves as the number of landmarks increases. A$k$NN search degrades with increasing $m$, until $m{=}8$ when it begins to improve with increasing density. This aligns with observations on bound performance in past work~\cite{potamias2009distest,goldberg2005alt}, where tightness needs a certain threshold on the number of global landmarks with diminishing returns thereafter. A closer inspection suggested that a large part of this improvement came from improving the root landmark bounds ($M^-_R$ and $M^+_R$). Thus, it may be worthwhile to explore combining landmark policies or setting different policies at different levels in future work.

\subsection{Performance on Other Spatial Networks}\label{sec:exp:other}

Although the conference version~\cite{abey2020aknn} was evaluated on road networks, the core concepts in our data structures are not restricted to a particular type of graph, and will work as long as a disjoint partitioning of vertices can be obtained. However, the key indicator of bound performance is whether the landmarks can be chosen such that they produce tight bounds for the vertices in their partitions. The intuition for this in road networks is that they are spatially embedded and separator-like borders~\cite{lipton1979planar} can be selected, with shortest paths tending to enter and leave partitions via a few similar routes. Such behavior may hold for other types of similar graphs, and we posit that our data structures can be applied broadly to spatial networks. To this end we evaluate COL-Trees on the W-US road network with distance edge-weights with weaker hierarchical properties compared to travel time weights and two spatial networks that are not road networks with relevant real-world applications, the NACA0015 mesh and the SRTMGL3 terrain graph (Table~\ref{tab:datasets}). For brevity we present one query type per dataset with a clear motivating example.

\subsubsection{Heuristic Performance with Distance Edge-Weights}\label{sec:exp:rndist}

Road networks with travel-time edge-weights are well known~\cite{geisberger2008ch} to display hierarchical properties with shortest paths tending to pass through increasingly important vertices, which can be exploited to accelerate computation. For example, the PHL index used to compute network distances in our experiments utilizes the vertex ranking implied by such properties to improve both construction cost and query time. In contrast, road networks with distance edge-weights do not display such properties as strongly, and PHL performs less efficiently in that case~\cite{farhan2023h2cl}. In Figure~\ref{exp:rndist}, we investigate whether our data structures also depend on travel-time edge-weights by comparing heuristic performance for A$k$NN $sum$ queries between distance versus travel-time edge-weights for the Western US (W-US) road network. Note that we choose the W-US so that PHL indexes for both types of weights can be memory resident (PHL for the US road network with distance weights requires more memory than our setup).

The figure shows that COL-Trees offer superior performance with a similar margin irrespective of the edge-weight, in terms of both network distances computed and candidates retrieved. Notably, IER performs better on distance edge-weights for network distances computed in Figure~\ref{exp:rndist:sum:heur:varyqndist:dist:fh}, which is expected because Euclidean distance becomes a tighter lower-bound when edge-weights represent physical distance. Note that the experiments for both edge-weights use the same partitioning structure, object sets, and query vertices, with heuristic search during query processing differing only in the bounds computed. As seen in Section~\ref{exp:preprocess}, PHL results in significantly larger labels for distance edge-weights, which translates to slower network distance queries. Therefore, we isolate the performance of the search heuristics by limiting comparison to machine-independent heuristic performance metrics because query time is also affected by the network distance computation. 

\begin{figure}[htb]
    \centering
    \begin{subfigure}[b]{0.24\linewidth}
        \includegraphics[width=\linewidth]{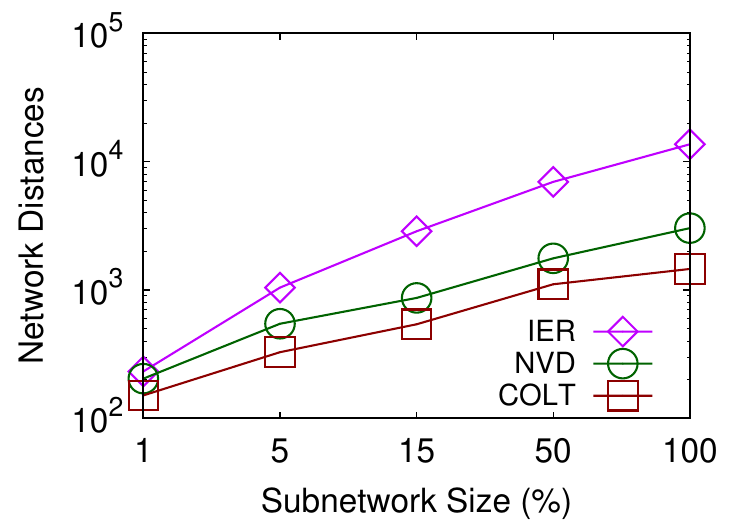}
        \caption{Network Distances (Dist)}
        \label{exp:rndist:sum:heur:varyqndist:dist:fh}
    \end{subfigure}
    \hfill
    \begin{subfigure}[b]{0.24\linewidth}
        \includegraphics[width=\linewidth]{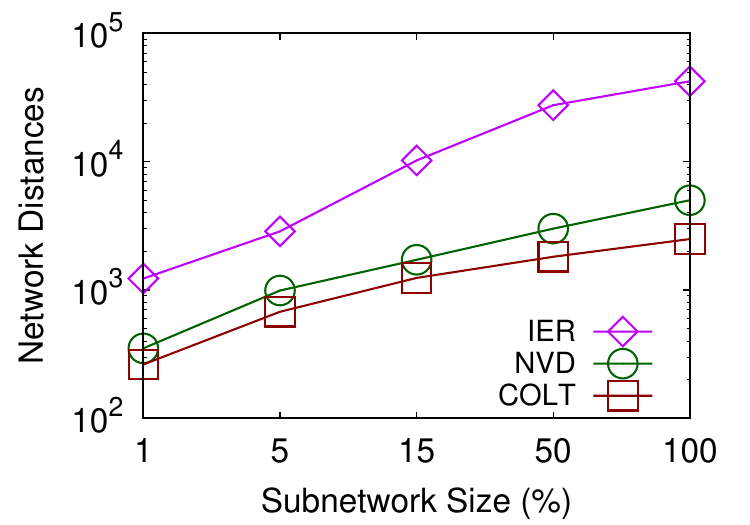}
        \caption{Network Distances (Time)}
        \label{exp:rndist:sum:heur:varyqndist:time:fh}
    \end{subfigure}
    \hfill
    \begin{subfigure}[b]{0.24\linewidth}
        \includegraphics[width=\linewidth]{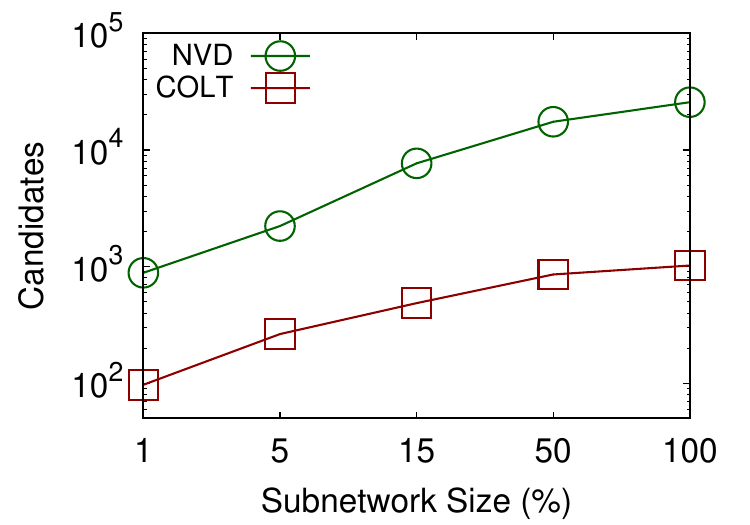}
        \caption{Candidates (Dist)}
        \label{exp:rndist:sum:heur:varyqndist:dist:lb}
    \end{subfigure}
    \hfill
    \begin{subfigure}[b]{0.24\linewidth}
        \includegraphics[width=\linewidth]{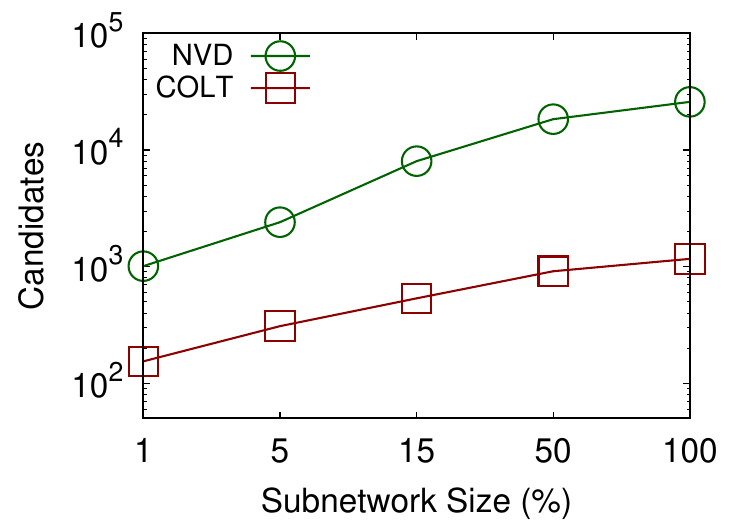}
        \caption{Candidates (Time)}
        \label{exp:rndist:sum:heur:varyqndist:time:lb}
    \end{subfigure}
	\caption{Heuristic performance of distance vs.~travel time edge-weights on Western US dataset (A$k$NN by $sum$)
	}
	\label{exp:rndist}
\end{figure}

\subsubsection{NACA0015 Airfoil Mesh Network}\label{sec:exp:naca0015}

Shortest paths in mesh networks are a common approximation for geodesic distance~\cite{tenenbaum2000geodesic}, giving an upper-bound on the true geodesic. Farthest neighbor queries find a natural application in such mesh network settings, e.g., Farthest Point Sampling~\cite{peyre2006geodesic} where the first operation is exactly a $k$FN query with $k{=}1$. COL-Trees can even be extended to handle subsequent steps by treating the sampled points as the query set and adjusting the bound computation. In this experiment, we investigate the performance of $k$FN queries on the NACA0015 airfoil 2D mesh network, which is effectively planar and lacks the hierarchy of road networks. Figure~\ref{exp:naca0015} highlights the applicability of our techniques beyond road networks, with significant improvement in query time over AUB-PHL (up to more than an order of magnitude) while computing only a very small number of network distances and retrieving similarly few candidates. Here, COL-Trees used $m{=}4$ random landmarks. Notably, existing geodesic distance methods such as MMP~\cite{mitchell1987mmp} rely on expansion-based search which, as discussed in Section~\ref{query:fkn}, is less desirable when considering the farthest metric, making further application of COL-Tree in this domain promising.

\begin{figure}[htb]
    \centering
    \begin{subfigure}[b]{0.24\linewidth}
        \includegraphics[width=\linewidth]{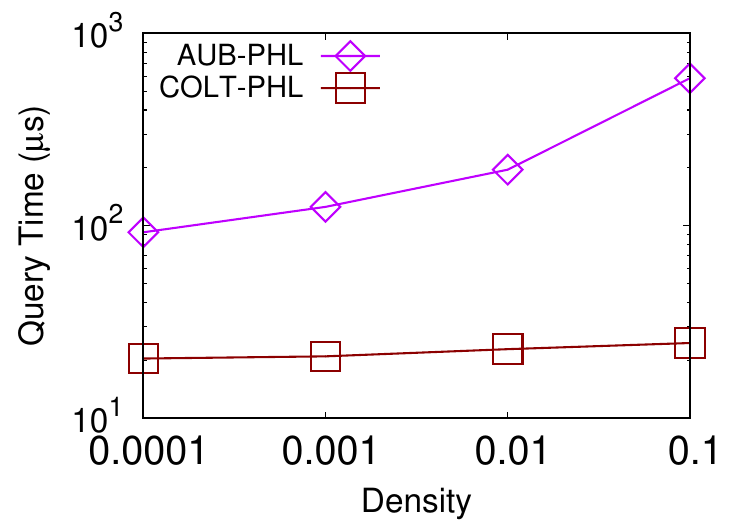}
        \caption{Varying $d$}
        \label{exp:naca0015:time:d}
    \end{subfigure}
    \hfill
    \begin{subfigure}[b]{0.24\linewidth}
        \includegraphics[width=\linewidth]{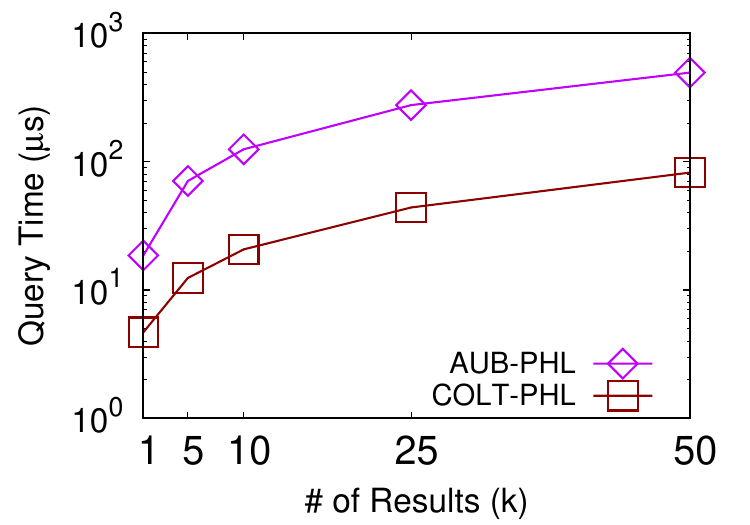}
        \caption{Varying $k$}
        \label{exp:naca0015:time:k}
    \end{subfigure}
    \hfill
    \begin{subfigure}[b]{0.24\linewidth}
        \includegraphics[width=\linewidth]{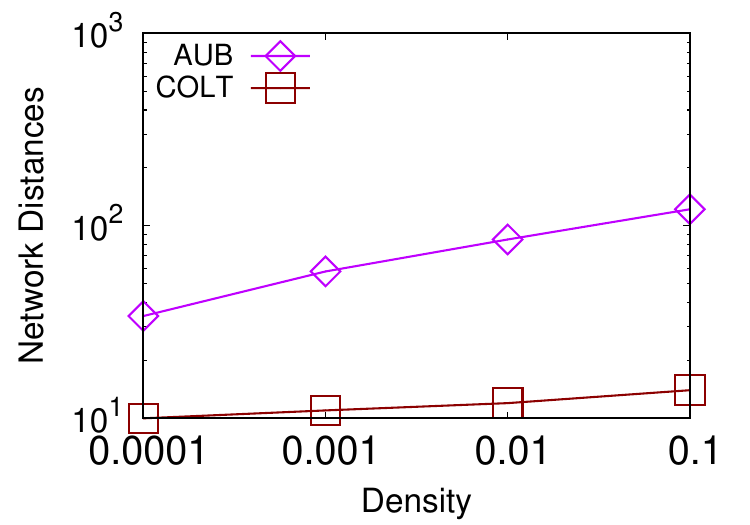}
        \caption{Network Distances}
        \label{exp:naca0015:heur:spdist}
    \end{subfigure}
    \hfill
    \begin{subfigure}[b]{0.24\linewidth}
        \includegraphics[width=\linewidth]{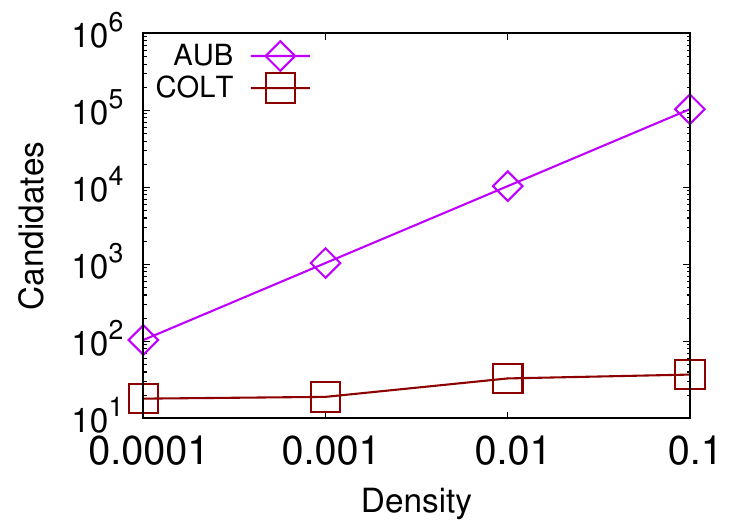}
        \caption{Candidates Retrieved}
        \label{exp:naca0015:heur:lbs}
    \end{subfigure}
	\caption{NACA0015 ($k$FN)
	}
	\label{exp:naca0015}
\end{figure}

\subsubsection{SRTMGL3 Terrain Network}\label{sec:exp:srtmgl3}

We also evaluate query performance for other query types on spatial networks that lack the hierarchical properties of road networks. Specifically, Figure~\ref{exp:srtmgl3} depicts the performance of A$k$NN queries on the SRTMGL3 terrain graph covering a mountainous area of Colorado, with $sum$ as the aggregate function and $d{=}0.01$. This experiment models, e.g., a group of hikers who want to rendezvous at one of a number of potential overnight camping spots and wish to minimize the total distance they need to travel. Figures~\ref{exp:srtmgl3:heur:spdist} and~\ref{exp:srtmgl3:heur:lbs} show a similar magnitude of improvement to that of A$k$NN queries on road networks in terms of heuristic performance on network distances computed and candidates retrieved. While there is also substantial improvement in query time (Figures~\ref{exp:srtmgl3:time:k} and~\ref{exp:srtmgl3:time:a}), the magnitude is smaller than on road networks. This is likely explained by the significantly larger size of the SRTMGL3 PHL index (23.0~GB in Table~\ref{tab:index:graph}) compared to the US index (15.8~GB) despite having 16\texttimes~fewer vertices. The larger label sizes result in costlier network distance computation. A technique such as H2CL~\cite{farhan2023h2cl}, which performs better for distance edge-weights, may be more suitable than PHL in such graphs. Nonetheless, as COL-Tree heuristic performance is largely unaffected by the lack of hierarchy, SRTMGL3 provides further evidence of the broader applicability of our approach. 

\begin{figure}[htb]
    \centering
    \begin{subfigure}[b]{0.24\linewidth}
        \includegraphics[width=\linewidth]{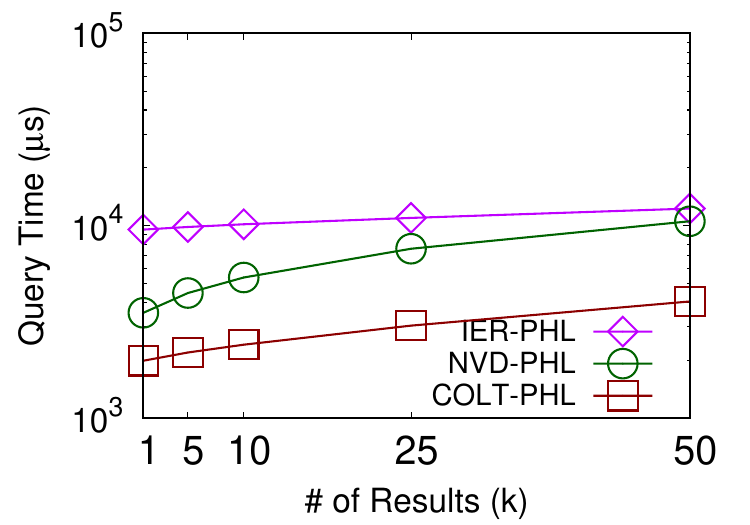}
        \caption{Varying $k$}
        \label{exp:srtmgl3:time:k}
    \end{subfigure}
    \hfill
    \begin{subfigure}[b]{0.24\linewidth}
        \includegraphics[width=\linewidth]{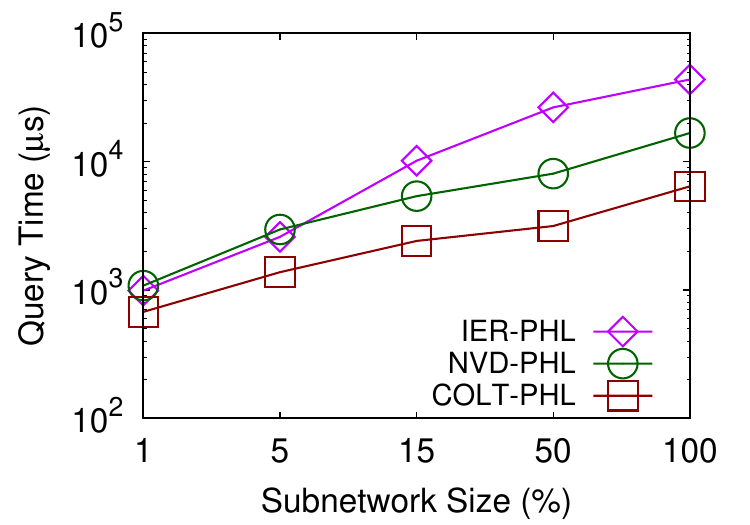}
        \caption{Varying $A$}
        \label{exp:srtmgl3:time:a}
    \end{subfigure}
    \hfill
    \begin{subfigure}[b]{0.24\linewidth}
        \includegraphics[width=\linewidth]{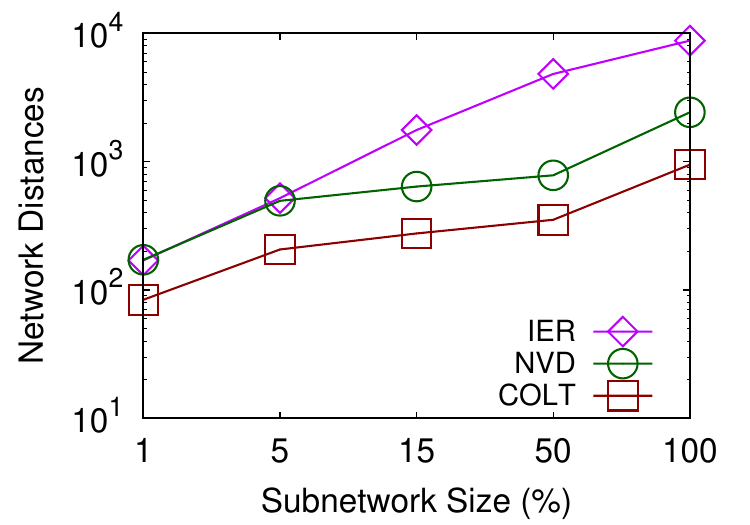}
        \caption{Network Distances}
        \label{exp:srtmgl3:heur:spdist}
    \end{subfigure}
    \hfill
    \begin{subfigure}[b]{0.24\linewidth}
        \includegraphics[width=\linewidth]{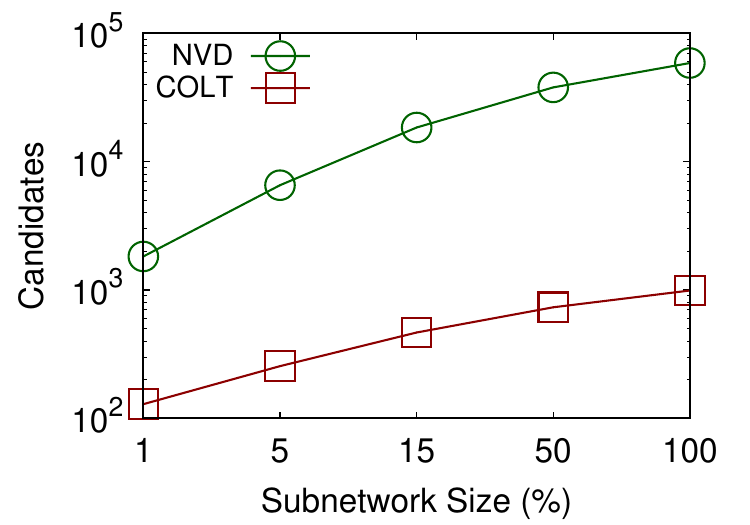}
        \caption{Candidates Retrieved}
        \label{exp:srtmgl3:heur:lbs}
    \end{subfigure}
	\caption{SRTMGL3 (A$k$NN by $sum$)
	}
	\label{exp:srtmgl3}
\end{figure}

\section{Related Work}\label{sec:relwork}
In this section, we review additional literature related to object search queries, shortest path indexing techniques, and other work broadly related to landmarks. 

\smallHeadIndent{Aggregate Nearest Neighbors:} In Section~\ref{intro}, we extensively discussed the shortcomings of the state-of-the-art A$k$NN techniques for road networks~\cite{yiu2006aknn,zhu2010vn3aknn}, which motivated COL-Tree. The IER-based technique suffers from the inaccuracy of Euclidean distance, especially when weights represent travel time along edges. This is a more serious problem in A$k$NN search as the error will also be aggregated, further limiting search space pruning. The NVD-based concurrent expansion suffers from high pre-processing costs and excessive expansions before candidates can be generated. $k$ Flexible ANNs ($k$FANNs)~\cite{yao2018fann,chen2021kwfann} generalize A$k$NN queries by choosing the best $\phi|Q|$ subset of the query vertices. While including $k$FANNs is beyond our scope, COL-Tree A$k$NN search can be extended for this problem as the aggregate lower-bound remains admissible. Moreover, the main techniques in~\cite{yao2018fann} are derived from the less desirable $k$NN heuristics and reduce to similar hierarchical Euclidean (IER) and expansion (NVD-like) baselines when tested on A$k$NN search ($\phi{=}1$). One recent work on $k$FANNs~\cite{chung2022kfann} (published after the conference version of our work) utilizes M-trees~\cite{ciaccia1997mtree}. The M-tree resembles a simplified version of COL-Trees, where reference nodes are chosen from the object set and network distances must be computed for each object set. In contrast, COL-Tree construction only requires look-ups in the shared SUL-Tree. By being both a $k$FANN method and allowing costlier pre-processing, the M-tree approach is not directly comparable. Furthermore, the M-tree approach does not utilize graph partitioning, tightening lower-bounds over multiple landmarks (i.e., reference nodes), or other features of COL-Trees such as subgraph vertex ordering that make COL-Tree highly suited to object search queries in spatial network graphs.

\smallHeadIndent{Farthest Neighbors and Maximization Queries:} In contrast to maximization problems in the context of heuristic search~\cite{stern2014max} where costs are replaced with rewards and total rewards must be maximized, the work relevant to COL-Trees is when the metric is still shortest distance but the goal is to find the \textit{farthest} results (objects) by this metric. This necessitates retrieval of candidate objects in decreasing order of upper-bounds in order to terminate optimally. Farthest Point Diagrams (FPD) in planar graphs~\cite{bose2013farthestptdiag} are the theoretical dual of an NVD. For each object, an FPD identifies the set of vertices for which it is the farthest object. However, FPD construction settles up to $|V|^2$ vertices for \textit{each object set}, which is significantly more expensive than even shortest path index construction. Unlike NVDs (where only $|V|$ vertices are settled in total), FPD construction cannot prune the Dijkstra search from each object when it encounters a vertex already settled by a search from another object. Consequently, FPDs are computationally infeasible to evaluate experimentally and as many spatial networks are not strictly planar (e.g., overpasses in road networks), FPDs cannot be used to generate candidates for $k$FN queries. 

Query variants incorporating the farthest metric have been studied in both Euclidean space and road networks. Euclidean approaches~\cite{yao2009rfn,wang2016rkfn,gao2011afn} are not applicable as they cannot provide an admissible heuristic for our setting (as discussed in Section~\ref{query:fkn}). Road network approaches often degenerate to evaluating all objects. \citet{tran2009rfkn} answer reverse farthest neighbor queries using NVD-based expansions from every object, evaluating nearly all objects when sets are dense~\cite{xu2017rfnroad}. Other techniques~\cite{cho2020gmp,cho2022mkfn} essentially evaluate all objects, as grouping objects along the same road segment does not alter the required network distance computations. Our AUB-PHL baseline (Section~\ref{sec:exp}) represents an optimized version of these approaches for $k$FN. The remaining road network approaches require expensive per-object-set pre-processing involving network distance computations. Hierarchical approaches do exist but require processing the road network for each object set, e.g., applying a hierarchical clustering of objects by network distance~\cite{wang2014afn} or require pre-computing the farthest objects for all partition boundary vertices~\cite{xu2017rfnroad}. Both also rely on a fixed set of global landmarks, which are known to make objects harder to differentiate with increasing density~\cite{abey2017knn}. In contrast, COL-Tree pre-processing is independent of the graph (depending only on $P$), as SUL-Trees are built once and shared across all object sets. Moreover, COL-Trees provide a generalized data structure for incremental candidate retrieval by both lower- and upper-bounds with more numerous ``local'' landmarks for tighter bounds at only a small increase in index size (Section~\ref{sec:complexity}). Ultimately, COL-Tree is a suitable candidate for application to all aforementioned query variants.

\smallHeadIndent{$k$ Nearest Neighbors:} 
Early $k$NN methods involved creating a single index incorporating the object set $P$ and the road network $G$ together, which is then used to both find objects and compute distances. For example, Network Voronoi Diagrams were first utilized in VN\textsuperscript{3} \cite{kolahdouzan2004vknn}, where the index includes distances between Voronoi regions to simultaneously compute distances to new candidates. Furthermore, many of these early techniques only supported fixed $k$~\cite{cho2005unicons}. Most recent techniques have proposed road network indexes that are utilized by multiple \textit{object indexes}, e.g., G-tree and its Occurrence Lists~\cite{zhong2015gtree}. Object indexes are typically much less costly to compute and, as previously observed~\cite{abey2016knn}, this \textit{decoupling} is highly advantageous and aligns more clearly with real-world usage. For example, a ride-hailing service frequently computes $k$NNs as drivers move and new passengers make requests. Thus, the object indexes change much more frequently than the road network, benefiting from fast rebuilding. Moreover, this decoupling has allowed more in-depth study into the candidate retrieval heuristics. For example, \citet{abey2016knn} found that a Euclidean distance heuristic could more efficiently utilize the G-tree index than its own search heuristic. It is worthwhile to underscore these benefits, as recent work~\cite{wang2024topknn} has reverted to non-decoupled indexes and fixed $k$. We refer the reader to an experimental review~\cite{abey2016knn} for a detailed overview of $k$NN techniques. 

\smallHeadIndent{Shortest Path Indexes:} Indexing techniques to more efficiently answer network distance queries on road networks have been extensively studied for several decades~\cite{wu2012shortest}. Advancement in this area is  orthogonal to our work, as object search utilizing COL-Tree is decoupled from the network distance computation. Thus, faster network distance querying will further improve the performance of COL-Tree. Typically, the lowest query times are provided by the most pre-processing intensive indexes, e.g., 2-hop labels. We chose PHL~\cite{kawata2014phl} because it provides one of the best trade-offs between network distance query time and pre-processing time for 2-hop label techniques. Nonetheless, COL-Tree can easily utilize the latest 2-hop labeling methods~\cite{farhan2023h2cl} with faster queries but more expensive pre-processing. The powerful flexibility offered by a decoupled index allows the real-world system designer to choose a method that matches the amount of pre-processing they can tolerate. For example, as an alternative, if the graph index must be rebuilt or updated frequently, the user can choose a lower-cost index such as G-tree~\cite{zhong2015gtree} or Contraction Hierarchies~\cite{geisberger2008ch}. 

\smallHeadIndent{Landmarks and Other Search Heuristics:} In road networks, landmarks were originally utilized as an improved heuristic for A* shortest path search~\cite{goldberg2005alt}. This extensively studied concept is also referred to as differential heuristics~\cite{sturtevant2009memheur} or beacons~\cite{fonseca2005beacons}. Many related works focus on reducing the storage and memory costs, e.g., computing landmark distances to only part of the graph and compensating for this by computing distances between the landmarks~\cite{kriegal2008hge}. Other work~\cite{goldberg2007reachlmk} has focused on improving accuracy, e.g., through better landmark selection by combining landmarks with \textit{reach}-based pruning. Many of these works are complementary to our techniques, e.g., integrating reach into COL-Tree could improve heuristic pruning and reduce false candidates or Compressed Differential Heuristics~\cite{goldenberg2011cdh} could improve lower-bound computation time without loss of accuracy. Other notable heuristics include Euclidean distance-based heuristics~\cite{rayner2011euclid}, compressed path databases~\cite{bono2019cpd}, portal-based true-distance~\cite{goldenberg2010portal}, and FastMap~\cite{cohen2018fastmap} heuristics. 

\smallHeadIndent{Landmark-Based Object Search in Road Networks:} Past attempts~\cite{kriegal2008hge} to leverage road network LLBs for object search based on the multi-step $k$NN paradigm~\cite{seidel1998knn} required a ranking that necessarily involved computing lower-bounds to all objects. As discussed, this approach is not scalable with object set size and does not compare favorably to methods that are capable of heuristic pruning. This is confirmed by the performance of the baseline $k$FN method, which takes a similar approach. Road Network Embedding (RNE) \cite{shahabi2002rne} involves transforming the road network into higher-dimensional space and using Minkowski metrics to estimate network distance. However, the proposed $k$NN method is approximate. \citet{qiao2013knk} propose to use shortest path trees to compute distance estimates based on tree distance, but their technique is also approximate and is applied to the alternative \emph{keyword} search problem in road networks. In contrast, COL-Tree provides a general data structure with an admissible heuristic that can be employed to answer a variety of spatial network queries with exact results.

\smallHeadIndent{Landmark-Based Search in Other Contexts:} \citet{mouratidis2016joint} proposed computing landmark lower-bounds to groups of users in a social network similar to Eq.~(\ref{eq:nodellb}). However, their technique is not designed for hierarchical graph traversals and, unlike COL-Tree, does not address the problem of computing more accurate lower-bounds (e.g., with more landmarks) necessary for an effective hierarchical search. M-trees~\cite{ciaccia1997mtree}, typically used for similarity search and multimedia databases, can be considered as a simplified version of COL-Tree. However, applying them to hierarchical graph traversal is not straightforward, with large spatial networks in particular presenting domain-specific challenges. For example, COL-Trees utilize a SUL-Tree to significantly reduce construction time through pre-processing the graph first and use multiple landmarks (i.e., reference nodes) to obtain better bounds. Landmark upper-bounds have been studied as a means for distance estimation in large-scale networks with the observation that they tend to be closer to the true network distances than lower-bounds~\cite{potamias2009distest}, but the utility in search was not considered.

\section{Conclusion \& Future Work}
We present COL-Tree, an efficient index to support hierarchical search in spatial networks using landmark lower-bounds. We show  that COL-Trees are highly suited to several queries where traditional search heuristics fall short: A$k$NN, $k$FN, and range queries. For example, A$k$NN query results are often not close to any single agent location and are more easily located using a hierarchical subgraph traversal supported by COL-Tree. By exploiting a useful property of COL-Tree leaf nodes for convexity-preserving aggregate functions, we can efficiently retrieve more promising A$k$NN candidates and terminate the search sooner. Moreover, the proposed data structure is lightweight in both theory and practice. Lastly, our experiments revealed significantly improved performance of up to three orders of magnitude, particularly for A$k$NN and $k$FN queries, over varying types of spatial networks including road networks, terrain graphs, and physics meshes.

\smallHeadIndent{Future Directions:} COL-Tree is a versatile data structure with potential application to other heuristic search problems. For example, we believe COL-Tree has the potential to be widely used in a similar way to R-trees in location-based services (LBS). While road network distance is more accurate than Euclidean distance, many LBS rely on R-trees for lower pre-processing cost despite poor accuracy. As COL-Tree has significantly reduced the gap, it is now more practicable to consider revisiting graph-based approaches for challenging problems such as Reverse $k$NNs~\cite{yang2015rknn}. In fact, COL-Tree has recently been applied to improve real-world driver-passenger matching in ride-hailing~\cite{abey2021ikm}. There are also many possible directions to improve heuristic performance, such as landmark selection or applying Compressed Differential Heuristics~\cite{goldenberg2011cdh}.

\section*{Acknowledgments}
We thank the anonymous reviewers of both the conference and journal submissions for their constructive feedback. We are especially grateful to the reviewer who evaluated both versions, whose thoughtful suggestions, including those related to distance edge-weights, helped strengthen this extended version. Tenindra Abeywickrama is supported by JSPS KAKENHI Grant Number 25K24399. Muhammad Aamir Cheema is supported by Australian Research Council DP230100081. 

\section*{Declaration of generative AI and AI-assisted technologies in the manuscript preparation process}

During the preparation of this work, the authors used Claude to assist with grammatical checks, minor presentation improvements, and AI-assistance in the development of the software code artifacts. The authors reviewed and edited the output and take full responsibility for the content of the published article.

\bibliographystyle{elsarticle-num-names}

{\small\bibliography{references}}

\appendix

\section{\texorpdfstring{$k$}{k}NN Querying with COL-Tree}\label{app:knn}

Recall that one of the motivations for COL-Tree was the fact that most existing heuristics were designed with $k$NN queries in mind and did not translate well to other problems such as A$k$NN and $k$FN search. However, in Section~\ref{sec:exp:range}, we found that COL-Tree could still perform on par with or better than state-of-the-art $k$NN heuristics for range queries due to certain advantages offered by COL-Tree. This was particularly beneficial given the significantly lower pre-processing cost of COL-Trees. In this appendix we investigate the performance of COL-Tree for $k$NN queries to illustrate scenarios where it could be a better option. We omit the pseudocode for the $k$NN query algorithm, as it can be easily derived by simplifying Algorithm~\ref{alg:hpolaknn} based on the assumption that the query vertex set $Q$ contains only one vertex.

\subsection{Experimental Results}\label{sec:exp:knn}
As expected, the state-of-the-art NVD heuristic offers the best performance on real-world POI datasets in Figure~\ref{exp:knn:rw}, with COLT-PHL coming second in all cases. The performance gap between NVD-PHL and COLT-PHL is negligible on small POI sets because the heuristics need only differentiate between a small pool of candidates to determine exact results. Moreover, even on the largest datasets, the difference is not significant. For example, the worst performance for COLT-PHL is for schools where it is 2\texttimes~slower, but the difference is only 14$\mu$s. Based on Table~\ref{tab:index}, a system using NVDs would take $5.2$s to construct an index and process a workload of $10{,}000$ $k$NN queries. In that time, a COL-Tree could be constructed and serve more than $175{,}000$ queries, making it preferable when amortization opportunities are limited. COL-Tree parameters for $k$NN querying were $m{=}4$ \policy{Slice} landmarks, $b{=}8$, and $\lambda{=}256$.

{
	\begin{figure}[htb]
		\centering
		\includegraphics[width=0.65\linewidth]{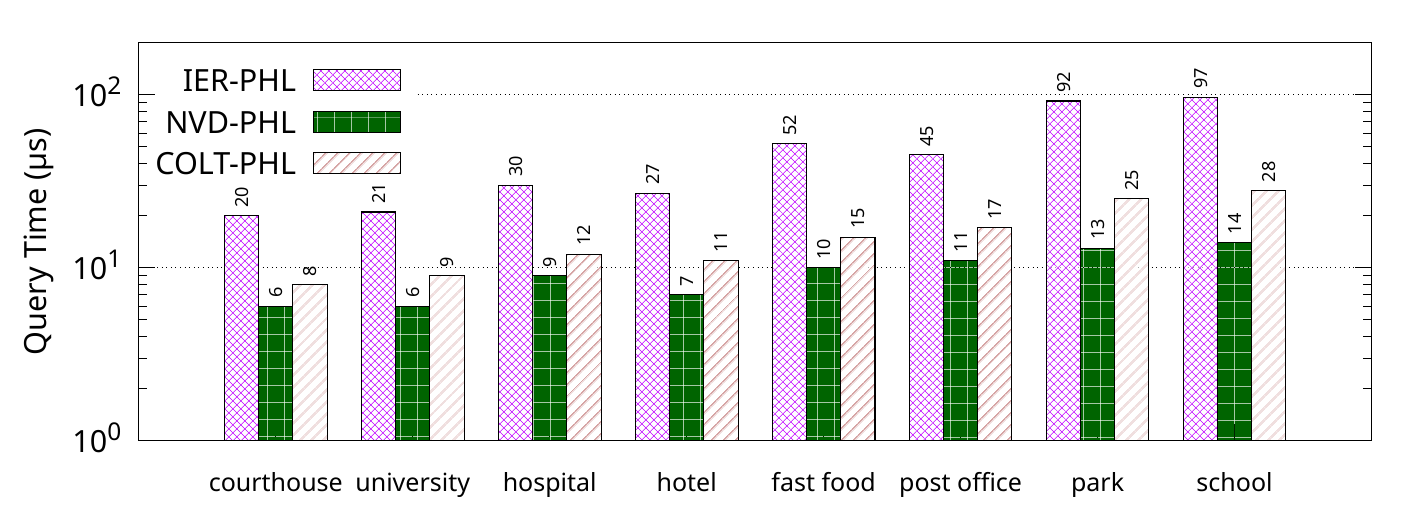}
		\caption{$k$NN query performance on different real-world POI sets  
        }
        \label{exp:knn:rw}
	\end{figure}
}

\subsubsection{Sensitivity Analysis}

\smallHead{Effect of $\boldsymbol{d}$:} 
With increasing density $d$, the total number of objects increases, and as a result hierarchical methods like IER-PHL and COLT-PHL scale more poorly than NVD-PHL in Figure~\ref{exp:knn:varyd}. In this scenario, hierarchical methods are disadvantaged because they must traverse increasingly larger tree paths. Conversely, NVDs need only consider the neighboring Voronoi regions of the NNs found so far, the number of which does not increase significantly (and in fact converges to the average vertex degree with increasing density). 

\smallHead{Effect of $\boldsymbol{k}$:} 
COLT-PHL is only marginally worse than NVD-PHL for varying $k$ in Figure~\ref{exp:knn:varyk} as it is much less disadvantaged on the default density of $0.001$. We note that $d{=}0.001$ corresponds to a more realistic object set size judging by the densities of real-world POI sets in Table~\ref{tab:pois}. The only exception is $k{=}1$, where NVD-PHL is optimal (the exact $1$NN can be retrieved immediately by NVD definition).

\smallHead{Heuristic Efficiency:}
Given that both methods rely on landmark lower-bounds, COLT-PHL does not need to compute significantly more network distances at lower densities than NVD-PHL as shown in Figure~\ref{exp:knn:heur:varyd:fh}. While at higher densities lower-bounds become less accurate, NVD-PHL is insulated against this because of the previously mentioned convergence of the number of neighboring Voronoi regions to the average vertex degree with increasing $d$. The main reason for COLT-PHL's poorer performance is illustrated in Figure~\ref{exp:knn:heur:varyd:lb}, which shows a significantly higher number of candidate objects are retrieved. This suggests that hierarchical traversals are more susceptible to false candidates in proximity search.

\begin{figure}[htb]
    \centering
    \begin{subfigure}[b]{0.24\linewidth}
        \includegraphics[width=\linewidth]{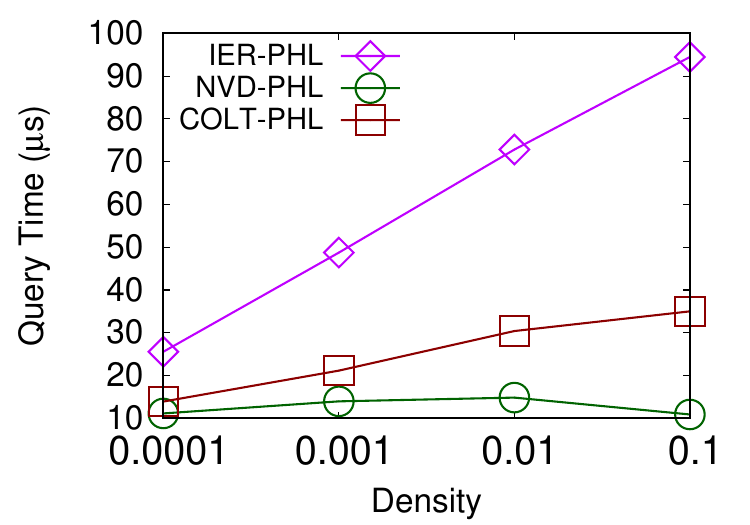}
        \caption{Varying $d$}
        \label{exp:knn:varyd}
    \end{subfigure}
    \hfill
    \begin{subfigure}[b]{0.24\linewidth}
        \includegraphics[width=\linewidth]{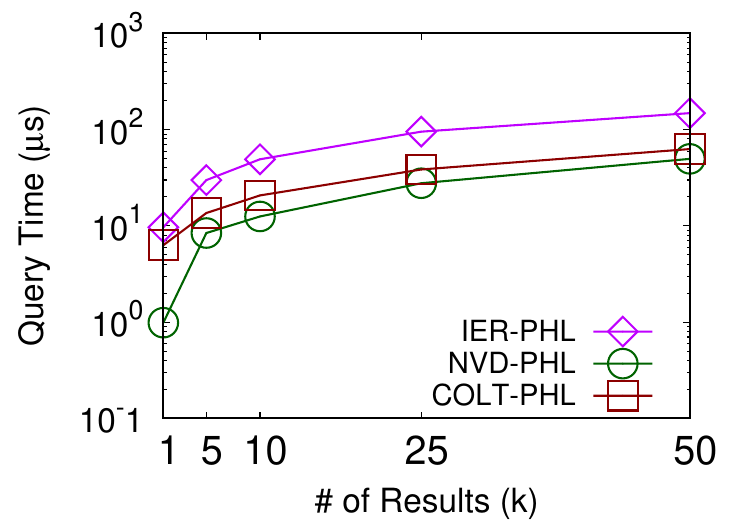}
        \caption{Varying $k$}
        \label{exp:knn:varyk}
    \end{subfigure}
    \hfill
    \begin{subfigure}[b]{0.24\linewidth}
        \includegraphics[width=\linewidth]{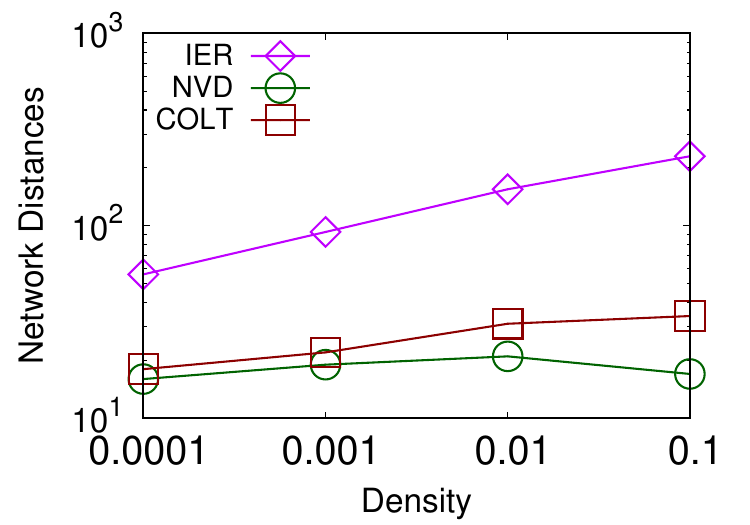}
        \caption{Network Distances}
        \label{exp:knn:heur:varyd:fh}
    \end{subfigure}
    \hfill
    \begin{subfigure}[b]{0.24\linewidth}
        \includegraphics[width=\linewidth]{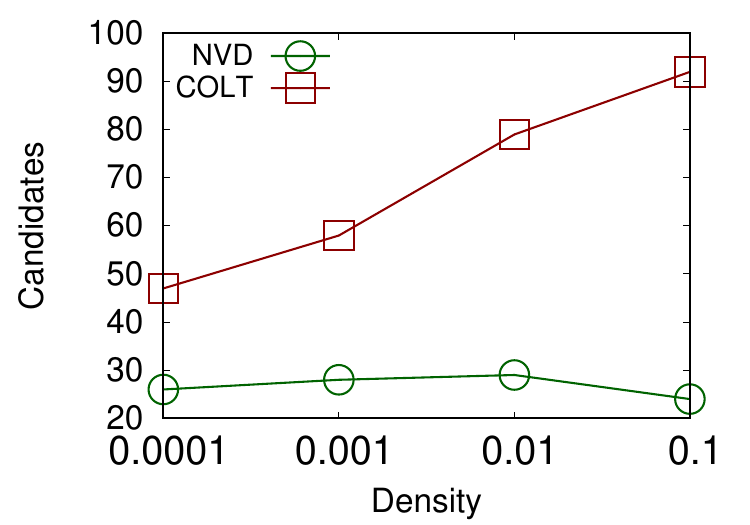}
        \caption{Candidates Retrieved}
        \label{exp:knn:heur:varyd:lb}
    \end{subfigure}
	\caption{Performance for $k$NN queries 
	}
	\label{exp:knn}
\end{figure}

\end{document}